\documentclass{mystyle}

\newtheorem{exercise}[theorem]{Exercise}

\def\emptyset{\varnothing}

\def\cA{{\cal A}}
\def\cB{{\cal B}}

\def\cD{{\cal D}}

\def\cF{{\cal F}}

\def\cH{{\cal H}}

\def\cK{{\cal K}}

\def\sA{{\mathsf A}}

\def\sH{{\mathsf H}}

\def\bC{{\mathbb C}}           

\def\bH{{\mathbb H}}

\def\bN{{\mathbb N}}

\def\bR{{\mathbb R}}

\def\bS{{\mathbb S}}

\def\bT{{\mathbb T}}
\def\bZ{{\mathbb Z}}

\def\gA{{\mathfrak A}}       
\def\gB{{\mathfrak B}}

\def\gg{{\mathfrak g}}

\def\gL{{\mathfrak L}}
\def\gM{{\mathfrak M}}
\def\gN{{\mathfrak N}}

\def\gQ{{\mathfrak Q}}
\def\gR{{\mathfrak R}}
\def\gS{{\mathfrak S}}

\def\beq{\begin{eqnarray}}
\def\eeq{\end{eqnarray}}

\begin{document}


%
%

\title{Mathematical Foundations of  Quantum Mechanics: An Advanced Short Course}

\author{Valter Moretti}

\address{Department of Mathematics of the University of Trento and INFN-TIFPA,\\
via Sommarive 14, I-38122 Povo (Trento), Italy\,\\
\email{valter.moretti@unitn.it} }

\maketitle
%

\begin{abstract}
{\bf Abstract}. This paper  collects and extends the lectures I gave at the ``XXIV International Fall Workshop on Geometry and Physics''  held in Zaragoza (Spain) during September 2015.
Within these lectures I  review the formulation of Quantum Mechanics, and quantum theories in general, from a mathematically advanced viewpoint, essentially based on the orthomodular lattice of elementary propositions, discussing some fundamental ideas, mathematical tools and theorems also related to the representation of physical symmetries. The final step consists of an elementary introduction the so-called (C*-) algebraic formulation of quantum theories. 
\end{abstract}


\tableofcontents

\section{Introduction: Summary of elementary facts of QM}	 
This paper  collects and technically extends the lectures given by the author at the ``XXIV International Fall Workshop on Geometry and Physics''  held in Zaragoza, August 31 - September 4, 2015. These lecture notes contain much more written material than the lectures themselves.

A concise account of the basic structure of quantum mechanics and {\em quantization procedures} has already been presented  in \cite{ercolessi} with several crucial examples. In the rest of Section 1, we quickly review again some elementary facts and properties, either of physical or mathematical nature,   related to Quantum Mechanics, without fully entering into the mathematical details. 

Section 2 is instead devoted to present some technical definitions and results of spectral analysis in complex Hilbert spaces, especially the basic elements of spectral theory, including the classic theorem about spectral decomposition of (generally unbounded) selfadjoint operators and the so called measurable functional calculus. A brief presentation of the three most important operator topologies  for applications in QM closes Section 2.
 
Within Section 3, the {\em corpus} of the lectures, we pass to analyse the mathematical structure of QM from a finer and advanced viewpoint, adopting the  
framework based on orthomodular lattices' theory. This approach permits one to justify some basic assumptions of QM,
like the mathematical nature of the observables represented by selfadjoint operators and the quantum states viewed as trace class operators. QM is essentially a probability measure on the  non-Boolean lattice ${\cal L}(\cH)$ of elementary observables.  A key tool of that analysis is the theorem by Gleason characterising   the notion probability measure  on ${\cal L}(\cH)$ in terms of certain trace class operators. We also discuss the structure of the algebra of observables in the presence of superselection rules after having introduced the mathematical notion  of von Neumann algebra. The subsequent part of the third section is devoted to present  the idea of quantum symmetry, illustrated in terms of Wigner and Kadison theorems.
Some basic mathematical facts  about groups of quantum symmetries are introduced and discussed, especially in relation with the problem of their unitarisation. Bargmann's condition is stated.  The particular case of a strongly continuous one-parameter unitary group will be analysed  in some more detail,  mentioning von Neumann's theorem and the celebrated Stone theorem, remarking its use  to describe the time evolution of quantum systems. A quantum formulation of Noether theorem ends this part. 
The last part of Section 3 aims to introduce some elementary results about continuous unitary representations of Lie groups, discussing in particular  a theorem by Nelson  which proposes sufficient conditions for lifting a (anti)selfadjoint  representation of a Lie algebra to a unitary representation of the unique simply connected Lie group associated to that Lie algebra.

 The last section closes the paper focussing on elementary ideas and results of the so called algebraic formulation of quantum theories. Many examples and exercises (with solutions) accompany the theoretical text at every step.

\subsection{Physical facts about Quantum Mechanics}\label{subsec1}
Let us quickly review the most relevant and common features of quantum systems. Next we will present a first elementary mathematical formulation which will be improved in the rest of the lectures, introducing a suitable mathematical technology.

\subsubsection{When a physical system is quantum}
Loosely speaking, Quantum Mechanics is the physics of microscopic world (elementary particles, atoms, molecules). That realm is characterized by a universal physical constant denoted by $h$ and   called  {\bf Planck constant}. A related constant -- nowadays much  more used -- is the {\bf reduced Planck constant}, pronounced ``h-bar'',
$$\hbar  :=  \frac{h}{2\pi}= 1.054571726 \times 10^{-34}  J\cdot s\:.$$
The physical dimensions of $h$ (or $\hbar)$ are those of an {\em action}, i.e. {\em energy} $\times$ {\em time}.
A rough check on the appropriateness  of a quantum physical description  for a given physical system is obtained by comparying the value of some characteristic action of the system with $\hbar$. For a macroscopic pendulum (say, length $\sim 1m$, mass $\sim 1 kg$ maximal speed $\sim 1 ms^{-1}$), multiplying the period of oscillations and the maximal kinetic energy,   we obtain a typical action of $\sim 2 Js >\!\!> h$. In this case quantum physics is expected to be  largely inappropriate, exactly as we actually know from our  experience of every days. Conversely, referring to a hydrogen  electron orbiting around its proton, the first {\em ionization energy} multiplied with the orbital period of the electron (computed using the classical formula with a value of the radius of the order of  $1$ \r{A}) produces  a typical action of the order of $h$.  Here quantum mechanics is necessary.

\subsubsection{General properties of quantum systems} \label{GP1}
Quantum Mechanics (QM) enjoys a triple of features which seem to be very far from properties of Classical Mechanics (CM).  These remarkable general properties concern  the physical quantities of physical systems. In QM physical quantities are called {\em observables}.

 {\bf (1) Randomness}.   When we perform a measurement of an observable of a quantum system, 
the outcomes turn out to be {\em stochastic}: Performing measurements of the same observable $A$ on completely identical systems prepared in the {\em same} physical state, one  generally  obtains different results $a,a', a'' \ldots$. \\ Referring to  the standard interpretation of the formalism of QM (see \cite{stanford} for a nice up-to-date account on the various interpretations),  the   randomness of measurement outcomes should not be considered as  due to an incomplete knowledge of the state of the system as it happens, for instance, in Classical Statistical  Mechanics. Randomness is not {\em epistemic}, but it is {\em ontological}. It is a fundamental property of quantum systems.\\
On the other hand,  {\em QM permits one to compute the {\em probability distribution} of all the outcomes of a given observable, once the state of the system is known}. \\ 
Moreover, it is always  possible to prepare a state $\psi_a$  where a certain observable $A$ is {\em defined} and takes its value $a$. That is, repeated measurements of $A$ give rise to the same value $a$ with probability $1$. 
(Notice that we can perform simultaneous measurements on identical systems all prepared in the state $\psi_a$, or we can perform different subsequent measurements on the same system in the state $\psi_a$. In the second case, these measurements have to be performed very close to each other in time to prevent the state of the system from  evolving in view of Schr\"odinger evolution as said in (3) below.) Such states, where observable take definite values, cannot be prepared for {\em all} observables simultaneously as discussed in (2) below.

{\bf (2) Compatible and  Incompatible Observables}. The second  noticeable feature of  QM is the existence of {\em incompatible observables}. Differently from CM, there are physical quantities which cannot be measured simultaneously. There is no physical instrument capable to do it.   If an  observable $A$ is {\em defined}  in a given state $\psi$ -- i.e. it attains a precise value $a$ with probability $1$ in case of a measurement -- an  observable $B$ {\em incompatible} with $A$ turns out to be {\em not defined} in the state $\psi$ --  i.e.,  it generally  attains several different values $b,b',b''\ldots$, {\em none with probability $1$},  in case of measurement. So, if we perform a measurement of $B$,  we generally obtain  a spectrum of values described by  a probabilistic distribution as preannounced in (1) above. \\ Incompatibility  is a {\em symmetric} property: $A$ is incompatible with $B$ if and only if  $B$ is incompatible with $A$. However it is not {\em transitive}.\\
There are also  {\em compatible observables}  which, by definition, can be measured simultaneously. An example is the component $x$ of the position of a particle and the component $y$ of the momentum of that particle, referring to a given inertial reference frame.   
A popular case of  incompatible observables is a pair of {\em canonically conjugated observables} \cite{ercolessi} like the position $X$ and the momentum $P$ of a particle both  along the same  fixed axis of a reference frame. In this case there is a lower  bound for the product of the standard deviations, resp. $\Delta X_\psi$, $\Delta P_\psi$, of the outcomes of the  measurements  of these observables in a given state $\psi$ (these measurement has to be performed on different identical systems all prepared in the same state $\psi$). This lower bound does not depend on the state and is encoded in the celebrated mathematical formula of the {\em Heisenberg principle} (a theorem in the modern formulations):
\begin{equation}\Delta X _\psi\Delta P_\psi \geq \hbar/2\:,\label{Hp}\end{equation}
where Planck constant shows up.

{\bf (3) Post measurement Collapse of the State}.  In QM, measurements {\em generally change the state of the system} and produce a post-measurement state from the state on which the measurement is performed. (We are here referring to idealized measurement procedures, since   measurement procedures are  very often destructive.) If the measured state is $\psi$, immediately after the measurement of an observable $A$ obtaining the value $a$ among a plethora of possible values $a,a',a'',\ldots$, the state changes to $\psi'$ generally different form $\psi$. In the new state $\psi'$,  the distribution of probabilities  of the outcomes of $A$ changes to $1$ for the outcome $a$ and $0$ for all other possible outcomes.  $A$ is therefore {\em defined} in $\psi'$. \\ When 
we perform repeated and alternated measurements of a pair of incompatible observables, $A$, $B$,  the outcomes  disturb each other: If the first outcome of $A$ is $a$, after a measurement of $B$, a subsequent measurement of $A$ produces $a'\neq a$ in general. 
Conversely, if $A$ and $B$ are compatible, the outcomes of their subsequent measurements do not disturb each other. \\
In CM there are  measurements that, in practice, disturb and change the state of the system. It is however  possible to  decrease the disturbance arbitrarily, and nullify it in ideal measurements.
 In QM it is not always possible as for instance  witnessed by (\ref{Hp}). \\

\noindent In QM, there are two types of time evolution of the state of a system. One is the usual one due to the dynamics and encoded in the famous {\em Schr\"odinger equation}  we shall see shortly. It is nothing but a quantum version of classical {\em Hamiltonian evolution} \cite{ercolessi}.
The other  is the sudden change of the state  due to measurement procedure of an observable, outlined  in (3): The {\em collapse of the state} (or {\em wavefunction}) of the system. \\ The nature of the  second type of evolution  is still source of an animated debate in the scientific community of physicists and philosophers of Science. 
There are many attempts  to reduce the collapse of the state to the standard time evolution  referring to the quantum evolution of the whole physical  system, also including  the measurement apparatus and the environment ({\em de-coherence processes}) \cite{stanford,decoherence}. None of these approaches seem to be completely satisfactory up to now. 

\remark{\em Unless explicitly stated,  we henceforth adopt a physical unit system such that $\hbar =1$.}

\subsection{Elementary formalism for the finite dimensional case}\label{subsecFD}  To go on with this introduction,  let us add some further technical details to the presented picture to show how practically (1)-(3) have to be mathematically  interpreted (reversing the order of (2) and (3) for our convenience). The rest of the paper is devoted to  make technically precise, justify and widely develop these ideas from a mathematically more advanced viewpoint than the one of \cite{ercolessi}. \\
To mathematically simplify this introductory discussion, throughout this section, except for Sect \ref{subsecID}, we assume  that $\cal H$ denotes a {\em finite dimensional} complex vector space equipped with a Hermitian scalar product, denoted by $\langle \cdot| \cdot \rangle$, where the linear entry is the second one.
 With $\cal H$ as  above,   $L({\cal H})$ will denote the complex algebra of operators ${A}: \cal H \to \cal H$.
We remind the reader that, if $A \in L({\cal H})$ with $\cal H$ finite dimensional, the {\em adjoint operator},  $A^* \in L({\cal H})$, is the unique linear operator such that  \begin{equation}\langle A^*x|y \rangle = \langle x| Ay\rangle \quad\mbox{for all $x,y \in \cal H$.} \label{agg}\end{equation}
$A$ is said to be {\em selfadjoint} if  $A=A^*$, so that, in particular
\begin{equation}\langle Ax|y \rangle = \langle x| Ay\rangle \quad\mbox{for all $x,y \in \cal H$.} \label{agg2}\end{equation}
Since $\langle \cdot | \cdot\rangle$ is linear in the second entry and antilinear in the first entry, we immediately  have that all eigenvalues of a selfadjoint operator $A$
are real.\\

\noindent Our assumptions on the mathematical description of quantum systems are the following ones.

\begin{enumerate}
\item   A quantum mechanical system $S$ is always associated to a complex vector space ${\cal H}$ (here finite dimensional) equipped with a Hermitian scalar product  $\langle \cdot | \cdot \rangle$;

\item  observables are  pictured in terms of {\em selfadjoint} operators $A$ on $\cal H $;

\item states are   equivalence classes of {\em unit} vectors $\psi \in {\cal H}$, where $\psi \sim \psi'$ iff $\psi = e^{ia} \psi'$ for some $a\in \mathbb R$.

\end{enumerate}

\remark $\null$

{\bf (a)} It is clear that states are therefore one-to-one  represented by all of the elements of the  complex projective space $P\cal H$. The states we are considering within this introductory  section are called {\em pure} states. A more general notion of state, already introduced in \cite{ercolessi}, will be discussed later.

{\bf (b)} $\cal H$ is an elementary version of complex Hilbert space since it is automatically complete it being finite dimensional. 

{\bf (c)} Since $\dim(\cal H)< +\infty$,  every self-adjoint operator $A\in L({\cal H})$ admits a spectral decomposition 
\begin{equation}
A = \sum_{a \in \sigma(A)}aP^{(A)}_a \label{sectradec0}\:,
\end{equation}
where $\sigma(A)$ is the {\em finite}  set of eigenvalues -- which must be {\em real} as $A$ is self-adjoint -- and $P^{(A)}_a$ is the 
orthogonal projector onto the eigenspace associated to $a$. Notice that $P_aP_{a'}=0$ if $a \neq a'$ as eigenvectors with different eigenvalue are orthogonal. 
\hfill $\blacksquare$\\

\noindent Let us show how the mathematical assumptions (1)-(3)  permit us to set the physical properties 
of quantum systems  (1)-(3) into a mathematically nice form.\\

{\bf (1) Randomness}:  
The eigenvalues of an observable  $A$ are physically interpreted as the possible values of the outcomes of a measurement of $A$.\\
Given a state, represented by the unit vector  $\psi \in \cal H$, the probability to obtain $a \in \sigma(A)$ as an outcome when measuring $A$  is
$$\mu^{(A)}_\psi(a) := || P^{(A)}_a \psi||^2\:.$$
Going along with this interpretation, the expectation value of $A$ ,when the state is represented by $\psi$, turns out to be
$$\langle A \rangle_\psi := \sum_{a\in \sigma(A)} a \mu^{(A)}_\psi(a)  = \langle \psi | A \psi \rangle\:.$$
So that the identity holds
\begin{equation}
\langle A \rangle_\psi = \langle \psi | A \psi \rangle \label{E}\:.
\end{equation}
Finally, the standard deviation $\Delta A_\psi$ results to be
\begin{equation}\Delta A_\psi^2 := \sum_{a\in \sigma(A)} (a-\langle A \rangle_\psi)^2  \mu^{(A)}_\psi(a)  = \langle \psi | A^2  \psi \rangle - \langle \psi |A \psi \rangle^2 \:.\label{Delta} \end{equation}

\remark $\null$

{\bf (a)}  Notice that the arbitrary phase affecting the unit vector $\psi \in \cal H$ ($e^{ia}\psi$ and $\psi$ represent the same quantum state for every $a\in \mathbb R$) is armless here.

{\bf (b)}  If $A$ is an observable and  $f: \mathbb R \to \mathbb R$ is given,  $f(A)$ is interpreted as an observable whose values are $f(a)$ 
if $a \in \sigma(a)$: Taking  (\ref{sectradec0}) into account,
\begin{equation}
f(A) := \sum_{a \in \sigma(A)} f(a) P^{(A)}_a \label{sectradec1}\:.
\end{equation}
For polynomials  $f(x) = \sum_{k=0}^n a_kx^k$, it results  $f(A)=  \sum_{k=0}^n a_kA^k$  as expected.
The  selfadjoint operator  $A^2$ can naturally be interpreted this way as the natural observable whose values are $a^2$ when $a \in \sigma(A)$. For this reason, looking at the last term
in (\ref{Delta}) and taking (\ref{E}) into account,
\begin{equation}\Delta A_\psi^2  =\langle A^2 \rangle_\psi - \langle A \rangle_\psi^2 \:.\label{Delta2}\end{equation} \hfill $\blacksquare$

 {\bf (3) Collapse of the state}: If $a$ is the outcome of the (idealized) measurement of $A$ when the state is represented by $\psi$, the new state immediately after the measurement is represented by the unit vector
\begin{equation}\psi' := \frac{P_a^{(A)}\psi}{||P_a^{(A)}\psi ||}\label{LvN}\:.\end{equation}

\remark  Obviously this formula does not make sense if $\mu^{(A)}_\psi(a)=0$ as expected. Moreover 
the arbitrary phase affecting $\psi$ does not lead  to troubles, due to the linearity of $P^{(A)}_a$ .\\

 {\bf (2) Compatible and  Incompatible Observables}: Two observables are compatible -- i.e. they can be simultaneously measured -- if and only if the associated operators  {\em commute}, that is
$$AB-BA =0\:.$$
Using the fact that $\cal H$ has finite dimension, one easily proves that the observables $A$ and $B$ are compatible if and only if the associated spectral projectors commute as well 
$$P^{(A)}_a P^{(B)}_b = P^{(B)}_b P^{(A)}_a \quad a \in \sigma(A)\:, b \in \sigma(B)\:.$$
In this case 
$$||P^{(A)}_a P_b^{(B)} \psi||^2 = || P_b^{(B)} P^{(A)}_a \psi||^2$$
has the natural interpretation of the probability to obtain the outcomes $a$ and $b$ for a simultaneous measurement of $A$ and $B$.  If instead $A$ and $B$ are incompatible,  it may happen that
$$||P^{(A)}_a P_b^{(B)} \psi||^2 \neq  || P_b^{(B)} P^{(A)}_a \psi||^2\:.$$
Sticking  to the case of  $A$ and $B$ incompatible, exploiting  (\ref{LvN}),  
\begin{equation}||P^{(A)}_a P_b^{(B)} \psi||^2 = \left|\left|P^{(A)}_a  \frac{P_b^{(B)} \psi}{|| P_b^{(B)} \psi ||}\right|\right|^2  || P_b^{(B)} \psi ||^2\label{measureincomp}\end{equation} 
has the natural meaning of {\em the probability of obtaining first $b$ and next  $a$ in a subsequent measurement of $B$ and $A$}. 

\remark $\null$

{\bf (a)} Notice that, in general,  we  cannot interchange the r\^ole of $A$ and $B$ in (\ref{measureincomp}) because, in general,  $P^{(A)}_a P^{(B)}_b \neq  P^{(B)}_b P^{(A)}_a$ if $A$ and $B$ are incompatible. The measurement procedures ``disturb each other'' as already said. 

{\bf (b)} The interpretation of (\ref{measureincomp}) as probability of subsequent measurements can be given also if $A$ and $B$ are compatible. In this case,  the probability of obtaining first $b$ and next  $a$ in a subsequent measurement of $B$ and $A$ is identical to the probability of measuring $a$ and $b$ simultaneously and, in turn, it coincides with  the probability of obtaining first $a$ and next  $b$ in a subsequent measurement of $A$ and $B$ 

{\bf (c)} $A$ is always compatible with itself. Moreover $P_a^{(A)}P_{a}^{(A)} = P_a^{(A)}$ just due to the definition of projector.  This fact has the immediate consequence that if we obtain $a$ measuring $A$ so that the state immediately after the measurement is represented by $\psi_a = ||P_a^{(A)}\psi||^{-1}\psi$, it will remain
$\psi_a$ even after other subsequent measurements of $A$ and the outcome will result to be always $a$. Versions of this  phenomenon, especially regarding the decay of unstable particles, are   experimentally confirmed and it is called the {\em quantum Zeno effect}. \hfill $\blacksquare$\\

\example  An electron admits  a triple of observables, $S_x$, $S_y$, $S_z$, known as the  components of the {\em spin}. Very roughly speaking, the spin can be viewed as  the  angular momentum of the particle referred to  a reference frame always at rest with  the centre of the  particle and  carrying its  axes parallelly  to the ones of the reference frame of the laboratory, where the electron moves.  In view of its peculiar  properties,  the spin cannot actually  have a  complete classical corresponding and thus that interpretation is untenable. For instance, one cannot ``stop'' the spin of a particle or change the constant value of $S^2 = S_x^2+S_y^2+S_z^2$: It is a given  property of the particle like the mass.
 The electron spin is described within  an {\em internal} Hilbert space ${\cal H}_{s}$, which has dimension $2$. Identifying ${\cal H}_s$ with $\mathbb C^2$, the three spin observables are defined in terms of the three Hermitian matrices (occasionally re-introducing the constant $\hbar$)
\beq S_x = \frac{\hbar}{2} \sigma_x\:, \qquad S_y = \frac{\hbar}{2} \sigma_y \:, \qquad S_z = \frac{\hbar}{2} \sigma_z\:, \label{spin}\eeq 
where we have introduced the well known  {\em Pauli matrices},
\beq \sigma_x =  \left[\begin{matrix}0\:  & \: 1 \\ 1 \:& \: \:\: 0\end{matrix}\right] \:, \quad 
\sigma_y =  \left[\begin{matrix}0 & -i \\ i &\: \:0\end{matrix}\right] \:, \quad 
\sigma_z =  \left[\begin{matrix}1 & \:\:0\\ 0 & -1\end{matrix}\right]\:.\label{pauli}\eeq
Notice that $[S_a, S_b] \neq 0$ if $a\neq b$ so that the components of the spin are incompatible observables. In fact one has
$$[S_x,S_y] = i \hbar S_z $$
and this identity holds also cyclically permuting  the three indices. These commutation relations are the same as for the observables $L_x$,$L_y$,$L_z$ describing the angular momentum referred to the laboratory system which have classical corresponding (we shall return on these observables in example \ref{exL}). So, differently from CM, the observables describing the components of the angular momentum are incompatible, they cannot be measured simultaneously. However the failure of the compatibility is related to the appearance of $\hbar$ on the right-hand side of 
$$[L_x,L_y] = i \hbar L_z\:. $$
That number is extremely small if compared with macroscopic scales. This is the ultimate reason why the incompatibility of $L_x$ and $L_z$ is negligible for macroscopic systems. \\
Direct inspection proves that $\sigma(S_a) = \{\pm \hbar/2\}$. Similarly $\sigma(L_a) = \{n \hbar \:|\: n \in \mathbb Z \}$. Therefore, differently from CM, the values of angular momentum components form  a discrete set of reals in QM.
Again  notice that the difference of two closest values is extremely small if compared with typical values of the angular momentum of macroscopic systems. This is the practical  reason why this discreteness disappears at macroscopic level. \hfill  $\blacksquare$\\

\noindent Just a few words about the time evolution and  composite systems   \cite{ercolessi} are necessary now, a wider discussion on the time evolution  will take place later in this paper.

\subsection{Time evolution}
Among the class of observables of a  quantum system described  in a given inertial reference frame, an observable $H$ called the (quantum) {\em Hamiltonian} plays a fundamental r\^ole. We are assuming here that 
the system interacts with a stationary environment. 
The one-parameter group of unitary operators associated to $H$ (exploiting (\ref{sectradec1}) to explain the notation)
\begin{equation} U_t := e^{-itH}  := \sum_{h \in \sigma(H)}e^{-ith} P^{(H)}_h\:,\quad t \in \mathbb R \label{Ham}\end{equation}
describes the {\em time evolution of quantum states} as follows. If the state at time $t=0$ is represented by the unit vector $\psi \in \cal H$, the state at the generic time $t$ is represented by the vector
$$\psi_t = U_t \psi\:.$$ 

\remark Notice that $\psi_t$ has norm $1$ as necessary to describe states,  since $U_t$ is norm preserving it being unitary. \hfill $\blacksquare$\\

\noindent Taking (\ref{Ham}) into account,  this identity is equivalent to 
\begin{equation}i\frac{d \psi_t}{dt} = H \psi_t\:.\label{Sc}\end{equation}
Equation (\ref{Sc}) is nothing but a form of the celebrated {\em Schr\"odinger equation}. If the environment is not stationary,  a more complicated description can be given where $H$ is replaced by a class of Hamiltonian (selfadjoint)  operators parametrized in time, $H(t)$, with $t \in \mathbb R$. This time dependence accounts for  the time evolution of   the external system interacting with our quantum system.
 In that case, it is simply assumed that the time evolution of states is again described by the equation above where $H$ is replaced by $H(t)$:
$$i\frac{d \psi_t}{dt} = H(t) \psi_t\:.$$
Again,  this equation permits one  to define a two-parameter {\em groupoid} of unitary operators $U(t_2,t_1)$, where $t_2,t_1 \in \mathbb R$, such that
$$\psi_{t_2} = U(t_2,t_1)\psi_{t_1}\:, \quad t_2, t_1 \in \mathbb R\:.$$
The  groupoid structure arises from the following identities: $U(t,t)=I$ and $U(t_3,t_2)U(t_2,t_1)= U(t_3,t_2)$ and $U(t_2,t_1)^{-1}= U(t_2,t_1)^* = U(t_1,t_2)$.

\remark\label{rem8} In our elementary case where $\cal H$ is finite dimensional,  {\em Dyson's formula} holds with the simple hypothesis that  the map $\mathbb R \ni t \mapsto H_t \in L({\cal H})$ is continuous (adopting any topology compatible with the vector space structure of $L({\cal H})$) \cite{moretti} 
$$U(t_2,t_1) = \sum_{n=0}^{+\infty} \frac{(-i)^n}{n!} \int_{t_1}^{t_2} \cdots \int_{t_1}^{t_2}  T[H(\tau_1)\cdots H(\tau_n)]  \: d\tau_1 \cdots d\tau_n\:.$$
Above, we define $T[H(\tau_1)\cdots H(\tau_n)]  = H(\tau_{\pi(1)})\cdots H(\tau_\pi(n))$, where the bijective function  $\pi :\{1,\ldots, n\}
\to \{ 1,\ldots, n\}$ is any permutation with $\tau_{\pi(1)} \geq \cdots \geq  \tau_{\pi(n)}$.  \hfill $\blacksquare$

\subsection{Composite systems} If a quantum system $S$ is made of two parts, $S_1$ and $S_2$,  respectively described in the Hilbert spaces ${\cal H}_1$ and ${\cal H}_2$, it is assumed that the whole system is described in the  space ${\cal H}_1 \otimes {\cal H}_2$ equipped with the  unique Hermitian scalar product $\langle \cdot | \cdot \rangle $  such that $\langle \psi_1\otimes \psi_2|  \phi_1\otimes \phi_2 \rangle = \langle \psi_1 | \phi_1 \rangle_1 \langle \psi_2 | \phi_2 \rangle_2$
 (in the infinite dimensional case ${\cal H}_1 \otimes {\cal H}_2$ is the Hilbert completion of the afore-mentioned algebraic tensor product).\\
 If ${\cal H}_1 \otimes {\cal H}_2$ is the space of a composite system $S$  as before  and $A_1$ represents an observable for  the part $S_1$, it is naturally identified with the selfadjoint operator  $A_1 \otimes I_2$
defined in ${\cal H}_1 \otimes {\cal H}_2$. A similar statement holds swapping $1$ and $2$.  Notice that 
$\sigma(A_1 \otimes I_2)= \sigma (A_1)$ as one easily proves. (The result survives the extension to the infinite dimensional case.)

\remark $\null$

{\bf (a)}  Composite systems are in particular systems made of many (either identical or not) particles.
If we have a pair of particles respectively described in  the Hilbert space ${\cal H}_1$ and ${\cal H}_2$, the full system is described in  ${\cal H}_1 \otimes {\cal H}_2$. Notice that the dimension of the final space is the {\em product} 
of the dimension of the component spaces. In CM the system would instead be described in a space of phases which is the Cartesian product of the two spaces of phases. In that case the dimension would be the {\em sum}, rather than the product,  of the dimensions of the component spaces. 

{\bf (b)} ${\cal H}_1 \otimes {\cal H}_2$ contains the  so-called {\em entangled states}. They are  states  represented by vectors   {\em not} 
factorized as $\psi_1 \otimes \psi_2$,  but they are {\em linear combinations} of such vectors. Suppose the whole  state is represented by the entangled state $$\Psi =\frac{1}{\sqrt{2}}\left(\psi_a \otimes \phi  + \psi_{a'}\otimes \phi'\right)\:,$$ where  $A_1\psi_ a = a \psi_a$ and 
 $A_1\psi_ {a'} = a' \psi_{a'}$ for a certain observable $A_1$ of the part $S_1$ of the total system. 
Performing a measurement of $A_1$  on $S_1$, due to the collapse of state phenomenon, we automatically act one the whole state and on the part describing $S_2$. As a matter of fact, up to normalization,  the state of the full system after the measurement of $A_1$ will be $\psi_a \otimes \phi$ if the outcome of $A_1$ is $a$, or it will  be $\psi_{a'} \otimes \phi'$ if the outcome of $A_1$ is $a'$. It could happen that the two measurement apparatuses, respectively measuring $S_1$ and $S_2$, are localized very far in the physical space. Therefore acting on $S_1$ by measuring $A_1$, we ``instantaneously'' produce a change of $S_2$ which can be seen performing mesurements on it,  even if the measurement apparatus of $S_2$ is very far from the one of $S_1$. 
This seems to contradict the fundamental relativistic postulate, the  {\em locality} postulate,  that  there is a maximal speed, the one of light,  for propagating  physical information. After the famous analysis of Bell, improving the original one by Einstein, Podolsky and Rosen,  the phenomenon  has been experimentally observed. Locality is truly violated, but in a such subtle way which does not allows   superluminal propagation of physical information.
Non-locality of QM is nowadays widely accepted as a real and fundamental feature of Nature \cite{ghirardi, stanford}.  \hfill  $\blacksquare$

\example  An electron also possesses an  {\em electric charge}. That is another {\em internal} quantum observable, $Q$, with two values $\pm e$, where $e= −1.602176565 \times 10^{-19} C$ is the value elementary electrical charge.
So there are two types of electrons. {\em Proper electrons}, whose internal state of charge is an  eigenvector of $Q$ with eigenvalue $-e$ and {\em positrons}, whose   internal state of charge is a  eigenvector of $Q$ with eigenvalue $e$. The simplest version of  the internal Hilbert space of the electrical charge is therefore ${\cal H}_c$ which\footnote{As we shall say later, in view of a {\em superselection rule} not all normalized vectors of ${\cal H}_c$ represent (pure) states.}, again, is  isomorphic to $\mathbb C^2$. With this representation $Q = e \sigma_3$.
The full Hilbert space of an electron must contain a factor ${\cal H}_s\otimes {\cal H}_c$. Obviously 
this is by no means sufficient to describe an electron, since we must introduce at least the observables describing the position of the particle in the physical space at rest with a reference (inertial) frame. \hfill $\blacksquare$

\subsection{A first look to the infinite dimensional case, CCR and quantization procedures}\label{subsecID} 
All the described formalism, barring technicalities we shall examine in the rest of the paper, holds also for quantum systems whose complex vector space of the states is {\em infinite} dimensional. \\
To extend the ideas treated in  Sect. \ref{subsecFD}  to the general case, dropping the hypothesis that ${\cal H}$ is finite dimensional, it seems to be natural to assume that ${\cal H}$ is complete with respect to the norm associated to $\langle \cdot | \cdot \rangle$.  In particular,  completeness assures the existence of spectral decompositions, generalizing (\ref{sectradec0}) for instance when referring to {\em compact} selfadjoint operators (e.g., see \cite{moretti}).  In other words, ${\cal H}$  is a {\em complex Hilbert space}. \\
The most elementary example of a quantum system described in an infinite dimensional Hilbert space 
is a quantum particle whose position is along the axis $\mathbb R$. In this case \cite{ercolessi}, the Hilbert  space is ${\cal H} := L^2(\mathbb R, dx)$, $dx$ denoting the standard Lebesgue measure on $\mathbb R$. 
States are still represented by elements of  $P\cal H$, namely equivalence classes $[\psi]$ of  measurable functions $\psi : \mathbb R \to \mathbb C$ with unit norm, $||[\psi]|| = \int_{\mathbb R} |\psi(x)|^2 dx =1$. 

\remark  We therefore have here {\em two} quotient procedures. $\psi$ and $\psi'$ define the same element $[\psi]$ of $L^2(\mathbb R, dx)$ iff $\psi(x)-\psi'(x) \neq 0$ on a zero Lebesgue measure set.  Two unit vectors $[\psi]$ 
and $[\phi]$ define the same state if $[\psi] = e^{ia}[\phi]$ for some $a\in \mathbb R$.  \hfill $\blacksquare$

\notation
{\em In the rest of the paper we adopt the standard convention of many textbooks on functional analysis denoting by $\psi$, instead of  $[\psi]$, the elements of spaces $L^2$ and tacitly identifying pair of functions which are different on a zero measure set.} 
\hfill $\blacksquare$\\

\noindent The functions $\psi$ defining (up to zero-measure set and phases) states,  are called {\em wavefunctions}. There is  a pair of fundamental observables describing our quantum particle moving in $\bR$. One is the {\em position observable}. The corresponding selfadjoint operator, denoted by  $X$, is defined as follows
$$(X\psi)(x):= x\psi(x)\:, \quad x \in \mathbb R \:,  \quad \psi \in L^2(\mathbb R, dx)\:.$$ 
The other observable is the one associated to the momentum and indicated by $P$. Restoring $\hbar$ for the occasion, the {\em momentum operator} is
$$(P\psi)(x) := -i \hbar \frac{d\psi(x)}{dx}\:,  \quad x \in \mathbb R \:, \quad \psi \in L^2(\mathbb R, dx)\:.$$

\noindent We immediately face  several mathematical problems with these, actually quite naive, definitions. Let us first focus on $X$. First of all, generally $X\psi \not \in L^2(\mathbb R, dx )$
even if $\psi \in L^2(\mathbb R ,dx)$.  To fix the problem, we can simply restrict the domain of $X$ to the  linear subspace of $L^2(\mathbb R, dx)$
\begin{equation} D(X) := \left\{ \psi \in L^2(\mathbb R, dx) \:\left|\: \int_{\mathbb R}  |x\psi(x)|^2 dx < +\infty\right.\right\}\:.\label{DX}\end{equation}
Though it holds 
\begin{equation} \langle X\psi |\phi \rangle = \langle \psi | X \phi \rangle \quad \mbox{for all $\psi, \phi \in D(X)$,}\label{hermX}\end{equation}
we cannot say that $X$ is selfadjoint simply because we have not yet given the definition of adjoint operator of an operator defined in a non-maximal domain in an infinite dimensional Hilbert space. In this general case, the identity  (\ref{agg}) does not define a (unique) operator $X^*$ without further technical requirements. We just say here, to comfort the reader, that 
$X$ is truly selfadjoint with respect to a general definition we shall give in the next section, when its domain is (\ref{DX}).\\ 
Like (\ref{agg2}) in the finite dimensional case,  the identity (\ref{hermX}) implies that all eigenvalues of $X$ must be real if any. Unfortunately, for every fixed $x_0 \in \mathbb R$ there is no $\psi \in L^2(\mathbb R, dx)$ with $X\psi = x_0 \psi$ and $\psi \neq 0$. (A function $\psi$ satisfying $X\psi = x_0 \psi$ must also satisfy $\psi(x)=0$ if $x\neq x_0$, due to the definition of $X$. Hence  $\psi=0$, as an element of $L^2(\mathbb R, dx)$
just because  $\{x_0\}$ has zero Lebesgue measure!)
All that seems to prevent the existence of a spectral decomposition of $X$ like the one in  (\ref{sectradec0}), since $X$ does not admit eigenvectors in $L^2(\mathbb R, dx )$ (and {\em a fortiori} in $D(X)$). \\
The definition of $P$ suffers from  similar troubles. The domain of $P$ cannot be the whole $L^2(\mathbb R,dx)$ but should be restricted to a subset of (weakly) differentiable functions with derivative in $L^2(\mathbb R, dx)$. The simplest definition is
\begin{equation} D(P) := \left\{ \psi \in L^2(\mathbb R, dx) \:\left|\: \exists \:\mbox{w-}\frac{d\psi(x)}{dx}\:, \: \int_{\mathbb R}  \left| \mbox{w-}\frac{d\psi(x)}{dx}\right|^2 dx < +\infty\right.\right\}\:.\label{DP}\end{equation}
 Above $\mbox{w-}\frac{d\psi(x)}{dx}$ denotes the {\em weak derivative} of $\psi$\footnote{$f : \mathbb R \to \mathbb C$, defined up to zero-measure set, is the weak derivative of $g \in L^2(\mathbb R, dx)$ if it holds $\int_{\mathbb R} g \frac{dh}{dx} dx = -\int_{\mathbb R} fh dx$ for every $h \in C_0^\infty(\mathbb R)$.  If $g$ is differentiable,
its standard derivative coincide  with the weak one.}. As a matter of fact $D(P)$ coincides with the {\em Sobolev space} $H^1(\mathbb R)$.\\
Again, without a precise definition of adjoint operator in an infinite dimensional Hilbert space (with non-maximal domain) we cannot say anything  more precise about the selfadjointness of $P$ with that domain.  We  say however  that 
$P$ turns out to be selfadjoint with respect to the general definition we shall give in the next section provided  its domain is (\ref{DP}).\\ 
From the definition of the domain of $P$ and passing to the Fourier-Plancherel transform, one  finds again (it is not so easy to see it)
\begin{equation} \langle P\psi |\phi \rangle = \langle \psi | P \phi \rangle \quad \mbox{for all $\psi, \phi \in D(P)$,}\label{hermP}\end{equation}
so that, eigenvalues are real if exist.
However  $P$ does not admit eigenvectors. 
The naive eigenvectors with eigenvalue $p \in \mathbb R$ are functions proportional to the map $\mathbb R \ni x \mapsto e^{ipx/\hbar}$, which does not belong to $L^2(\mathbb R, dx)$ nor  $D(P)$.
We will tackle all these issues  in the next section in a very general fashion.\\  
We observe that the space of {\em Schwartz functions}, ${\cal S}(\mathbb R)$
\footnote{${\cal S}(\mathbb R^n)$ is the vector space of  the  $C^\infty$ complex valued functions on $\mathbb R^n$ which, together with their derivatives of all orders in every set of  coordinate,  decay  faster than every negative integer  power of $|x|$ for $|x|\to +\infty$.}
 satisfies 
$${\cal S}(\mathbb R) \subset D(X) \cap D(P)$$
and furthermore ${\cal S}(\mathbb R)$ is dense in $L^2(\mathbb R, dx)$ and {\em invariant} under $X$ and $P$: $X({\cal S}(\mathbb R))\subset {\cal S}(\mathbb R)$ and $P({\cal S}(\mathbb R)) \subset {\cal S}(\mathbb R)$.

\remark  Though we shall not pursue this approach within these notes, we stress that $X$ admits a set of eigenvectors
 if we extend the domain of $X$   to 
 the space  ${\cal S}'(\mathbb R)$ of {\em Schwartz distributions} in a standard way: If $T \in {\cal S}'(\mathbb R)$,
$$\langle X(T), f \rangle := \langle T, X(f) \rangle \quad \mbox{for every  $f \in {\cal S}(\mathbb R)$.}$$
With this extension, the eigenvectors in ${\cal S}'(\mathbb R)$  of $X$ with eigenvalues $x_0 \in \mathbb R$ are the distributions $c\delta(x-x_0)$ \cite{ercolessi}.   This class of eigenvectors can be exploited  to build  a spectral decomposition of $X$  similar to  that  in   (\ref{sectradec0}).\\
Similarly, $P$ admits eigenvectors in ${\cal S}'(\mathbb R)$ with the same procedure. They are just  the above exponential functions. Again, this calss of eigenvectors can be used to construct a spectral decomposition of $P$ like the one in   (\ref{sectradec0}). 
 The idea of this procedure can be traced back to Dirac \cite{Dirac11} and, in fact, something like ten years later it gave rise to the rigorous   {\em theory of distributions} by L. Schwartz. The modern formulation of this approach to construct spectral decompositions of selfadjoint operators was developed by Gelfand 
in terms of the so called {\em rigged Hilbert spaces} \cite{gelfand}.  \hfill $\blacksquare$ \\

\noindent Referring to a quantum particle moving in $\mathbb R^n$, whose Hilbert space is $L^2(\mathbb R^n, dx^n)$, one can introduce observables $X_k$ and $P_k$ representing position and momentum with respect to the $k$-th axis, $k=1,2,\ldots, n$.
These operators, which are defined analogously to the case $n=1$, have domains 
smaller than the full Hilbert space. We do not write the form of these domain (where the operators turn out to be properly selfadjoint referring to the general definition we shall state  in the next section). We just mention the fact  that all these operators admit ${\cal S}(\mathbb R^n)$
as common invariant subspace included in their domains. Thereon
\begin{equation}
(X_k\psi)(x)= x_k \psi(x) \:, \qquad (P_k\psi)(x)= -i \hbar \frac{\partial \psi(x)}{\partial x_k}\:, \quad \psi \in {\cal S}(\mathbb R^n)\label{XP}
\end{equation}
and so
\begin{equation} \langle X_k\psi |\phi \rangle = \langle \psi | X_k \phi \rangle\:, \quad
\langle P_k\psi |\phi \rangle = \langle \psi | P_k \phi \rangle\quad 
 \mbox{for all $\psi, \phi \in {\cal S}(\mathbb R^n)$,}\label{herm}\end{equation}

\noindent By direct inspection one easily proves that the {\em canonical commutation relations} (CCR) hold when all the operators in the subsequent formulas are supposed to be restricted to ${\cal S}(\mathbb R^n)$
\begin{equation}
[X_h, P_k] = i\hbar \delta_{hk} I\:,\quad [X_h, X_k] =0\:, \quad  [P_h, P_k] = 0\:. \label{CCR}
\end{equation}
We have introduced the {\em commutator} $[A,B]:= AB-BA$ of the operators $A$ and $B$ generally with different domains, defined on a subspace where both compositions $AB$ and $BA$ makes sense,  ${\cal S}(\mathbb R^n)$ in the considered case. Assuming that (\ref{E}) and  (\ref{Delta2}) are still valid for $X_k$ and $P_k$ referring to  $\psi \in {\cal S}(\mathbb R^n)$,  (\ref{CCR}) easily leads to the {\em Heisenberg uncertainty relations},
\begin{equation}
\Delta X_{k\psi} \Delta P_{k\psi} \geq \frac{\hbar}{2}\:, \quad \mbox{for}\:\: \psi \in {\cal S}(\mathbb R^n)\:,\quad  ||\psi||=1 \label{Heisenberg}\:.
\end{equation}

{\bf \exercise\label{eH}}   {\em Prove inequality  (\ref{Heisenberg}) assuming  (\ref{E})  and (\ref{Delta2}).} 

{\bf Solution}.  Using (\ref{E}), (\ref{Delta2}) and the Cauchy-Schwarz  inequality, one easily finds (we omit the index $_k$ for simplicity),
$$\Delta X_\psi \Delta P_\psi = ||X'\psi||  ||P'\psi||
\geq |\langle X'\psi | P'\psi\rangle|\:.$$
where $X' := X - \langle X\rangle_\psi I$ and $P' := X - \langle X\rangle_\psi I$.
 Next notice that
$$ |\langle X'\psi | P'\psi\rangle| \geq 
|Im \langle X' \psi | P'\psi\rangle | = \frac{1}{2} | \langle X' \psi | P'\psi\rangle - \langle P' \psi | X'\psi\rangle|$$
Taking advantage from (\ref{herm}) and the definitions of $X'$ and $P'$ and exploiting (\ref{CCR}),
$$ | \langle X' \psi | P'\psi\rangle - \langle P' \psi | X'\psi\rangle| = 
 | \langle  \psi | (X'P'- P'X')\psi\rangle|  = | \langle  \psi | (XP- PX)\psi\rangle| =\hbar |\langle \psi|\psi \rangle|$$
Since $\langle \psi|\psi \rangle= ||\psi||^2=1$ by hypotheses, (\ref{Heisenberg}) is proved.
Obviously the open problem is to justify the validity of (\ref{E})  and (\ref{Delta2})  also in the infinite dimensional case. \hfill $\Box$ \\

\noindent Another philosophically  important consequence of the CCR (\ref{CCR}) is that they resemble the {\em classical canonical commutation relations} of the  Hamiltonian variables $q^h$ and $p_k$, referring to the standard {\em Poisson brackets} $\{\cdot, \cdot \}_P$,
\begin{equation}
\{q^h, p_k\}_P = \delta^h_k \:,\quad \{q^h, q^k\}_P =0\:, \quad  \{p_h, p_k\}_P = 0\:. \label{CCRc}
\end{equation}
 as soon as one identifies   $(i\hbar)^{-1}[\cdot, \cdot]$ with $\{\cdot, \cdot\}_P$. This fact, initially noticed by Dirac \cite{Dirac11}, leads to the idea of ``quantization'' of a classical Hamiltonian theory \cite{ercolessi}. \\
One starts from a classical system described on a symplectic manifold $(\Gamma, \omega)$, for instance $\Gamma =\mathbb R^{2n}$ equipped with the standard symplectic form as $\omega$ and 
considers  the (real) Lie algebra $(C^\infty(\Gamma, \mathbb R), \{\cdot, \cdot\}_P)$.
To ``quantize'' the system one  looks for a map associating classical observables $f \in C^\infty(\Gamma, \mathbb R)$
to quantum observables $O_f$, i.e.  selfadjoint operators restricted\footnote{The restriction should be such that it admits a unique selfadjoint extension. A sufficient requirement on ${\cal S}$ is that every $O_f$ is {\em essentially selfadjoint} thereon, notion we shall discuss in the next section.} to a common invariant domain ${\cal S}$ of a certain Hilbert space $\cal H$. (In case $\Gamma = T^*Q$, $\cal H$ can be chosen as $L^2(Q, d\mu)$ where $\mu$ is some natural measure.) The map $f \mapsto O_f$ is expected to satisfy a set of constraints. The most important are listed here
\begin{enumerate}
\item $\mathbb R$-linearity;
\item $O_{id} = I|_{\cal S}$;
\item $O_{\{f,g\}_P} = -i\hbar[O_f,O_g]$ 
\item If $(\Gamma, \omega)$ is $\mathbb R^{2n}$ equipped with the standard symplectic form, they must hold 
$O_{x_k}= X_k|_{\cal S}$ and $O_{p_k}= P_k|_{\cal S}$, $k=1,2,\ldots,n$.
\end{enumerate}
The penultimate requirement says that the map $f \mapsto O_f$ transforms the real Lie algebra 
$(C^\infty(\Gamma, \mathbb R), \{\cdot, \cdot\}_P)$ into a real Lie algebra of operators whose  Lie bracket is $i\hbar[O_f,O_g]$.
A map fulfilling these constraints, in particular the third one, is possible if $f$, $g$ are both functions of only the $q$ or the $p$ coordinates  separately or if
they are linear in them. But it is false  already if we consider elementary physical systems \cite{ercolessi}. The ultimate reason of this obstructions due to the fact that
the operators $P_k$, $X_k$ do not commute, contrary to the functions $p_k$, $q^k$ which do.
The problem can be solved, in the paradigm of the so-called {\em Geometric Quantization}\cite{ercolessi}, replacing $(C^\infty(\Gamma, \mathbb R), \{\cdot, \cdot\}_P)$ with a 
sub-Lie algebra (as large as possible). There are other remarkable procedures of ``quantization'' in the literature, we shall not insist on them any further here \cite{ercolessi}.

\example $\null$\\
{\bf (a)} The full Hilbert space of an electron is therefore given by the  
tensor product $L^2(\mathbb R^3, d^3x)\otimes {\cal H}_s \otimes {\cal H}_c$.\\
{\bf (b)} Consider a particle in $3D$ with mass $m$, whose potential energy is a bounded-below real function $U \in C^\infty(\mathbb R^3)$ with polynomial growth. Classically, its Hamiltonian function reads
$$h:= \sum_{k=1}^3\frac{p_k^2}{2m} + U(x)\:.$$
A brute force quantization procedure in $L^2(\mathbb R^3, d^3x)$  consists of replacing every classical object with corresponding operators. It may make sense at most  when there are no ordering ambiguities in translating functions like $p^2x$, since classically $p^2x = pxp = xp^2$, but these identities are false at quantum level. In our case these problems do not arise so that
\beq H:= \sum_{k=1}^3\frac{P_k^2}{2m} + U\:,\label{ham}\eeq
where $(U\psi)(x) := U(x) \psi(x)$, could be accepted as first quantum model of the Hamiltonian function of our system. The written operator is at least defined on ${\cal S}(\mathbb R^3)$, where it satisfies
$\langle H \psi | \phi \rangle = \langle \psi | H\phi \rangle $. The existence of selfadjoint extensions is a delicate issue \cite{moretti} we shall not address here. Taking 
(\ref{XP}) into account, always on ${\cal S}(\mathbb R^3)$, one immediately finds
$$H:= -\frac{\hbar^2}{2m} \Delta + U\:,$$
where $\Delta$ is the standard Laplace operator in $\mathbb R^3$. If we assume that the equation describing the evolution of the quantum system is again\footnote{A factor $\hbar$ has to be added in front of the left-hand side of (\ref{Sc}) if we deal with a unit system where $\hbar \neq 1$.} (\ref{Sc}), in our case we find
the known form of the Schr\"odinger equation,
$$i\hbar \frac{d\psi_t}{dt} = -\frac{\hbar^2}{2m} \Delta \psi_t + U \psi_t\:,$$
when $\psi_\tau \in {\cal S}(\mathbb R^3)$ for $\tau$ varying in a neighborhood of $t$ (this requirement may be relaxed).
Actually the meaning of the derivative on the left-hand side should be specified. We only say here that it is computed with respect to the natural topology of $L^2(\mathbb R^3, d^3x)$.  \hfill $\blacksquare$

\section{Observables in infinite dimensional Hilbert spaces: Spectral Theory}\label{secstatic}
The main goal of this section is to present a suitable mathematical  theory, sufficient  to extend to the infinite dimensional case the mathematical formalism of QM introduced in the previous section. As seen in Sect. \ref{subsecID},  the main issue  concerns the fact that, in the infinite dimensional case,  there are operators representing observables which do not have proper eigenvalues and eigenvectors, like $X$ and $P$. So,  naive expansions as (\ref{sectradec0})   cannot be literally extended to the general case. These expansions, together with the interpretation of the eigenvalues as values attained by the observable associated with a selfadjoint operator,  play a crucial r\^ole in the mathematical interpretation of the quantum phenomenology introduced in Sect. \ref{subsec1} and mathematically discussed in Sect. \ref{subsecFD}. In particular we need a precise definition of selfadjoint operator and a result regarding a spectral decomposition in the infinite dimensional case.  These tools are basic elements of the so called {\em spectral theory in Hilbert spaces}, literally invented  by von Neumann in his famous book  \cite{vonNeumann}
 to give a rigorous form to Quantum Mechanics and successively developed by various authors towards  many different directions of pure and applied mathematics. 
 The same notion of abstract Hilbert space, as  nowadays  known,  was born in the second chapter of  that book, joining and elaborating  previous  mathematical constructions  by Hilbert and Riesz.
 The remaining part of this section is devoted to introduce the reader to  some basic  elements of that formalism. Reference books are, e.g.,  \cite{R,moretti,S,analysisnow}

\subsection{Classes of  (especially unbounded) operators in  Hilbert spaces} As is well known,  a  {\bf complex  Hilbert space}  is a complex vector space, $\cH$, equipped with a Hermitian  scalar product $\langle\cdot|\cdot\rangle$ 
-- for us the anti-linear entry being  the left one -- 
and $\cH$ is complete with respect to the norm  $||x|| := \sqrt{\langle x|x \rangle}$, $x \in \cH$.\\
In particular, just in view of positivity of the scalar product and regardless
the completeness property, the {\bf Cauchy-Schwarz inequality} holds 
$$|\langle x|y\rangle| \leq ||x||\:||y||\:,\quad x,y \in \cH\:.$$
\noindent Another elementary purely algebraic fact is 
the {\bf polar decomposition} of the Hermitian  scalar product (here, $\cH$ is not necessarily complete)
\beq 4 \langle x| y \rangle = || x+y ||^2 - ||x-y||^2 -i||x+i y||^2 +i|| x-iy||^2\label{polardec}\quad \mbox{for of $x,y \in \cH$,}\eeq
 which immediately implies the following elementary result.
\begin{proposition}  If $\cH$ is a complex vector space with Hermintian scalar product $\langle\:\:|\:\:\rangle$, a linear map $L: \cH \to \cH$
which is an isometry -- $||Lx||=||x||$ if $x\in \cH$ --
 also preserves the scalar product -- $\langle Lx|Ly\rangle = \langle x|y\rangle$ for $x,y\in \cH$.
\end{proposition}
\noindent  The converse proposition is obviously true.\\
We henceforth assume that the reader be familiar with the basic theory of normed, Banach and Hilbert spaces and notions like {\em Hilbertian basis}  (also called  {\em complete orthonormal systems}) and that their properties and use be well known \cite{R,moretti}. 
We only remind the reader  the validity of an elementary though  fundamental tecnical result (e.g., see \cite{R,moretti}):
\begin{theorem}[Riesz' lemma]\label{RL} Let $\cH$ be a complex Hilbert space.  $\phi: \cH \to \bC$ is linear and continuous if and only if has the form $\phi = \langle x|\:\:  \rangle$ for some  $x \in \cH$. The vector $x$  is uniquely determined by $\phi$.
\end{theorem}

\noindent Our  goal is to present some basic results of {\em spectral analysis}, useful in QM.\\
From now on, an {\bf operator $A$ in ${\cH}$} always means a {\em linear} map  $A: D(A) \to \cH$, whose {\bf domain}, $D(A)\subset \cH$, is a {\em subspace} of $\cH$. In particular,
 $I$ always  denotes the {\bf identity operator} defined on the {\em whole} space  ($D(I)= \cH$) $$I : \cH \ni x \mapsto x \in \cH\:.$$
\noindent If $A$ is an operator in $\cH$, $Ran(A):= \{Ax\:|\: x \in D(A)\}$ is the {\bf image} or {\bf range} of $A$.

\notation If $A$ and $B$ are operators in $\cH$  $$A\subset B \mbox{  means that $D(A) \subset D(B)$  and $B|_{D(A)}=A$,}$$  
where $|_S$ is the standard ``restriction to $S$'' symbol. 
We also adopt usual  conventions regarding {\bf standard domains}  of combinations of operators $A,B$:

 (i) $D(AB) := \{x \in D(B) \:|\: Bx \in D(A)\}$

(ii) $D(A+B) := D(A) \cap D(B)$,

(ii)  $D(\alpha A)= D(A)$ for $\alpha \neq 0$.
 \hfill $\blacksquare$\\

\noindent To go on, we define some abstract algebraic structures naturally arising in the space of operators on a Hilbert space.

\begin{definition} {\em 
Let $\gA$ be an associative complex algebra $\gA$.\\
 {\bf (1)} $\gA$ is a {\bf Banach algebra} if  it is a Banach space such that $||ab|| \leq ||a||\:||b||$ for $a,b \in \gA$. An {\bf unital} Banach algebra is a Banach algebra with unit multiplicative element $1\!\!1$, satisfying $||1\!\!1||=1$.\\
{\bf (2)}  $\gA$ is an (unital)  {\bf $^*$-algebra} if it is an (unital) algebra equipped with an anti linear map $\gA \ni a \mapsto a^* \in \gA$, called {\bf involution}, such that $(a^*)^*=a$ and $(ab)^*=b^*a^*$ for $a,b \in \gA$.\\
{\bf (3)} $\gA$ is a (unital) {\bf $C^*$-algebra} if it is a (unital) Banach algebra $\gA$
which is also a $^*$-algebra and $||a^*a||= ||a||^2$ for $a \in \gA$.\\
A $^*$-{\bf homomorphism} from the $^*$-algebra $\cA$ to the the $^*$-algebra $\cB$ is an algebra homomorphism preserving the involutions (and the unities if both present). A bijective $^*$-homomorphism is called $^*$-{\bf isomorphism}.}  \hfill $\blacksquare$

\end{definition} 
{\bf \exercise}   {\em Prove that $1\!\!1^* = 1\!\!1$ in a unital $^*$-algebra
and that  $||a^*||=||a||$ if $a\in \gA$ when $\gA$ is a $C^*$-algebra.}

{\bf Solution}. From $ 1\!\!1 a = a  1\!\!1 =a$ and the definition of $^*$, we immediately have
  $a^* 1\!\!1^* = 1\!\!1^*a^* =a^*$. Since $(b^*)^* =b$, we have found that $b 1\!\!1^* = 1\!\!1^*b =b$ for every $b \in \gA$. Uniqueness of the unit implies $ 1\!\!1^*= 1\!\!1$.
Regarding the second property, $||a||^2 = ||a^*a|| \leq ||a^*||\: ||a||$ so that $||a|| \leq ||a^*||$. Everywhere replacing $a$ for $a^*$ and using $(a^*)^*$, we also obtain $||a^*|| \leq ||a||$, so that $||a^*||=||a||$. $\Box$\\

\noindent We remind the reader that a linear map $A : X \to Y$, where $X$ and $Y$ are normed complex vector spaces 
with resp. norms $||\cdot||_X$ and $||\cdot||_Y$, is said to be {\bf bounded} if 
\beq
||Ax||_Y \leq  b||x||_X \quad \mbox{for some $b \in [0,+\infty)$ and all $x\in X$.}\label{bounded}
\eeq
 As is well known \cite{R,moretti}, it turns out that: {\em  $A$ is continuous if and only	 if it is bounded}.  \\
From now on $\gB(\cH)$ denotes the set of  bounded operators  $A: \cH \to \cH$.
This set acquires the structure of a {\em unital Banach algebra}: The complex vector space structure is the standard  one of operators, the associative algebra product is the composition of operators with  unit given by  $I$, and the norm being the usual {\bf operator norm}, $$||A|| := \sup_{0 \neq x \in {\cal H}} \frac{||Ax||}{||x||}\:.$$
This definition of $||A||$ can be given also for an operator $A: D(A) \to \cH$, if $A$
is bounded and $D(A) \subset \cH$ but $D(A) \neq \cH$. It immediately arises that
$$||Ax|| \leq ||A||\:||x|| \quad \mbox{if $x\in D(A)$.}$$ 
$\gB(\cH)$ is also an unital  $C^*$-algebra if we introduce the notion of adjoint of an operator. To this end we have the following general definition concerning also unbounded operators defined on non-maximal domains.

\begin{definition} {\em  Let $A$ be a densely defined operator in the complex Hilbert space $\cH$.
Define the subspace of $\cH$,
$$D(A^*) := \left\{ y \in \cH \:|\: \exists z_y \in \cH \mbox{ s.t. }
\langle y|Ax\rangle   = \langle  z_y|x\rangle \: \forall x \in D(A)  \right\}\:.$$ 
The linear map $A^*: D(A^*)\ni y \mapsto z_y$ is called the {\bf adjoint} operator of $A$.}  \hfill $\blacksquare$
\end{definition} 

\remark$\null$

{\bf (a)} Above, $z_y$ is uniquely determined by $y$, since $D(A^*)$ is dense. If both $z_y,z_y'$ satisfy $\langle y|Ax\rangle   = \langle  z_y|x\rangle$ and $\langle y|Ax\rangle   = \langle  z'_y|x\rangle$, then $\langle  z_y-z_y'|x\rangle=0$ for every $x\in D(A)$. Taking a sequence $D(A)\ni x_n \to z_y-z_y'$, we conclude that $||z_y-z_y'||=0$. Thus $z_y=z_y'$.
The fact that $y \mapsto z_y$ is linear can immediately be checked.

 {\bf (b)} By construction, we immediately  have that 
 $$\langle A^*y|x\rangle = \langle y|Ax\rangle\quad \mbox{for $x\in D(A)$ and  $y\in D(A^*)$}$$ and also  $$\langle x| A^*y\rangle = \langle Ax|y \rangle\quad  \mbox{for $x\in D(A)$ and $y\in D(A^*)$}\:,$$
if taking the complex conjugation of the former identity. $\hfill \blacksquare$

{\bf \exercise}   {\em Prove that  $D(A^*)$ can equivalently be defined as the set (subspace) of $y \in \cal H$ such that the linear  functional $D(A) \ni x \mapsto \langle y|Ax\rangle$ is continuous.}

{\bf Solution}.  It is a  simple application of Riesz' lemma, after having uniquely extended 
$D(A) \ni x \mapsto \langle y|Ax\rangle$  to a continuous linear functional defined on  $\overline{D(A)}= \cH$ by continuity. \hfill $\Box$

\begin{remark}\label{remagg}$\null$

{\bf (a)} If $A$ is densely defined and $A\subset B$ then $B^* \subset A^*$. The proof is elementary.

{\bf (b)}  If $A \in \gB(\cH)$ then $A^* \in \gB(\cH)$ and $(A^*)^*=A$. Moreover $$||A^*||^2=||A||^2=||A^*A|| = ||AA^*||\:.$$

{\bf (c)} Directly from given definition of adjoint one has, for densely defined operators $A,B$ on $\cH$,
$$A^*+B^* \subset (A+B)^*\quad \mbox{and} \quad A^*B^* \subset (BA)^*\:.$$
Furthermore 
\beq A^*+B^*= (A+B)^* \quad \mbox{and} \quad A^*B^* = (BA)^*\:,\label{ON}\eeq
 whenever $B \in \gB(\cH)$ and $A$ is densely defined.

{\bf (d)} From (b) and the last statement in (c) in particular, it is clear that  $\gB(\cH)$ is a unital $C^*$-algebra  with  involution $\gB(\cH) \ni A \mapsto A^* \in \gB(\cH)$.    \hfill $\blacksquare$
\end{remark}

\begin{definition}[$^*$-{\bf representation}]\label{defrap}
{\em If $\gA$ is a (unital) $^*$-algebra and $\cH$ a Hilbert space, a $^*$-{\bf representation}
on $\cH$ is a $^*$-homomorphism  $\pi : \gA \to \gB(\cH)$ referring to the natural (unital) $^*$-algebra structure of $\gB(\cH)$.}   \hfill $\blacksquare$
\end{definition}

{\bf \exercise}  {\em Prove that $A^* \in \gB(\cH)$ if $A \in \gB(\cH)$ and that,
in this case $(A^*)^*=A$,   $||A||=||A^*||$ and $||A^*A||=||AA^*||=||A||^2$.}

{\bf  Solution}. If $A\in \gB(\cH)$, for every $y\in \cH$,
 the linear map $\cH \ni x \mapsto \langle y | Ax\rangle$ is continuous ($|\langle y | Ax\rangle| \leq ||y|| \: ||Ax|| \leq ||y||\:||A||\:||x||$) therefore Theorem \ref{RL} proves that there exists a unique 
$z_{y,A}\in \cH$ with $ \langle y | Ax\rangle=  \langle z_{y,A} | x\rangle$ for all $x,y \in \cH$. The map $\cH \ni y \mapsto z_{y,A}$ is linear as consequence of the said uniqueness and the antilinearity of the left entry of scalar product. The map   $\cH \ni y \mapsto z_{y,A}$ fits the definition of $A^*$, so it coincides with $A^*$ and $D(A^*)=\cH$.  Since $\langle A^*x| y\rangle = \langle x| A y \rangle$ for $x,y \in \cH$ implies (taking the complex conjugation)  $\langle y|A^*x \rangle = \langle  A y |x\rangle$ for $x,y \in \cH$, we have  $(A^*)^*=A$.
To prove that $A^*$ is bounded observe that
$||A^*x||^2 = \langle A^*x|A^*x \rangle = \langle x |AA^*x\rangle \leq ||x|| \: ||A||\:||A^*x||$ so that 
$||A^*x|| \leq ||A||\: ||x||$ and $||A^*||\leq ||A||$. Using $(A^*)^*=A$ we have $||A^*||=||A||$. Regarding the last 
identity, it is evidently enough to prove that $||A^*A||=||A||^2$. First of all $||A^*A||\leq ||A^*||\:||A|| = ||A||^2$, so that
$||A^*A|| \leq ||A||^2$. On the other hand $||A||^2 = (\sup_{||x||=1} ||Ax||)^2 = \sup_{||x||=1} ||Ax||^2 =
\sup_{||x||=1} \langle Ax|Ax\rangle = \sup_{||x||=1} \langle x|A^*Ax\rangle \leq \sup_{||x||=1}  ||x|| ||A^*Ax|| =
  \sup_{||x||=1}   ||A^*Ax|| = ||A^*A||$. We have found that $||A^*A|| \leq ||A||^2\leq ||A^*A||$ so that  $||A^*A||=||A||^2$. \hfill $\Box$

\begin{definition}\label{defcore}{\em Let $A$ be an operator in the complex Hilbert space $\cH$.\\
{\bf (1)} $A$  is said to be {\bf closed} if  the {\bf graph}   of $A$, that is the set pairs $(x, Ax) \subset \cH \times \cH$ with $x\in D(A)$,  is closed in the product topology of $\cH\times \cH$. \\
{\bf (2)} $A$ is {\bf closable}
if it admits extensions in terms of closed operators. This is equivalent to say that  the closure of the  graph of $A$ is the graph of an operator, denoted by $\overline{A}$, and called the {\bf closure} of $A$.\\
{\bf (3)} If $A$ is closable, a subspace $S \subset D(A)$ is called {\bf core} for $A$ if $\overline{A|_S} = \overline{A}$.} \hfill $\blacksquare$
\end{definition} 

\begin{remark}\label{remarkclosure}$\null$

{\bf (a)} Directly from the definition,  $A$ is closable if and only if  there are no sequences of elements $x_n \in D(A)$ such that $x_n \to 0$ and $Ax_n \not \to 0$ as $n \to +\infty$. In this case $D(\overline{A})$ is made of the elements $x\in \cH$ such that $x_n \to x$ and $Ax_n \to y_x$ for some sequences $\{x_n\}_{n \in \bN} \subset D(A)$
and some $y_x \in D(A)$. In this case $\overline{A}x = y_x$.

{\bf (b)} As a consequence of (a) one has that,
if $A$ is closable, then $aA+bI$ is closable and   $\overline{aA+bI} = a \overline{A} + b I$ for every $a,b\in \bC$.

{\bf (c)}  Directly from the definition,  $A$ is closed if and only if 
$D(A) \ni x_n \to x \in \cH$ and $Ax_n \to y \in \cH$ imply both $x \in D(A)$ and $y=Ax$.

{\bf (d)} If $A$ is densely defined,
$A^*$ is closed from the definition of adjoint operator and (c) above. Moreover, a densely defined operator $A$ is closable if and oly if $D(A^*)$ is dense. In this case $\overline{A}= (A^*)^*$. For the proof see, e.g., \cite{moretti}.\end{remark} 
\noindent The Hilbert space version of the {\bf closed graph theorem} holds (e.g., see \cite{moretti}).  \hfill $\blacksquare$

\begin{theorem}[Closed graph Theorem]\label{cgt}
Let  $A: \cH \to \cH$ be an operator, $\cH$ being a complex Hilbert space. $A$  is closed if and only if $A \in \gB(\cH)$.
\end{theorem}

{\bf \exercise}  {\em Prove that, if $B \in \gB(\cH)$ and $A$ is  a closed operator in $\cH$ such that $Ran(B) \subset D(A)$, then  $AB \in \gB(\cH)$.}

{\bf Solution}. $AB$ is well defined by hypothesis and $D(AB)=\cH$.  Exploiting (c) in remark \ref{remarkclosure} and continuity of $B$, one easily finds that $AB$ is closed as well. Theorem \ref{cgt} finally proves that 
$AB \in \gB(\cH)$. \hfill $\Box$

\begin{definition}\label{defop} {\em    An operator $A$ in the complex Hilbert space $\cH$ is said to be

(1) {\bf  symmetric} if it is densely defined and   $\langle Ax|y\rangle = \langle x|Ay\rangle$ for  $x,y \in D(A)$,

 which is equivalent to say  that $A \subset A^*$.

(2) {\bf selfadjoint}  if it is symmetric and  $A=A^*$,

(3) {\bf essentially self-adjoint} if it is symmetric and  $(A^*)^* = A^*$. 

(4) {\bf unitary} if  $A^{*}A= AA^* = I$,

(5) {\bf normal} if it is closed, densely defined and $AA^*=A^*A$.}  \hfill $\blacksquare$
\end{definition}

\begin{remark}\label{remagg2}$\null$

{\bf (a)} If $A$ is unitary then $A, A^* \in \gB(\cH)$. Furthermore $A: \cH \to \cH$
is unitary if and only if it is surjective and norm preserving. (See the exercises
\ref{vex} below).

{\bf (b)} A selfdjoint operator  $A$ does not admit proper symmetric extensions. 
 (See the exercises \ref{esopagg} below).

{\bf (c)} A symmetric operator $A$ is always closable because $A \subset A^*$ and $A^*$ is closed ((d) remark \ref{remarkclosure}), moreover for that operator the following conditions are equivalent:
 
(i) $(A^*)^* = A^*$ ($A$ is essentially self adjoint), 

(ii) $\overline{A}= A^*$, 

(iii) $\overline{A}= (\overline{A})^*$. \\ If these conditions are valid, $\overline{A}= (A^*)^* =A^*$ is the unique selfadjoint extension of $A$ (e.g., see \cite{moretti} and the exercises \ref{esopagg} below).

{\bf (d)} Unitary and selfadjoint operators are cases of normal operators.  \hfill $\blacksquare$
\end{remark}

{\bf \exercise\label{esinvarianceU}}  {\em Let $U: \cH \to \cH$ be a unitary operator in the complex  Hilbert space $\cH$ and $A$
another  operator in $\cH$. Prove that $UAU^*$ with  domain $U(D(A))$ (resp. $U^*AU$
 with  domain $U^*(D(A))$)
is symmetric, selfadjoint, essentially selfadjoint, unitary, normal 
if $A$ is respectively   symmetric, selfadjoint, essentially selfadjoint, unitary, normal.}

{\bf Solution}. Since $U^*$ is unitary when $U$ is and $(U^*)^*=U$, it is enough to establish the thesis for  $UAU^*$.
First of all notice that $D(UAU^*) = U(D(A))$ is dense if $D(A)$ is dense  since $U$ is bijective and isometric and $U(D(A)) = \cH$ if $D(A)=\cH$ because $U$ is bijective. By direct inspection, applying the definition of adjoint operator, one sees that
$(UAU^*)^* = UA^*U^*$ and $D((UAU^*)^*)= U(D(A^*))$.
 Now, if $A$ is symmetric $A\subset A^*$ which implies $UAU^* \subset UA^*U^*= (UAU^*)^*$ so that $UAU^*$ is symmetric as well. If $A$ is selfadjoint $A=A^*$ which implies $UAU^* = UA^*U^* = (UAU^*)^*$ so that 
$UAU^*$ is self adjoint as well. If $A$ is essentially self adjoint it is symmetric and  $(A^*)^* = A^*$, so that
$UAU^*$ is symmetric and $U(A^*)^*U^* = UA^*U^*$ that is $(UA^*U^*)^* = UA^*U^*$ which means
$((UAU^*)^*)^* = (UAU^*)^*$ so that $ UA^*U^*$ is essentially selfadjoint.
If $A$ is unitary, we have $A^*A=AA^*=I$ so that $UA^*AU^*=UAA^*U^*=UU^*$ which, since $U^*U=I = UU^*$, is equivalent to $UA^*U^*UAU^*=UAU^*UA^*U^*=U^*U=I$, that is 
$(UA^*U^*)UAU^*=(UAU^*)UA^*U^*=I$ and thus $UAU^*$ is unitary as well. If $A$ is normal $UAU^*$ is normal too,  with the same reasoning as in the unitary case. \hfill $\Box$\\

\noindent An elementary though important result, helping understand  why in QM observables are very often described by selfadjoint operators 
which are unbounded and defined in proper subspaces,  is the following proposition  (see (c) in remark \ref{remst}).
\begin{theorem}[Hellinger-Toepliz theorem]\label{HT} Let $A$ be a self-adjoint operator in the complex Hilbert space $\cH$.
$A$ is bounded if and only if $D(A)=\cH$ (thus $A\in \gB(\cH)$).
\end{theorem}

\begin{proof} As $A=A^*$ we have $D(A^*)=\cH$. Since $A^*$ is closed, Theorem \ref{cgt} implies the $A^*(=A)$ is bounded. Conversely, if $A=A^*$ is bounded, since $D(A)$ is dense,  we can continuously extend it to a bounded operator $A_1 :\cH \to \cH$. That extension, by continuity,  trivially satisfies
$\langle A_1x|y \rangle = \langle x| A_1 y \rangle$ for all $x,y \in \cH$ thus $A_1$ is symmetric. (b) in remark \ref{remagg2} implies $A=A_1$.
\end{proof}

{\bf \exercise \label{vex}} $\null$\\
{\bf (1)} {\em Prove that if $A$ is unitary then $A, A^* \in \gB(\cH)$.}

{\bf Solution}. 
 It holds $D(A)=D(A^*)=D(I)=\cal H$ and $||Ax||^2 = \langle Ax|Ax\rangle = \langle x|A^*Ax\rangle
||x||^2$ if $x \in \cH$, so that $||A||=1$. Due to (b) in remark \ref{remagg}, $A^* \in \gB(\cH)$. $\Box$\\

\noindent {\bf (2)} {\em Prove that $A: \cH \to \cH$ is unitary if and only if is surjective and norm preserving.}

{\bf Solution}. If $A$ is unitary ((3) Def \ref{defop}), it  is evidently bijective, moreover  as $D(A^*)=\cH$
$||Ax||^2 = \langle Ax|Ax\rangle = \langle x| A^*A x\rangle = \langle x|x \rangle = ||x||^2$, so $A$ is isometric
If $A : \cH \to \cH$ is isometric its norm is $1$ and thus $A\in \gB(\cH)$. Therefore $A^* \in \gB(\cH)$.
 The condition $||Ax||^2= ||x||^2$ can be re-written
$ \langle Ax|Ax\rangle = \langle x|A^*A x \rangle = \langle x|x\rangle$ and thus $\langle x| (A^*A-I)x\rangle =0$
for $x\in \cH$. Using $x= y\pm z$ and $x= y \pm i z$, the found indentity implies  $\langle z| (A^*A-I)y\rangle =0$
for all $y,z \in \cH$. Taking $z= (A^*A-I)y$, we finally have $||(A^*A-I)y||=0$ for all $y\in \cH$ and thus 
$A^*A=I$. In particular $A$ is injective as it admits the left inverse $A^*$.
 Since $A$ is also surjective, it is bijective and thus its left inverse $(A^*)$ is also a right inverse, that is   $AA^*=I$. \\

\noindent {\bf (3)} {\em Prove that, if $A: \cH \to \cH$ satisfies $\langle x|Ax \rangle \in \bR$ for all $x\in \cH$ (and in particular if $A\geq 0$, which means  $\langle x|Ax \rangle \geq 0$
 for all $x\in \cH$), then $A^*=A$ and $A\in \gB(\cH)$.}

{\bf Solution}. We have $\langle x|Ax \rangle = \overline{\langle x|Ax \rangle} = \langle Ax|x \rangle =\langle x|A^*x \rangle$ where,  as $D(A)=\cH$, the adjoint $A^*$ is well defined everywhere on $\cH$. Thus $\langle x |(T-T^*)x \rangle =0$ for every $x\in \cH$. Using there  $x = y\pm z$ and $x= y \pm i z$ we obtain
$\langle y |(T-T^*)z \rangle =0$ for all $y,z\in \cH$. Choosing $y =(T-T^*)z$, we conclude that $T=T^*$. Theorem \ref{HT} concludes the proof.
\hfill $\Box$

\begin{example} \label{FPop} The {\bf Fourier transform}, ${\cal F} : {\cal S}(\bR^n) \to {\cal S}(\bR^n)$,
defined as\footnote{In QM, adopting units with $\hbar \neq 1$, $k\cdot x$ has to be replaced for $\frac{k\cdot x}{\hbar}$ and $(2\pi)^{n/2}$ for $(2\pi \hbar)^{n/2}$.} \beq({\cal F}f)(k) := \frac{1}{(2\pi)^{n/2}}\int_{\bR^n} e^{-ik\cdot x} f(x) d^nx\label{ft}\eeq
($k\cdot x$ being the standard $\bR^n$ scalar product of $k$ and $x$) is a bijective linear  map with inverse
\beq ({\cal F}_-g)(x) := \frac{1}{(2\pi)^{n/2}}\int_{\bR^n} e^{ik\cdot x} g(k) d^nk\:. \label{fti}\eeq
Both ${\cal F}$ and ${\cal F}_-$ preserve the scalar product (and thus the norm) of $L^2(\bR^n, d^nx)$. As a consequence (exercise), using the fact that ${\cal S}(\bR^n)$ is dense in $ L^2(\bR^n, d^nx)$, one easily proves that ${\cal F}$ and ${\cal F}_-$  uniquely continuously extend to unitary operators, respectively,
$\hat{\cal F} : L^2(\bR^n, d^nx) \to L^2(\bR^n, d^nk)$ and $\hat{\cal F}_- : L^2(\bR^n, d^nk) \to L^2(\bR^n, d^nx)$ such that $\hat{\cal F}^* = \hat{\cal F}^{-1}= 
\hat{\cal F}_-$. $\hat{\cal F}$ is the {\bf Fourier-Plancherel} (unitary) {\bf operator}. \hfill $\blacksquare$

\end{example}

{\bf \exercise \label{esopagg}} $\null$\\
\noindent {\bf (1)} {\em Prove that a selfadjoint operator  $A$ does not admit proper symmetric extensions.}

{\bf Solution}. Let $B$ be a symmetric extension of $A$.
$A\subset B$ then  $B^*\subset A^*$ for (a) in remark \ref{remagg}. As $A=A^*$ we have $B^* \subset A \subset  B$. Since $B\subset B^*$, we conclude that $A=B$. $\null$
\hfill $\Box$

\noindent {\bf (2)} {\em Prove that an essentially selfdjoint operator  $A$ admits a unique selfadjoint extension, and that this extension is $A^*$. }

{\bf Solution}. Let $B$ be a selfadjoint extension of the essentially selfadjoint operator $A$, so that $A\subset B$. Therefore
$A^* \supset B^* = B$ and $(A^*)^* \subset B^* =B$. Since $A$ is essentially selfadjoint, we have found $A^* \subset B$. Here $A^*$ is selfadjoint and $B$ is symmetric because selfadjoint. The previous exercise implies $A^*=B$. That is, every selfadjoint extension of $A$ coincides with $A^*$. \hfill $\Box$\\

\noindent
 If $A$ is a densely defined symmetric operator in the complex Hilbert space $\cH$,  define  the {\bf deficiency indices}$,  n_\pm := dim \cH_\pm$ (cardinal numbers in general) where $\cH_\pm$ are the (closed) subspaces of the solutions 
of $(A^* \pm iI)x_\pm =0$
\cite{R,moretti, S} .
\begin{proposition}\label{propindex}
If $A$ is a densely defined symmetric operator in the complex Hilbert space $\cH$ the following holds.\\
{\bf (a)}  $A$ is essentially selfadjoint (thus it admits an unique selfadjoint extension) if  $n_\pm =0$, that is
$\cH_\pm = \{0\}.$\\
{\bf (b)} $A$ admits selfadjoint extensions if and only if  $n_+=n_-$ and these extension are labelled by means of $n_+$ parameters.
\end{proposition}
\remark \label{remvon} An easy  sufficient condition, due to von Neumann, for $n_+=n_-$ is that $CA \subset AC$ where $C: \cH \to \cH$ is a {\bf conjugation} that  
is an isometric  surjective {\em antilinear}\footnote{In other words $C(\alpha x + \beta y)= \overline{\alpha}Cx + \overline{\beta}Cy$ if $\alpha,\beta \in \bC$ and $x,y \in \cH$.} map with $CC=I$ \cite{moretti}.\\
Taking $C$ as the standard conjugation of functions in $L^2(\mathbb R^n, d^nx)$, this result proves in particular that all  operators in QM of the {\em Sch\"ordinger} form  as (\ref{ham}) admit selfadjoint extensions when defined on dense domains. \hfill $\blacksquare$

{\bf \exercise} {\em Prove that a symmetric operator that admits a unique self-adjoint extension is necessarily essentially selfadjoint.}

{\bf Solution}. By (b) of Proposition \ref{propindex}, $n_+=n_-$. If $n_\pm \neq 0$ there are many selfadjoint extension. The only possibility for the uniqueness of the selfadjoint extension is $n_\pm =0$. (a)  of Proposition \ref{propindex} implies that $A$ is  essentially selfadjoint. 
$\null$ \hfill $\Box$\\

\noindent A very useful criterion to establish the essentially selfadjointness of a symmetric operator is due to Nelson. It relies upon an important definition.
\begin{definition}\label{defanalitic} {\em Let  $A$ be  an operator in the complex  Hilbert space  $\cH$. \\If  $\psi \in \cap_{n \in \bN} D(A^n)$ satisfies
$$\sum_{n=0}^{+\infty} \frac{t^n}{n!}||A^n\psi||< +\infty\quad \mbox{for some $t>0$,} $$
then $\psi$ is said to be an {\bf analytic vector} of $A$.} \hfill $\blacksquare$
\end{definition}

\noindent We can state Nelson's criterion here \cite{moretti}.
\begin{theorem}[Nelson's essentially selfadjointness criterium] Let $A$ be a symmetric operator in the complex Hilbert space $\cH$, $A$ is essentially selfadjoint if $D(A)$ contains a dense set  $D$ of analytic vectors (or -- which  is equivalent --a set $D$ of analytic vectors whose finite  span dense in $\cH$).
\end{theorem}
\noindent The above equivalence is due to the fact that a finite linear combination of analytic vector is an analytic vector as well, the proof being elementary. We have the following evident corollary.
\begin{corollary}\label{CN}
If $A$ is a symmetric operator admitting a Hilbertian basis of eigenvectors in $D(A)$, then $A$ is essentially selfadjoint.
\end{corollary}

\example \label{esXP} $\null$\\
{\bf (1)} For $m \in \{1,2,\ldots, n\}$, consider  the operators  $X'_m$ and $X''_m$ in  $L^2(\mathbb R^n, d^nx)$ with dense domains
$D(X'_m)  = C_0^\infty(\mathbb R^n ; \bC)$, $D(X''_m)  = {\cal S}(\mathbb R^n)$ for $x \in \bR^n$
and, for $\psi, \phi$ in the respective domains,
$$(X'_m\psi)(x) := x_m\psi(x)\:, \quad  (X''_m\phi)(x) := x_m\phi(x)\:,$$
where $x_m$is the $m$-th component of $x\in \bR^n$.
Both operators are symmetric but not selfadjoint. They admit selfadjoint extensions 
because they commute with the standard complex conjugation of functions (see remark 
\ref{remvon}).
 It is furthermore possible to prove that both operators are 
essentially selfadjoint as follows.  First define the {\bf $k$-axis position operator} 
  $X_m$ in  $L^2(\mathbb R^n, d^nx)$ with domain
$$D(X_m) := \left\{\psi \in L^2(\mathbb R^n, d^nx) \:\left|\: \int_{\bR^n} |x_m\psi(x)|^2 d^kn\right. \right\}$$
and
\beq (X_m\psi)(x) := x_m\psi(x)\:,\quad x \in \bR^n\:. \label{Xm}\eeq
Just by applying the definition of adjoint one sees that $X_m^*=X_m$ so that 
$X_m$ is selfdjoint \cite{moretti}. Again applying the definition of adjoint, one sees 
that ${X'_m}^*={X_m''}^* = X_m^*$ \cite{moretti} where we know that  the last one is selfadjoint: $(X_m^*)^*= (X_m)^*= X_m^*$.
By definition, $X'_m$ and $X_m''$ are therefore essentially selfadjoint.
By (c) in remark \ref{remagg2} $X'_m$ and $X_m''$  admit a unique selfadjoint extension which must coincide with  $X_m$ itself.
We conclude that $ C_0^\infty(\mathbb R^n ; \bC)$ and ${\cal S}(\mathbb R^n)$ are {\em cores} (Def. \ref{defcore})
for the $m$-axis position operator.\\
{\bf (2)} For $m \in \{1,2,\ldots, n\}$,   the {\bf $k$-axis momentum  operator}, $P_m$, is obtained from the position operator using the Fourier-Plancherel unitary operator $\hat{\cal F}$ introduced in example \ref{FPop}.
$$D(P_m) := \left\{\psi \in L^2(\mathbb R^n, d^nx) \:\left|\: \int_{\bR^n} |k_m(\hat{\cal F}\psi)(k)|^2 d^nk\right. \right\}$$ and
\beq (P_m\psi)(x) := (\hat{\cal F}^* K_m \hat{\cal F}\psi)(x)\:,\quad x \in \bR^n\:.\label{Pm}\eeq
Above $K_m$ is the $m$-axis {\em position operator} just written for functions (in $L^2(\bR^n, d^nk)$) whose variable, for pure convenience, is denoted by $k$ instead of $x$.
Since $K_m$ is selfadjoint,  $P_m$  is selfadjoint as well, as established in exercise \ref{esinvarianceU} as a consequence of the fact that $\hat{\cal F}$ is unitary.\\
It is possible to give a more explicit form to $P_m$ if restricting its domain.
Taking $\psi \in C_0^\infty(\mathbb R^n ; \bC)\subset {\cal S}(\mathbb R^n)$ or directly  $\psi \in {\cal S}(\mathbb R^n)$, 
$\hat{\cal F}$ reduces to the standard integral Fourier transform (\ref{ft}) with inverse (\ref{fti}). Using these integral expressions we easily obtain
\beq (P_m\psi)(x) = (\hat{\cal F}^* K_m \hat{\cal F}\psi)(x) = - i \frac{\partial}{\partial x_m}\psi(x)\eeq
because in $ {\cal S}(\mathbb R^n)$, which is invariant under the Fourier (and inverse Fourier) integral transformation,
$$\int_{\bR^n} e^{ik\cdot x} k_m (\cF \psi)(k) d^nk = -i\frac{\partial}{\partial x_m} \int_{\bR^n} e^{-ik\cdot x} (\cF\psi)(k) d^nk\:.$$
This way leads us to consider the operators  $P'_m$ and $P''_m$ in  $L^2(\mathbb R^n, d^nx)$ with 
$$D(P'_m)  = C_0^\infty(\mathbb R^n ; \bC)\:, \quad D(P''_m)  = {\cal S}(\mathbb R^n)$$
and, for $x\in \bR^n$ and $\psi, \phi$ in the respective domains,
$$(P'_m\psi)(x) :=- i \frac{\partial}{\partial x_m}\psi(x)\:, \quad  (P''_m\phi)(x) :=- i \frac{\partial}{\partial x_m}\phi(x)\:.$$
Both operators are symmetric as one can easily prove by integrating by parts, but not selfadjoint. They admit selfadjoint extensions 
because they commute with the conjugation $(C\psi)(x) = \overline{\psi(-x)}$
(see remark 
\ref{remvon}). It is furthermore possible to prove that both operators are 
essentially self-adjoint by direct use of Proposition \ref{propindex} \cite{moretti}.  
However we already know that $P_m''$ is essentially selfadjoint as it coincides with 
the essentially selfadjoint operator $\hat{\cal F}^* K''_m \hat{\cal F}$ beacause ${\cal S}(\mathbb R^n)$ is invariant under  $\hat{\cal F}$.\\
The unique selfadjoint extension of both operators turns out to be $P_m$.
We conclude that $C_0^\infty(\mathbb R^n ; \bC)$ and ${\cal S}(\mathbb R^n)$ are {\em cores}
for the $m$-axis momentum operator. \\
With the given definitions of selfadjoint operators $X_k$ and $P_k$, ${\cal S}(\mathbb R^n)$ turns out to be an invariant domain and thereon the CCR (\ref{CCR}) hold rigorously.\\
As a final remark to conclude, we say that, if $n=1$, $D(P)$ coincides to the already introduced domain (\ref{DP}). In that domain $P$ is nothing but the weak derivative times the factor $-i$. \\
{\bf (3)} The most elementary example of application of Nelson's criterion is in $L^2([0,1],dx)$. Consider $A= -\frac{d^2}{dx^2}$ with dense  domain $D(A)$ given by the functions in $C^\infty([0,1]; \bC)$ such that $\psi(0)=\psi(1)$ and $\frac{d\psi}{dx}(0)=\frac{d\psi}{dx}(1)$.   $A$ is symmetric thereon as it arises immediately by integration by parts, in particular its domain is dense  since it includes the Hilbert basis of
exponentials $e^{i2\pi n x}$, $n\in \bZ$, which are eigenvectors of $A$. Thus $A$ is also essentially selfadjoint on the above  domain.\\
A more interesting case is the {\bf Hamiltonian operator of the harmonic oscillator}, $H$ \cite{ercolessi} obtained as follows. One starts by  $$H_0 = -\frac{1}{2m}\frac{d^2}{dx^2}+ \frac{m\omega^2}{2} x^2$$
 with $D(H_0) := {\cal S}(\bR)$. Above, $x^2$ is the multiplicative operator and  $m,\omega >0$ are constants. This operator is evidently symmetric on $D(H_0)$ and admits a Hilbert basis of the {\em Hermite functions}  $\psi_n(x)$ \cite{moretti} with corresponding eigenvalues $\omega(n + \frac{1}{2})$. So $H_0$ is essentially selfadjoint on $D(H_0)$ and thus $H:= \overline{H_0} = H_0^*$.
\hfill $\blacksquare$

\subsection{Spectrum of an operator}
Our goal is to extend (\ref{sectradec0}) to a formula valid in the infinite dimensional case. As we shall see shortly, passing to the infinite dimensional case,  the sum is replaced by an integral and $\sigma(A)$ must be enlarged with respect to the pure set of eigenvalues of $A$. This is because, as already noticed in the first section, there are operators which should be decomposed with the prescription (\ref{sectradec0})  but they do not have  eigenvalues, though they play a crucial r\^ole in QM.

\begin{notation} If $A:D(A) \to \cH$ is injective, $A^{-1}$ indicates its inverse when the co-domain of $A$ is restricted to $Ran(A)$. In other words,
$A^{-1} : Ran(A) \to D(A)$. \hfill $\blacksquare$
\end{notation}

\noindent The definition of {\em spectrum} of the operator $A: D(A) \to \cH$  extends  the notion  of set of eigenvalues. The eigenvalues of $A$ are the numbers $\lambda \in \bC$ such that
$(A-\lambda I)^{-1}$ does not exist. When passing to infinite dimensions, topological issues take place. As a matter of fact, even if $(A-\lambda I)^{-1}$ exists, it may be bounded or unbounded and its domain $Ran(A-\lambda I)$ may or may not be dense. These features permit us to define a suitable extension of the notion of set of eigenvalues.

\begin{definition} {\em Let  $A$ be an operator in the complex Hilbert space $\cH$. 
The {\bf resolvent set} of $A$ is the subset of $\bC$,
 $$\rho(A) := \{\lambda \in \bC \:|\:  (A-\lambda I) \mbox{ is injective, } \overline{Ran(A-\lambda I)}=\cH\:, (A- \lambda I)^{-1} \mbox{is bounded} \}$$
The {\bf spectrum} of $A$ is the complement $\sigma(A):=\mathbb{C}\setminus\rho(A)$ and it is given by the union of the following pairwise disjoint three parts:

(i) the  {\bf point-spectrum}, $\sigma_p(A)$,  where $A-\lambda I$ not injective ($\sigma_p(A)$ is the set of {\em eigenvalues} of $A$),

(ii) the  {\bf continuous spectrum},  $\sigma_c(A)$, where $A-\lambda I$  injective, $\overline{Ran(A-\lambda I)} = \cH$ and $(A-\lambda I)^{-1}$ not bounded, 

(iii) the {\bf residual spectrum}, $\sigma_r(A)$, where $A-\lambda I$ injective and $\overline{Ran(A-\lambda I)} \neq \cH$.}  \hfill $\blacksquare$
\end{definition}

\begin{remark}\label{remspectra}$\null$

{\bf (a)} It turns out that $\rho(A)$ is always {\em open}, so that $\sigma(A)$ is always {\em closed} \cite{R,moretti,S}.

 {\bf (b)} If $A$ is closed and normal, in particular,  if $A$  is either selfadjoint or unitary), $\sigma_r(A)= \emptyset$ (e.g., see \cite{moretti}).
Furthermore, if $A$ is closed (if $A \in \gB(\cH)$ in particular), $\lambda \in \rho(A)$
if and only if $A-\lambda I$ admits inverse in $\gB(\cH)$ (see (2) in exercise \ref{enr}).

{\bf (c)} If $A$ is selfadjoint, one finds  $\sigma(A)\subset \bR$ (see (1) in exercise \ref{enr}).

{\bf (d)} If $A$ is unitary  one finds $\sigma(A)\subset  \bT :=\{ e^{ia}\:|\: a \in \bR\}$ (e.g., see \cite{moretti}).

{\bf (e)} If $U: \cH \to \cH$ is unitary and  $A$ is any operator in the complex Hilbert space $\cH$, just by applying the definition one finds  $\sigma(UAU^*)= \sigma(A)$ and in particular,
\beq
\quad  \sigma_p(UAU^*)= \sigma_p(A)\:, \quad
\sigma_c(UAU^*)= \sigma_c(A)\:, \quad \sigma_r(UAU^*)= \sigma_r(A)\:. \label{unitinvariances}
\eeq 
The same result holds replacing $U:\cH \to \cH$ for $U :\cH \to \cH'$ and $U^*$ for $U^{-1}$, where $U$ is now a Hilbert space isomorphism (an isometric surjective linear map) and $\cH'$ another complex Hilbert space. \hfill $\blacksquare$
\end{remark} 

{\bf \exercise\label{enr}} $\null$\\
{\bf (1)} {\em Prove that if $A$ is a selfadjoint operator in the complex Hilbert space $\cH$ then

 (i) $\sigma(A)\subset \bR$,

 (ii) $\sigma_r(A)=\emptyset$, 

(iii) eigenvectors with different eigenvalues are orthogonal.}

{\bf Solution}.
Let us begin with (i). Suppose $\lambda =\mu + i\nu$, $\nu\neq 0$ and let us prove $\lambda \in \rho(A)$. If $x\in D(A)$, 
$$\langle (A -\lambda I)x | (A -\lambda I)x\rangle = \langle(A -\mu I)x | (A -\mu I)x \rangle +
\nu^2 \langle x|x\rangle + i\nu [\langle Ax|x\rangle  - \langle x|Ax\rangle]\:. $$
The last summand vanishes for $A$ is selfadjoint. Hence 
$$||(A-\lambda I)x|| \geq |\nu|\: ||x||\:.$$
With a similar argument we obtain
$$||(A-\overline{\lambda} I)x|| \geq |\nu|\: ||x||\:.$$
The operators $A-\lambda I$ and $A-\overline{\lambda}I$  are injective, and $||(A-\lambda I)^{-1}|| \leq |\nu|^{-1}$, where $(A-\lambda I)^{-1} : Ran(A-\lambda I) \to D(A)$.
Notice 
$$\overline{Ran(A-\lambda I)}^\perp = [Ran(A-\lambda I)]^\perp =
Ker(A^* - \overline{\lambda}I) =  Ker(A- \overline{\lambda}I) =\{0\}\:,$$
where the last equality makes use of the injectivity of $A-\overline{\lambda}I$. Summarising: 
$A -\lambda I$ in injective, $(A-\lambda I)^{-1}$ bounded and $\overline{Ran(A-\lambda I)}^\perp = \{0\}$,
i.e. $Ran(A-\lambda I)$ is dense in $\cH$;  therefore $\lambda \in\rho(A)$, by definition of resolvent set. Let us pass to (ii).  Suppose $\lambda\in\sigma(A)$, but $\lambda \not \in \sigma_p(A)$. Then $A-\lambda I$ must be one-to-one and $Ker(A-{\lambda} I) =\{0\}$. Since $A=A^*$ and  $\lambda\in \bR$ by (i), we have $Ker(A^*-\overline{\lambda} I) =\{0\}$, so  $[Ran(A -\lambda I)]^\perp = Ker(A^*-\overline{\lambda} I)= \{0\}$ and $\overline{Ran(A -{\lambda} I)} =\cH$. 
Consequently $\lambda \in \sigma_c(A)$.
Proving (iii) is easy: if $\lambda\neq \mu$ and $Au= \lambda u$, $Av=\mu v$, then 
$$(\lambda-\mu) \langle u|v\rangle= \langle Au|v\rangle-\langle u|Av\rangle = \langle u|Av\rangle- \langle u|Av\rangle  =0\:;$$
from $\lambda,\mu \in \bR$ and $A=A^*$. But $\lambda-\mu\neq 0$, so $\langle u|v\rangle=0$. \hfill $\Box$\\

\noindent {\bf (2)} {\em Let $A :D(A) \to \cH$ be a closed operator in $\cH$ (in particular $A \in \gB(\cH)$). Prove that $\lambda \in \rho(A)$ if and only if $A-\lambda I$ admits an inverse which belongs to $\gB(\cH)$.}

{\bf Solution}. If $(A-\lambda I)^{-1} \in \gB(\cH)$, it must be $\overline{Ran(A-\lambda I)}= Ran(A-\lambda I) 
= \cH$
and $(A-\lambda I)^{-1}$ is bounded, so that $\lambda \in \rho(A)$ by definition. Let us prove the converse. Suppose that $\lambda \in \rho(A)$. 
We know that $(A-\lambda I)^{-1}$ is defined on the dense domain $Ran(A-\lambda I)$ and is bounded.
To conclude, it is therefore enough proving that $y \in \cH$ implies $y \in Ran(A-\lambda I)$. To this end, notice that
if $y \in \cH = \overline{Ran(A-\lambda I)}$, then $y = \lim_{n\to +\infty}(A-\lambda I)x_n$ for some $x_n \in D(A-\lambda I)$. The sequence of $x_n$ converges. Indeed   $\cH$ is complete and $\{x_n\}_{n\in \bN}$ is Cauchy as (1)
$x_n = (A-\lambda I)^{-1}y_n$, (2)
$||x_n-x_m|| \leq ||(A-\lambda I)^{-1}||\: ||y_n-y_m||$, and (3) $y_n \to y$. To end the proof, we observe that,  $A-\lambda I$ is closed since $A$ is such ((b) in remark \ref{remarkclosure}).  It must consequently be ((c) in remark \ref{remarkclosure}) $x= \lim_{n\to +\infty}x_n \in D(A-\lambda I)$ and $y= (A-\lambda I)x \in Ran(A-\lambda I)$.  \hfill $\Box$

\example\label{specXP} The $m$-axis position operator $X_m$ in $L^2(\bR^n, d^nx)$ introduced in (1) of example \ref{esXP} satisfies
\beq
\sigma(X_m)= \sigma_c(X_m) = \bR \label{sX}\:.
\eeq
The proof can be obtained as follows. First  observe that $\sigma (X_m) \subset \bR$ since the operator is selfadjoint. However
$\sigma_p(X_m) = \emptyset$ as observed in the first section and $\sigma_r(X_m)= \emptyset$
because $X_m$ is self-adjoint ((1) in exercise \ref{enr}).  Suppose that, for some $r\in \bR$,
$(X_m-rI)^{-1}$ is bounded. If $\psi \in D(X_m-rI)= D(X_m)$ with $||\psi||=1$ we have
$||\psi|| = ||(X_m-rI)^{-1} (X_m-rI)\psi||$ and thus
$||\psi|| \leq ||(X_m-rI)^{-1}||\: || (X_m-rI)\psi||$. Therefore
$$||(X_m-rI)^{-1}|| \geq \frac{1}{||(X_m-rI)\psi||}$$
 For every fixed $\epsilon>0$, it is simply constructed $\psi \in D(X_m)$ with $||\psi||=1$  and $||(X_m-rI)\psi|| < \epsilon$. Therefore $(X_m-rI)^{-1}$ cannot be bounded and thus $r \in \sigma_c(X_m)$.
In view of (e) in remark \ref{remspectra}, we also conclude that
\beq
\sigma(P_m)= \sigma_c(P_m) = \bR \label{sP}\:,
\eeq
just because the momentum operator $P_m$ is related to the position one by means of a unitary operator given by the Fourier-Plancherel operator $\hat{\cal F}$ as discussed in (2) of example \ref{esXP}. \hfill $\blacksquare$
\subsection{Spectral measures} Let us pass the the notion of  {\em orthogonal projector} which will be later exploited to state the spectral decomposition theorem.
\notation If $M\subset \cH$, $M^\perp := \{ y \in \cH\:|\: \langle y|x\rangle =0 \quad \forall x \in M \}$ denotes the  {\bf orthogonal} of $M$.
 \hfill $\blacksquare$\\

\noindent Evidently $M^\perp$ is a closed subspace of $\cH$. $^\perp$ enjoys several nice properties (e.g. see \cite{R,moretti}), in particular, \beq \overline{\mbox{span}(M)} = (M^\perp)^\perp\quad \mbox{  and  }\quad \cH = \overline{\mbox{span}(M)} \oplus M^\perp\label{propperp}\eeq
where the bar denotes the topological closure and  $\oplus$  the direct orthogonal sum. From the definition of adjoint, one easily has for $A: D(A) \to \cH$ densely defined,
$$Ker(A^*-\overline{\lambda} I)= [Ran(A-\lambda I)]^\perp \quad \mbox{and}\quad 
Ker(A -\lambda I)\subset  [Ran(A^*-\overline{\lambda} I)]^\perp \quad \forall \lambda \in \bC$$
where the inclusion becomes an identity if $A \in \gB(\cH)$.

\begin{definition} {\em Let $\cH$ be a complex Hilbert space.
 $P \in \gB(\cH)$ is called {\bf orthogonal projector} when $PP=P$ and $P^*=P$.  ${\cal L}(\cH)$ denotes the set  of orthogonal projectors of $\cH$.}  \hfill $\blacksquare$
\end{definition}
\noindent We have the well known relation between orthogonal projectors and closed subspaces \cite{R,moretti}
\begin{proposition}\label{propproj} {If $P \in {\cal L}(\cH)$, then $P(\cH)$ is a closed subspace. If $\cH_0 \subset \cH$ is a closed subspace, there exists exactly one  $P \in  {\cal L}(\cH)$
such that $P(\cH) = \cH_0$. Finally, $I-P \in  {\cal L}(\cH)$ and it projects onto $\cH_0^\perp$ (e.g., see \cite{moretti}). } 
\end{proposition} 
\noindent We can now state one of the most important definitions in spectral theory.
\begin{definition}\label{defPVM} {\em Let $\cH$ be a complex Hilbert space and $\Sigma(X)$ a $\sigma$-algebra over $X$. A {\bf projector-valued measure (PVM)} on $X$, $P$, is a map $\Sigma(X) \ni E \mapsto P_E \in {\cal L}(\cH)$ sucht that

 (i)  $P_X=I$, 

(ii) $P_EP_F = P_{E\cap F}$,

 (iii) If $N \subset \bN$  and $\{E_k\}_{k\in N} \subset \Sigma(X)$
satisfies  $E_j \cap E_k = \emptyset$ for $k\neq j$, then
$$\sum_{j \in N} P_{E_j}x= P_{\cup_{j\in N}E_j}x \quad \mbox{for every $x\in \cH$.}$$
(If $N$ is infinite,  the sum  on the left hand side of (iii) is computed referring  to the topology of $\cH$)} \hfill $\blacksquare$
\end{definition}

\begin{remark}\label{remD}$\null$

{\bf (a)}  (i) and  (iii) with $N=\{1,2\}$ imply that $P_\emptyset =0$ using $E_1=X$ and $E_2 = \emptyset$.
Next (ii) entails that $P_EP_F=0$ if $E\cap F=\emptyset$. An important consequence is that 
for  $N$ infinite,  the vector given by the sum  on the left hand side of (iii)  is independent  from the chosen order because that vector is a sum  of pairwise orthogonal vectors $P_{E_j}x$.

{\bf (b)} If $x,y \in \cH$, 
$\Sigma(X) \ni E \mapsto \langle  x|P_Ey\rangle =: \mu^{(P)}_{xy}(E)$ is a {\em complex  measure}   whose (finite) {\em total variation} \cite{R} will be denoted by $|\mu^{(P)}_{xy}|$.  From the definition of $\mu_{xy}$, we immediately have:
 
(i) $\mu^{(P)}_{xy}(X)= \langle x| y \rangle$,

(ii) $\mu^{(P)}_{xx}$ is always positive and finite and $\mu^{(P)}_{xx}(X)= ||x||^2$;

(iii) if $s = \sum_{k=1}^n s_k \chi_{E_k}$ is a {\em simple function}  \cite{R}, $\int_X s d\mu_{xy} = \langle x| \sum_{k=1}^n s_k P_{E_k} y\rangle$. \hfill $\blacksquare$
\end{remark}

\example \label{exPVM}$\null$\\
{\bf (1)} The simplest example of PVM is related to a countable Hilbertian basis $N$ in a separable Hilbert space $\cH$.  We can define $\Sigma(N)$ as the class of all subsets of $N$ itself. Next, for $E\in \Sigma(N)$ and $z\in \cH$ we define
$$P_{E}z := \sum_{x\in E} \langle x| z\rangle x$$
and $P_{\emptyset}:=0$. It is easy to prove that the class of all $P_E$ defined this way form a PVM on $N$. (This definition can be also given if $\cH$ is non-separable and $N$ is uncountable, since for every $y\in \cH$ only an at most countable subset of elements $x\in E$ satisfy $\langle x|y\rangle \neq 0$). In particular $\mu_{xy}(E) = \langle x|P_Ey \rangle = \sum_{z\in E} \langle x|z\rangle \langle z | y \rangle$
and
$\mu_{xx}(E) = \sum_{z\in E}|\langle x| z \rangle |^2$.\\
{\bf (2)} A more complicated version of (1) consists of a PVM  constructed out 
of a orthogonal Hilbertian decomposition of a separable Hilbert space, 
$\cH = \oplus_{n \in \bN} \cH_n$, where $\cH_n\subset \cH$ is a closed subspace and $\cH_n \perp \cH_m$ if $n\neq m$. Again defining $\Sigma(\bN)$ as the set of subsets of $\bN$, for $E\in \Sigma(N)$ and $z\in \cH$ we define
$$P_{E}z := \sum_{x\in E} Q_nz$$
where $Q_n$ is the orthogonal projector onto $\cH_n$ (the reader can easily check that the sum always converges using Bessel's inequality).
It is easy to prove that the class of $P_E$s defined this way form a PVM on $\bN$.
In particular $\mu_{xy}(E) = \langle x|P_Ey \rangle = \sum_{n\in E} \langle x|Q_n y \rangle$
and
$\mu_{xx}(E) = \sum_{n\in E}||Q_n x||^2$.\\
{\bf (3)} In $L^2(\mathbb R, dx)$ a simple PVM, not related with a Hilbertian basis, is made as follows. To every $E \in {\cal B}(\bR)$, the Borel $\sigma$-algebra, associate
the orthonormal projector $P_E$ such that, if $\chi_E$ is the {\bf characteristic function  of $E$} -- $\chi_E(x)=0$ if $x\not \in E$ and $\chi_E(x)=1$ if $x\in E$ --
$$(P_E\psi)(x) := \chi_E(x)\psi(x)\quad \forall \psi \in L^2(\bR, dx)\:.$$
Moreover $P_{\emptyset}:=0$. It is easy to prove that the collection of the $P_E$ is a PVM.
In particular $\mu_{fg}(E) = \langle f|P_Eg \rangle = \int_{E}\overline{f(x)}g(x) dx$
and
$\mu_{ff}(E) = \int_E |f(x)|^2 dx$. \hfill $\blacksquare$\\

\noindent We have the following fundamental result \cite{R,moretti,S}.
\begin{proposition}\label{propint}
Let $\cH$ be a complex Hilbert space  and $P: \Sigma(X) \to {\cal L}(\cH)$ a PVM. 
If $f: X \to \bC$ is measurable, define 
$$\Delta_f := \left\{x \in \cH \:\left|\:  \int_{X} |f(\lambda)|^2 \mu^{(P)}_{xx}(\lambda)< +\infty \right.\right\}\:.$$
$\Delta_f$ is a dense subspace of $\cH$ and there is a unique operator  
\beq \int_X f(\lambda) dP(\lambda) : \Delta_f \to \cH \label{intop}\eeq
such that 
\beq \left\langle x  \left| \int_X f(\lambda) dP(\lambda) y \right.\right\rangle = 
\int_{X} f(\lambda) \mu^{(P)}_{xy}(\lambda)\quad \forall x \in \cH \:, \forall y \in \Delta_f \label{intop2}\eeq
The operator in (\ref{intop}) turns out to be closed and normal. It finally satisfies
\beq 
\left(\int_X  f(\lambda)\:  d P(\lambda)\right)^* = \int_X  \overline{f(\lambda)}\:  d P(\lambda) \label{aggf}
\eeq
and
\beq \left|\left| \int_X  f(\lambda)\:  d P(\lambda) x\right|\right|^2 = \int_X |f(\lambda)|^2 d \mu^{(P)}_{xx}(\lambda) \quad \forall x \in \Delta_f \:.\label{strongbound}\eeq
\end{proposition}

\noindent {\bf Idea of the existence part of  the Proof}.  The idea of the proof of existence of the operator in (\ref{intop}) relies upon the validity of the inequality  ((1) in exercises \ref{exPVM2} below)
\beq  \int_X |f(\lambda)|\:  d |\mu^{(P)}_{xy}|(\lambda) \leq ||x|| \sqrt{\int_X |f(\lambda)|^2 d \mu^{(P)}_{yy}(\lambda)}\qquad  \forall y \in \Delta_f \:,  \forall x \in \cH\:. \label{ieq}\eeq 
This inequality also proves that $f \in L^2(X, d\mu^{(P)}_{yy})$ implies $f \in L^1(X, d|\mu^{(P)}_{xy}|)$ for $x\in \cH$, so that (\ref{intop2}) makes sense.
Since from the general measure theory
$$ \left|\int_X f(\lambda)\:  d \mu^{(P)}_{xy}(\lambda)\right| \leq   \int_X |f(\lambda)|\:  d |\mu^{(P)}_{xy}|(\lambda)\:,$$ 
(\ref{ieq}) implies that $\cH \ni x \mapsto \int_X f(\lambda)\:  d \mu^{(P)}_{xy}(\lambda)$ is continuous at  $x=0$. This map is also anti-linear as follows from the definition of $\mu_{x,y}$. An elementary use of Riesz' lemma proves that 
there exists a vector, indicated by  $\int_X f(\lambda) dP(\lambda) y$, satisfying (\ref{intop2}). That is the action of an operator on a vector $y\in \Delta_f$ because 
$\Delta_f\ni y \mapsto \int_X f(\lambda)\:  d \mu^{(P)}_{xy}(\lambda)$ is linear. \hfill $\Box$

\begin{remark}\label{rembound}
Identity (\ref{strongbound}) gives $\Delta_f$ a direct meaning  in terms of boundedness  of $\int_X  f(\lambda)\:  d P(\lambda)$. Since $\mu_{xx}(X)= ||x||^2 <+\infty$, (\ref{strongbound}) together with the definition of $\Delta_f$
immediately implies that: if  $f$ is bounded or, more weakly {\em $P$-essentially bounded}\footnote{As usual, $||f||_\infty^{(P)}$ is the infimum of  positive reals $r$ such that $P(\{x \in X \:|\: |f(x)| > r\})=0$.} on $X$, then  $$\int_X  f(\lambda)\:  d P(\lambda) \in \gB(\cH)$$
and $$\left|\left|\int_X  f(\lambda)\:  d P(\lambda)\right|\right| \leq ||f||^{(P)}_\infty\leq ||f||_\infty\:.$$ The $P$-essentially boundedness
is also a {\em necessary} (not only sufficient) condition for $\int_X  f(\lambda)\:  d P(\lambda) \in \gB(\cH)$ \cite{R,moretti}.  \hfill $\blacksquare$
 \end{remark}

{\bf \exercise \label{exPVM2}} $\null$\\
{\bf (1)} {\em Prove inequality (\ref{ieq})}.

{\bf Solution}. Let $x\in \cH$ and $y \in \Delta_f$. If $s: X \to \bC$ is a {\em simple function} and $h: X \to \bC$ is the {\em Radon-Nikodym derivative} of $\mu_{xy}$ with respect to $|\mu_{xy}|$ so that $|h(x)|=1$ and  $\mu_{xy}(E)= \int_E h d|\mu_{xy}|$ (see, e.g., \cite{moretti}), we have for an increasing sequence of simple functions $z_n \to h$
pointwise, with $|z_n|\leq |h^{-1}|=1$, due to the dominate convergence theorem,
$$\int_X |s| d|\mu_{xy}| = \int_X |s| h^{-1}d\mu_{xy} = \lim_{n \to +\infty}  \int_X|s| z_n d\mu_{xy} = \lim_{n \to +\infty} \left\langle x\left| \sum_{k=1}^{N_n} z_{n,k} P_{E_{n,k}}\right. y \right\rangle\:.$$
In the last step we have made use of (iii)(b) in remark \ref{remD} for the simple function
$|s|z_n = \sum_{k=1}^{N_n} z_{n,k} \chi_{E_{n,k}}$. Cauchy Schwartz inequality immediately yields
$$\int_X |s| d|\mu_{xy}| \leq ||x|| \lim_{n \to +\infty}\left|\left|\sum_{k=1}^{N_n} z_{n,k} P_{E_{n,k}} y\right|\right| =  ||x|| \lim_{n \to +\infty}\sqrt{\int_X |s z_n|^2 d\mu_{yy}} \:,$$
where we have used $P_{E_{n,k}}^*P_{E_{n,k'}} = P_{E_{n,k}}P_{E_{n,k'}}= \delta_{kk'}P_{E_{n,k}}$ since $E_{n,k} \cap E_{n,k'}= \emptyset$ for $k\neq k'$.
Next observe that, as $|s z_n|^2  \to |sh^{-1}|^2=|s|^2$, dominate convergence theorem leads to $$\int_X |s| d|\mu_{xy}| \leq ||x||\sqrt{\int_X |s|^2 d\mu_{yy}}\:.$$ Finally, replace $s$ above for  a sequence of simple functions $|s_n| \to f \in L^2(X, d\mu_{yy})$ pointwise, with $s_{n}\leq |{s_{n+1}}| \leq |f|$. Monotone convergence theorem and dominate convergence theorem, respectively applied to the left and right-hand side of the found inequality, produce inequality (\ref{ieq}).

{\bf (2)} {\em Prove that, with the hypotheses of Proposition \ref{propint}, it holds
\beq 
\int_X  \chi_E(\lambda)\:  d P(\lambda) = P_E  \:, \quad \mbox{if $E \in \Sigma(X)$}\label{chi}
\eeq
and in particular
\beq 
\int_X 1 \:  d P(\lambda) = I\:.
\eeq}

{\bf Solution}. It is sufficient to prove (\ref{chi}) since we know that $P_X=I$.
To this end, notice that, by direct inspection
$$\left\langle x  \left| P_E y \right.\right\rangle = 
\int_{X} \chi_E(\lambda) \mu^{(P)}_{xy}(\lambda)\quad \forall x \in \cH \:, \forall y \in \Delta_{\chi_E}= \cH\:.$$
By the uniqueness property stated in Proposition \ref{propint} (\ref{chi}) holds. \hfill $\Box$\\

\noindent {\bf (3)} {\em Prove that if $P$ a PVM on $\cH$ and $T$ is an operator in $\cH$ with $D(T)= \Delta_f$   such that
\beq  \left\langle x  \left|T x \right.\right\rangle = 
\int_{X} f(\lambda) \mu^{(P)}_{xx}(\lambda)\quad \forall x \in \Delta_f  \label{xTx}\eeq
then $$T= \int_{X} f(\lambda) dP(\lambda)\:.$$}

{\bf Solution}.  From the definition of $\mu_{xy}$ we easily have (everywhere  omitting $^{(P)}$ for semplicity)
$$4\mu_{xy}(E) = \mu_{x+y,x+y}(E) -  \mu_{x-y,x-y}(E)  -i \mu_{x+iy,x+iy}(E)  +i \mu_{x-iy,x-iy}(E) $$
This identity implies that, if $x,y \in \Delta_f$, 
$$ 4\int_X f d\mu_{xy} = \int_X f d\mu_{x+y,x+y} -\int_X f d\mu_{x-y,x-y} -i \int_X f d\mu_{x+iy.x+iy} +i\int_X f d\mu_{x-iy,x-iy}$$
Similarly, from the elementary properties of the scalar product, when $x,y \in D(T)$
$$4\langle x|T y\rangle  =
\langle x+y|T(x+y)\rangle  -  \langle x-y|T(x-y)\rangle   -i \langle x+iy|T(x+iy)\rangle   +i \langle x-iy|T(x-iy)\rangle\:.$$
It is then obvious that (\ref{xTx}) implies  
$$  \left\langle x  \left|T y \right.\right\rangle = 
\int_{X} f(\lambda) \mu^{(P)}_{xy}(\lambda)\quad \forall x, y \in \Delta_f\:, $$
so that
$$\left\langle x \left| \left( T-  \int_{X} f(\lambda) dP(\lambda)\right)\right. y\right\rangle = 0 \quad \forall x, y \in \Delta_f$$
Since $x$ varies in a dense set $\Delta_f$, $Ty- \int_{X} f(\lambda) dP(\lambda)y=0$ for every $y \in \Delta_f$
which is the thesis. \hfill $\Box$

\example $\null$\\
{\bf (1)}  Referring to the PVM in (2) of example \ref{exPVM}, directly from the definition of $\int_X f(\lambda) dP(\lambda)$ or exploiting (3) in exercises \ref{exPVM2} we have that
$$\int_\bN f(\lambda) dP(\lambda)z = \sum_{n\in \bN}  f(n) Q_nz $$
for every $f: \bN \to \bC$ (which is necessarily measurable with our definition of $\Sigma(\bN)$). Correspondingly, the domain of $\int_N f(\lambda) dP(\lambda)$ results to be
$$\Delta_f := \left\{ z \in \cH \:\left|\: \sum_{n\in \bN} |f(n)|^2 ||Q_nz||^2  < +\infty \right.\right\}$$
We stress that we have found a direct generalization of the expansion (\ref{sectradec0}) if the operator $A$ is now hopefully written as
$$Az = \sum_{n\in \bN} n Q_nz\:. $$
We shall see below that it is the case.\\
{\bf (3)}   Referring to the PVM in (3) of example \ref{exPVM}, 
directly from the definition of $\int_X f(\lambda) dP(\lambda)$
or exploiting (3) in exercises \ref{exPVM2}
 we have that
$$\left(\int_\bR f(\lambda) dP(\lambda) \psi\right)(x)  = f(x) \psi(x)\:, \quad x \in \bR$$
 Correspondingly, the domain of $\int_\bR f(\lambda) dP(\lambda)$ results to be
$$\Delta_f := \left\{ \psi \in L^2(\bR, dx) \:\left|\: \int_\bR |f(x)|^2 |\psi(x)|^2 dx  < +\infty \right.\right\}$$ \hfill $\blacksquare$

\subsection{Spectral Decomposition and Representation Theorems} We are in a position to state  the fundamental result of  the spectral theory of selfadjoint operators, which extend the expansion (\ref{sectradec0})  to an integral formula valid also in the infinite dimensional case, and where the set of eigenvalues is replaced by the full spectrum of the selfadjoint operator.\\
To state the theorem, we preventively notice that (\ref{aggf}) implies that $\int f(\lambda) dP(\lambda)$ is 
selfadjoint 
if $f$ is real: The idea of the theorem is to prove that every selfadjoint operator can be written this way for a specific $f$ and with respect to a PVM on $\bR$ associated with the operator itself. 

\notation From now on ${\cal B}(T)$ denotes the Borel $\sigma$-algebra on the topological space $T$.  \hfill $\blacksquare$

\begin{theorem}[Spectral Decomposition Theorem for Selfadjoint Operators]\label{st}
Let $A$ be a selfadjoint operator  in the complex  Hilbert space $\cH$.

{\bf (a)}  There is a unique PVM, 
$P^{(A)} : {\cal B}(\bR) \to {\cal L}(\cH)$, such that
$$A = \int_{\bR} \lambda dP^{(A)}(\lambda)\:.$$
In particular $D(A)= \Delta_{id}$, where $id : \bR \ni \lambda \mapsto \lambda$.

{\bf (b)} Defining  the {\bf support} of $P^{(A)}$,  $supp(P^{(A)})$, as the complement in $\bR$ of  the union of all open sets  $O \subset \bC$ with $P_O^{(A)}=0$ it results 
 $$supp(P^{(A)}) = \sigma(A)$$
so that
\beq P^{(A)}(E) = P^{(A)}(E \cap \sigma(A))\:, \quad \forall E \in  {\cal B}(\bR) \label{Prest}\:.\eeq

{\bf (c)} $\lambda \in \sigma_p(A)$ if and only if $P^{(A)}(\{\lambda\}) \neq 0$, this happens in particular if $\lambda$ is an isolated point of $\sigma(A)$.

{\bf (d)}  $\lambda \in \sigma_c(A)$ if and only if $P^{(A)}(\{\lambda\}) = 0$ but $P^{(A)}(E)\neq 0$ if
$E \ni \lambda$ is an open set of $\bR$. 
\end{theorem}

\noindent The proof  can be found, e.g., in \cite{R,moretti,S}. 
\begin{remark}  Theorem \ref{st}  is a particular case of a more general theorem (see \cite{R,moretti} and especially \cite{S}) valid when $A$ is a (densely defined  closed) normal operator.  The general statement is identical, it is sufficient to replace everywhere $\bR$ for 
$\bC$. A particular case  is the one of  $A$ unitary. In this case the statement can be rephrased replacing everywhere $\bR$ for $\bT$ since it includes the spectrum of $A$ in this case ((d) remark \ref{remspectra}).  \hfill $\blacksquare$
\end{remark}

\notation In view of the said theorem, and (b) in particular, if $f: \sigma(A)  \to \bC$ is measurable (with respect to the $\sigma$-algebra 
obtained by restricting ${\cal B}(\bR)$ to $\sigma(A)$), we use the notation
\beq
f(A)  : =  \int_{\sigma(A)}  f(\lambda) dP^{(A)}(\lambda) :=  \int_\bR  g(\lambda) dP^{(A)}(\lambda) =: g(A) \:.\label{fA}
\eeq
where $g : \bR \to \bC$ is the extension of $f$ to the zero function outside $\sigma(A)$ or any other measurable function which coincides with $f$ on $supp(P^{(A)})= \sigma(A)$. Obviously $g(A)=g'(A)$ if $g,g' : \bR \to \bC$ coincide in $supp(P^{(A)})= \sigma(A)$. \hfill $\blacksquare$

{\bf \exercise\label{espos}} {\em Prove that if $A$ is a selfdjoint operator in the complex Hilbert space $\cH$, it holds $A \geq 0$ -- that is $\langle x|Ax \rangle \geq 0$ for every $x \in D(A)$ -- if and only if $\sigma(A)\subset [0, +\infty)$.}

{\bf Solution}. Suppose that $\sigma(A)\subset [0, +\infty)$. If $x \in D(A)$ we have $\langle x| A x \rangle = \int_{\sigma(A)} \lambda d \mu_{x,x} \geq 0$ in view of  (\ref{intop2}), the spectral decomposition theorem, since $\mu_{x,x}$ is a positive measure ad $\sigma(A)\in [0, +\infty)$. To conclude, we  prove that $A\geq 0$ is false if $\sigma(A)$ includes negative elements.
To this end  assume that, conversely,  $\sigma(A) \ni \lambda_0 <0$. Using (c) and (d) of Theorem \ref{st}, one finds an interval $[a,b] \subset \sigma(A)$ with $[a,b] \subset (-\infty, 0)$ and $P^{(A)}_{[a,b]} \neq 0$ (possibly $a=b=\lambda_0$). If $x \in P^{(A)}_{[a,b]} (\cH)$ with $x\neq 0$, it holds $\mu_{xx}(E) = \langle x|P_E x\rangle =
\langle x|P^*_{[a,b]}P_E x P_{[a,b]}\rangle = \langle x|P_{[a,b]}P_E P_{[a,b]}x \rangle =
\langle x|P_{[a,b]\cap E}x\rangle =0$ if $[a,b]\cap E= \emptyset$. Therefore,
$\langle x| A x \rangle = \int_{\sigma(A)} \lambda d \mu_{x,x} = \int_{[a,b]}  \lambda d \mu_{x,x} \leq   \int_{[a,b]}  b \mu_{x,x} < b ||x||^2 <0$. \hfill $\blacksquare$

\example  $\null$\\
{\bf (1)} Let us focus on the $m$-axis position operator $X_m$ in $L^2(\bR^n, d^nx)$ introduced in (1) of example \ref{esXP}. 
We know that $\sigma(X_m)= \sigma_c(X_m)= \bR$ from example \ref{specXP}. 
We are interested in the PVM $P^{(X_m)}$ of $X_m$ defined on $\bR= \sigma(X_m)$.
Let us fix $m=1$
the other cases are analogous.
The PVM associated to $X_1$ is
\beq (P^{(X_1)}_E \psi)(x) = \chi_{E \times \bR^{n-1}}(x) \psi(x) \quad \psi \in L^2(\bR^n , d^nx)\:,\label{cand}\eeq
where $E \in {\cal B}(\bR)$ is here identified with a subset of the first factor of 
$\bR \times \bR^{n-1} = \bR^n$. Indeed, indicating by $P$ the right-hand side of (\ref{cand}), one easily verifies that $\Delta_{x_1} = D(X_1)$ and\footnote{More generally 
$\int_{\bR}\int_{\bR^{n-1}} g(x_1) |\psi(x)|^2dx d^{n-1}x= \int_{\bR} g(x_1) d\mu^{(P)}_{\psi,\psi}(x_1)$ is evidently valid for simple functions and then it extends to generic measurable functions when both sides make sense 
in view of, for instance, Lebesgue's dominate convergence theorem for positive measures.}
$$\langle \psi| X_1 \psi \rangle = \int_{\bR} \lambda \mu^{(P)}_{\psi,\psi}(\lambda)\quad \forall \psi \in D(X_1)= \Delta_{x_1}$$
where 
$\mu^{(P)}_{\psi,\psi}(E) = \langle \psi| P_E \psi \rangle = \int_{E \times \bR^{n-1}}
 |\psi(x)|^2 d^nx$. (2) in exercise \ref{exPVM} proves that $X_1 = \int_\bR \lambda dP(\lambda)$ and thus (\ref{cand}) holds true.\\
{\bf (2)} Considering  the $m$-axis momentum  operator $P_m$ in $L^2(\bR^n, d^nx)$ introduced in (2) of example \ref{esXP}, taking (\ref{Pm}) into account where $\hat{\cal F}$ (and thus $\hat{\cal F}^*$) is unitary, in view of (i) in Proposition \ref{propint2} we immediately have that the PVM of $P_m$ is 
$$Q^{(P_m)}_E := \hat{\cal F}^* P^{(K_m)}_E \hat{\cal F}\:.$$
Above $K_m$ is the operator $X_m$ represented in $L^2(\bR^n, d^nk)$ as in  (1) of example \ref{esXP}. \\
{\bf (3)} More complicated cases exist. Considering an operator
 of the form
$$H := \frac{1}{2m}P^2 + U$$
where $P$ is the momentum operator in $L^2(\bR, dx)$, $m>0$ is a constant and $U$ is a real valued function on $\bR$ used as multiplicative operator. 
If $U =U_1+U_2$ with $U_1 \in L^2(\bR, dx)$
and $U_2 \in L^\infty(\bR, dx)$ real valued, and  $D(H) = C^\infty(\bR; \bC)$,  $H$ turns out to be (trivially) symmetric but also essentially selfadjoint \cite{moretti} as a consequence of a well known result ({\em Kato-Rellich's theorem}).  The unique selfadjoint extension $\overline{H}=(H^*)^*$ of $H$ physically represent the Hamiltonian operator of a quantum particle living along $\bR$ with a potential energy described by $U$. In this case, generally speaking, $\sigma(\overline{H})$ has both point and continuous part. $\int_{\sigma_p(\overline{H})} \lambda dP^{(\overline{H})}(\lambda)$ has a form 
like this
$$\int_{\sigma_p(\overline{H})} \lambda dP^{(\overline{H})}(\lambda) = \sum_{\lambda \in \sigma_p(\overline{H})} \lambda P_{\lambda}$$
where $P_\lambda$ is the orthogonal projector onto the eigenspace of $\overline{H}$ with eigenvalue $\lambda$. Conversely, $\int_{\sigma_c(\overline{H})} \lambda dP^{(\overline{H})}(\lambda)$ has an expression much more complicated and, under a unitary transform is similar to the integral decomposition of $X$.
$\null$ \hfill $\blacksquare$

\begin{remark}\label{remst}$\null$

{\bf (a)}  It is worth stressing that the notion (\ref{fA}) of a function  of a selfadjoint operator is just an extension of  in the analogous notion introduced for the  finite dimensional case (\ref{sectradec1}) and thus may be used in QM applications.\\
 It is possible to prove that if $f: \sigma(A) \to \bR$ is continuous, then
\beq
\sigma(f(A)) = \overline{f(\sigma(A))} \label{spectf}
\eeq
where the bar denotes the closure and, if $f:  \sigma(A) \to \bR$ is measurable, 
\beq
\sigma_p(f(A)) \supset f(\sigma_p(A)) \:.
\eeq
More precise statements based on the notion of {\em essential range} can be found in \cite{moretti}. It turns out that, for $A$ selfadjoint and $f: \sigma(A) \to \bC$ measurable,  $z\in \sigma(f(A))$ if and only if  $P^{(A)}(E_z)\neq 0$ for some open set $E_z\ni z$. Now $z \in  \sigma(f(A))$ is in $\sigma_p(f(A))$ iff $P^{(A)}(f^{-1}(z))\neq 0$
or  it  is in $\sigma_c(f(A))$ iff $P^{(A)}(f^{-1}(z))= 0$.

{\bf (b)}  It is fundamental to stress that in, QM,  (\ref{spectf}) permits us to adopt the standard operational approach on observables $f(A)$ as  the observable whose set of possible  values is (the closure of)  the set of reals  $f(a)$ where $a$ is a possible value of $A$.

{\bf (c)} The following important  fact holds
\begin{proposition}\label{propboundob}  A selfadjoint operator is bounded (and its domain coincide to the whole $\cH$)  if and only if $\sigma(A)$ is bounded.
\end{proposition}

\begin{proof} It essentially  follows from  (\ref{strongbound}) restricting the integration space to $X= \sigma(A)$.
In fact, if $\sigma(A)$ is bounded and thus compact it being closed,  the continuous function $id: \sigma(A) \ni \lambda \to \lambda$ is bounded and (\ref{strongbound}) implies that $A= \int_{\sigma(A)} id dP^{(A)}$ is bounded and the inequality holds
\beq
||A|| \leq \sup\{|\lambda| \:|\: \lambda \in \sigma(A)\}\:.\label{spectralradius}
\eeq 
 In this case it also hold $D(A)=\Delta_{id} = \cH$.\\
 If, conversely,  $\sigma(A)$ is not bounded, we can find a sequence $\lambda_n\in \sigma(A)$ with $|\lambda_n| \to \infty$ as $n\to +\infty$. With the help of (c) and (d)
in Theorem \ref{st}, it is easy to construct vectors  $x_n$ with $||x_n|| \neq 0$ and  $x_n \in P^{(A)}_{B(\lambda_n)}(\cH)$ where 
$B(\lambda_n):= [\lambda_n-1, \lambda_n+1]$. (\ref{strongbound}) implies
$$||A x_n ||^2 \geq ||x_n||^2 \inf_{z \in B(\lambda_n)} |id(z)|^2$$
Since $\inf_{z \in B(\lambda_n)} |id(z)|^2 \to +\infty$, we have that  $||Ax_n||/||x_n||$ is not bounded and $A$, 
in turn, cannot be  bounded. In this case, since $A=A^*$,  Theorem \ref{HT} entails that  $D(A)$ is {\em strictly} included in $\cH$. 
\end{proof}
\noindent {\em It is possible to prove \cite{moretti} that (\ref{spectralradius}) can be turned into an identity when $A \in \gB(\cH)$ also if $A$ is not selfadjoint but only normal}
\beq
||A|| = \sup\{|\lambda| \:|\: \lambda \in \sigma(A)\}\:,\label{spectralradius2}
\eeq
This is the well known {\em spectral radius formula}, the {\bf spectral radius} of $A \in \gB(\cH)$ being, by definition, the number in the right hand side.

 {\bf (d)}  The  result stated in (c) explains the reason   why observables $A$ in QM are very often represented by unbounded selfadjoint operators. $\sigma(A)$ is the set of values of the observable $A$. When, as it happens very often,  that observable is allowed to take arbitrarily large values (think of  $X$ or $P$), it cannot be represented by a bounded selfadjoint operator just because its spectrum is not bounded.

{\bf (e)} If $P$ is a PVM on $\bR$ and $f: \bR \to \bC$ is measurable, we can always write
$$\int_{\bR} f(\lambda) dP(\lambda) = f(A)$$
where we have introduced the  selfadjoint operator $A$ obtained as 
\beq
A= \int_{\bR} id(\lambda) dP(\lambda)\:,
\eeq
due to (\ref{aggf}) and 
where $id : \bR \ni \lambda \to \lambda$. Evidently $P^{(A)}=P$ due to the uniqueness part of the spectral theorem. This fact leads to the conclusion that, {\em in a complex Hilbert space $\cH$, all the PVM over $\bR$
are one-to-one associated to all  selfadjoint operators in $\cH$.}

{\bf (f)} An element $\lambda \in \sigma_c(A)$ is not an eigenvalue of $A$. However there is the following known  result arising from (d) in Theorem \ref{st} \cite{moretti} which proves that we can 
have approximated eigenvalues with arbitrary precision: With the said hypotheses, for every $\epsilon >0$ there is $x_\epsilon \in D(A)$
 such that $$||Ax_\epsilon - \lambda x_\epsilon|| < \epsilon\:, \quad \mbox{but  $||x_\epsilon|| =1$.}$$

{\bf (g)} If $A$ is selfadjoint and $U$ unitary, $UAU^*$, with  $D(UAU^*)= U(D(A))$,  is selfadjoint as well (exercise \ref{esinvarianceU}).
It is very simple to prove that the PVM of $UAU^*$ is noting but $UP^{(A)}U^*$.  \hfill $\blacksquare$
\end{remark}

\noindent The next theorem we state here \label{moretti} concerns a general explicit form of the integral decomposition $f(A) = \int_{\sigma(A)} f(\lambda) dP^{(A)}(\lambda)$. As a matter of facts, up to multiplicity, one can always reduce to a multiplicative operator in a $L^2$ space, as it happens for the position operator $X$. Again, this theorem can be restated for generally normal operators.

\begin{theorem}[Spectral Representation Theorem for Selfadjoint Operators]
Let $A$ be a selfadjoint operator in the complex Hilbert space $\cH$. The following facts hold.\\
{\bf (a)} $\cH$ may be decomposed a  Hilbert sum\footnote{$S$ is countable, at most, if  $\cH$ is separable.}
 $\cH= \oplus_{a\in S} \cH_a$,
whose summands $\cH_a$ are closed and orthogonal. 
Moreover:

(i) for any $a \in S$,  $$A(\cH_a\cap D(A)) \subset \cH_a$$ and, more generally,  for any measurable $f : \sigma(A) \to \bC$,
$$f(A)(\cH_a\cap D(f(A))) \subset \cH_a$$

(ii) for any $a \in S$ there exist a unique finite positive  Borel measure $\mu_a$ on $\sigma(A)\subset \bR$, and a surjective isometric operator  $U_a : \cH_a \to L^2(\sigma(A), \mu_a)$, such that:
$$U_a f(A)|_{\cH_a}U_a^{-1} = f\cdot \:$$
for any measurable $f : \sigma(A) \to \bC$,
where $f\cdot$ is the point-wise multiplication by $f$ on $L^2(\sigma(A),\mu_a)$.\\
{\bf(b)} If $supp\{\mu_a\}_{a\in S}$ is the complementary set to the numbers $\lambda \in \bR$ for which there exists an open set  $O_\lambda \subset \bR$ with  $O_\lambda \ni \lambda$, $\mu_{a}(O_\lambda)=0$ for any $a \in S$, then 
 $$\sigma(A) = supp\{\mu_a\}_{a\in S}\:.$$
\end{theorem}

\remark Notice that the theorem encompasses the case of an operator $A$ in $\cH$ with $\sigma(A)=\sigma_p(A)$. Suppose in particular that every eigenspace is one-dimensional and the whole Hilbert space is separable. Let $\sigma(A)= \sigma_p(A)= \{\lambda_n \:|\: n \in \bN\}$. In this case $$A = \sum_{n \in \bN} \lambda_n \langle x_n |\:\: \rangle x_n\:, $$
where $x_\lambda$ is a unit eigenvector with eigenvalue $\lambda_n$. Consider the $\sigma$-algebra on $\sigma(A)$ made of all subsets and define
$\mu(E) :=$ number of elements of $E\subset \sigma(E)$. In this case 
$\cH$ is isomorphic to $L^2(\sigma(A), \mu)$ and the isomorphism is
$U : \cH \ni x \mapsto \psi_x \in L^2(\sigma(A), \mu)$ with $\psi_x(n) := \langle x_n|x \rangle$ if $n \in \bN$. With this surjective isometry, trivially
$$Uf(A) U^{-1}= U\int_{\sigma(A)} f(\lambda) dP^{(A)}(\lambda)U^{-1}=U \sum_{n \in \bN} f(\lambda_n) \langle x_n |\:\: \rangle x_n U^{-1} = f\cdot\:.$$
If all eigenspaces have dimension $2$, exactly two copies of  $L^2(\sigma(A), \mu)$ are sufficient to improve the construction. If the dimension depends on the eigenspace, the construction can be rebuilt exploiting  many copies of different $L^2(S_k, \mu_k)$, where the $S_k$ are suitable (not necessarily disjoint) subsets of $\sigma(A)$ and $\mu_k$ the measure which counts the elements of $S_k$.  \hfill $\blacksquare$\\

\noindent The last tool we introduce is the notion of {\em joint spectral measure}. Everything is stated in the following theorem \cite{moretti}.

\begin{theorem}[Joint spectral measure]\label{JS}
Consider selfadjoint operators $A_1,A_2,\ldots, A_n$ in the complex separable Hilbert space $\cH$. Suppose that the spectral measures of those operators pairwise commute:
$$P_{E_k}^{(A_k)} P_{E_h}^{(A_h)} = P_{E_h}^{(A_h)} P_{E_k}^{(A_k)}  \quad \forall k,h \in \{1,\ldots,n\}\:, \forall E_k,E_h \in {\cal B}(\bR)\:.$$ 
There is a unique $PVM$, $P^{(A_1\times \cdots \times A_n)}$,  on $\bR^n$ such that
$$P^{(A_1\times \cdots \times A_n)}(E_1\times \cdots \times E_n) = P^{(A_1)}_{E_1} \cdots P^{(A_n)}_{E_n} \:,\quad \forall E_1,\dots, E_n \in {\cal B}(\bR)\:.$$
For every $f : \bR \to \bC$ measurable, it holds
\beq
\int_{\bR^n} f(x_k) dP^{(A_1\times \cdots \times A_n)}(x) = f(A_k) \:, \quad k=1,\ldots, n
\eeq
where $x= (x_1,\ldots, x_k,\ldots, x_n)$. 
\end{theorem} 

\begin{definition} {\em Referring to Theorem \ref{JS}, the PVM $P^{(A_1\times \cdots \times A_n)}$ is called the {\bf joint spectral measure} of $A_1,A_2,\ldots, A_n$
and its support  $supp(P^{(A_1\times \cdots \times A_n)})$, i.e. the complement in $\bR^n$ to the largest open set $A$ with $P_A=0$, is called the 
{\bf joint spectrum} of $A_1,A_2,\ldots, A_n$.} \hfill $\blacksquare$

\end{definition}

\example The simplest example is provided by considering the $n$ position operators 
$X_m$ in $L^2(\bR^n, d^nx)$. It should be clear that the $n$ spectral measures commute because $P^{(X_k)}_E$, for $E\in {\cal B}(\bR)$, is the multiplicative operator for 
$\chi_{\bR\times \cdots \times \bR \times  E \times \bR\times \cdots \times \bR}$
the factor $E$ staying in the $k$-th position among the $n$ Cartesian factors. In this case the joint spectrum of the $n$ operators $X_m$ coincides with $\bR^n$ itself.\\
A completely analogous discussion holds for the $n$ momentum operators $P_k$, since they are related to the position ones by means of the unitary Fourier-Plancherel operator as already seen several times. Again the joint spectrum of the $n$ operators $P_m$ coincides with $\bR^n$ itself. \hfill $\blacksquare$

\subsection{Mesurable functional calculus} The following proposition states some useful properties of $f(A)$, where $A$ is selfadjoint and $f: \bR \to \bC$ is Borel measurable. These properties  define the so called {\em measurable functional calculus}.
We suppose here that $A=A^*$, but the statements can be reformulated for normal operators \cite{moretti}.

\begin{proposition}\label{propint2}
Let $A$ be a  selfadjoint operator  in the complex  Hilbert space $\cH$,  $f, g : \sigma(A) \to \bC$ measurable functions, $f\cdot g$ and $f+g$ respectively  denote the point-wise product and the point-wise sum of functions. The following facts hold.\\

{\bf (a)} $f(A) = \sum_{k=0}^n a_k A^k$ where the right-hand side is defined in its standard domain $D(A^n)$
when $f(\lambda)= \sum_{k=0}^n a_k \lambda^k$ with $a_n \neq 0$.\\

{\bf (b)} $f(A) = P^{(A)}(E)$ if $f= \chi_E$ the characteristic function of $E \in {\cal B}(\sigma(A))$;\\

{\bf (c)}  $f(A)^* = \overline{f}(A)$ where the bar denotes the complex conjugation;\\

{\bf (d)} $f(A)+g(A) \subset (f+g)(A)$ and $D(f(A)+g(A)) \subset \Delta_f \cap \Delta_g$\\
(the symbol $``\subset''$ can be replaced by $``=''$ if and only if $\Delta_{f+ g} = \Delta_f \cap \Delta_g$)\:,\\

{\bf (e)}  $f(A)f(B) \subset (f\cdot g)(A)$ and $D(f(A)f(B))= \Delta_{f\cdot g}\cap \Delta_g$ \\ (the symbol $``\subset''$ can be replaced by $``=''$ if and only if $\Delta_{f\cdot g} \subset \Delta_g$)\:,\\

{\bf (f)}  $f(A)^*f(A) = |f|^2 (A)$  so that $D(f(A)^*f(A))= \Delta_{|f|^2}$\:,\\

{\bf (g)}   $\langle x |f(A) x \rangle \geq 0$ for $x \in \Delta_f$ if $f\geq 0$.\\

{\bf (h)}  $||f(A)x||^2 = \int_{\sigma(A)} |f(\lambda)|^2 d\mu_{xx}(\lambda)$, if $x \in \Delta_f$.\\ 
In particular, if $f$ is bounded or $P^{(A)}$-essentially bounded\footnote{Remark \ref{rembound}.} on $\sigma(A)$, $f(A) \in \gB(\cH)$ and
$$||f(A)|| \leq ||f||^{P^{(A)}}_{\infty} \leq ||f||_{\infty}\:.$$

{\bf (i)} If $U: \cH \to \cH$ is unitary, $Uf(A)U^* = f(UAU^*)$ and, in particular, 
$D(f(UAU^*)) = UD(f(A)) = U(\Delta_f)$.\\

{\bf (j)}  If $\phi : \bR \to \bR$ is measurable, then ${\cal B}(\bR) \ni E \mapsto P'(E):= P^{(A)}(\phi^{-1}(E))$ is a PVM on $\bR$. Introducing  the selfadjoint operator
$$A' = \int_{\bR} \lambda' dP'(\lambda') $$
such that $P^{(A')}= P'$,
we have
$$A' = \phi(A)\:.$$ 
Moreover,  if $f : \bR \to \bC$ is measurable,
$$f(A') = (f \circ \phi)(A) \quad \mbox{and} \quad \Delta'_f = \Delta_{f\circ \phi}\:.$$
\end{proposition}

\subsection{Elementary formalism for the infinite dimensional case}\label{formalism}
To complete the discussion in the introduction, 
 let us  show how practically the physical hypotheses on quantum systems (1)-(3) have to be mathematically  interpreted (again reversing the order of (2) and (3) for our convenience)
in the general case of infinite dimensional Hilbert spaces. 
\noindent Our general assumptions on the mathematical description of quantum systems are the following ones.

\begin{enumerate}
\item   A quantum mechanical system $S$ is always associated to complex Hilbert space ${\cal H}$, finite or infinite dimensional;

\item  observables are  pictured in terms of  (generally unbounded) {\em self-adjoint} operators $A$ in $\cal H $,

\item states are  of equivalence classes of {\em unit} vectors $\psi \in {\cal H}$, where $\psi \sim \psi'$ iff $\psi = e^{ia} \psi'$ for some $a\in \mathbb R$.

\end{enumerate}

\noindent Let us show how the mathematical assumptions (1)-(3)  permit us to set the physical properties 
of quantum systems  (1)-(3) of Section \ref{GP1} into mathematically nice form in the general case of an infinite dimesional Hilbet space $\cH$.\\

{\bf (1) Randomness}:  
The Borel subset $E \subset \sigma(A)$, represents the outcomes of measurement procedures of the observable associated with the selfadjoint operator $A$. (In case of continuous spectrum the outcome of a measurement is at least an interval in view of the experimental errors.) 
Given a state represented by the unit vector  $\psi \in \cal H$, the probability to obtain $E \subset \sigma(A)$ as an outcome when measuring $A$  is
$$\mu^{(P^{(A)})}_{\psi,\psi}(E) := || P^{(A)}_E \psi||^2\:,$$
where we have used the PVM $P^{(A)}$ of the operator $A$.\\
Going along with this interpretation, the {\bf expectation value}, $\langle A \rangle_\psi$, of $A$ when the state is represented by the unit vector $\psi\in \cH$, turns out to be
\beq \langle A \rangle_\psi := \int_{\sigma(A)} \lambda \:  d\mu^{(P^{(A)})}_{\psi,\psi}(\lambda)\:.\label{defexpt}\eeq
This identity makes sense provided $id : \sigma(A) \ni \lambda \to \lambda$ belongs to
$L^1(\sigma(A), \mu^{(P^{(A)})}_{\psi,\psi})$
 (which is equivalent to say that $\psi \in \Delta_{|id|^{1/2}}$ and, in turn, that $\psi \in D(|A|^{1/2})$), otherwise the expectation value is not defined.\\
Since $$L^2(\sigma(A), \mu^{(P^{(A)})}_{\psi,\psi}) \subset L^1(\sigma(A), \mu^{(P^{(A)})}_{\psi,\psi})$$ because $\mu^{(P^{(A)})}_{\psi,\psi}$ is finite, we have the popular identity arising from (\ref{intop2}),
\begin{equation}
\langle A \rangle_\psi = \langle \psi | A \psi \rangle \quad \mbox{if $\psi \in D(A)$} \label{E2}\:.
\end{equation}
The associated {\bf standard deviation}, $\Delta A_\psi$, results to be
\begin{equation}\Delta A_\psi^2 :=  \int_{\sigma(A)} (\lambda-  \langle A \rangle_\psi)^2\:  d\mu^{(P^{(A)})}_{\psi,\psi}(\lambda) \:.\label{Delta22} \end{equation}
This definition makes sense provided
$id \in L^2(\sigma(A), \mu^{(P^{(A)})}_{\psi,\psi})$ (which is equivalent to say that $\psi \in \Delta_{id}$ and, in turn, that $\psi \in D(A)$).\\ As before, the functional calculus permits us to write the other popular identity
\beq \Delta A_\psi^2  = \langle \psi | A^2  \psi \rangle - \langle\psi|  A \psi\rangle^2 
\quad \mbox{if $\psi \in D(A^2)\subset D(A)$}\:.\label{Delta3}\eeq
We stress that now, Heisenberg inequalities, as established in exercise \ref{eH}, are now completely justified as the reader can easily check.\\

 {\bf (3) Collapse of the state}: If the Borel set $E \subset \sigma(A)$ is the outcome of the (idealized) measurement of $A$, when the state is represented by the unit vector $\psi\in \cH$, the new state immediately after the measurement is represented by the unit vector
\begin{equation}\psi' := \frac{P_{E}^{(A)}\psi}{||P_{E}^{(A)}\psi ||}\label{LvN2}\:.\end{equation}

\remark  Obviously this formula does not make sense if $\mu^{(P^{(A)})}_{\psi,\psi}(E)=0$ as expected. Moreover 
the arbitrary phase affecting $\psi$ does not give rise to troubles due to the linearity of $P^{(A)}_E$ . \hfill $\blacksquare$\\

 {\bf (2) Compatible and  Incompatible Observables}: Two observables $A$, $B$ are compatible -- i.e. they can be simultaneously measured -- if and only if their {\bf spectral measures commute} which means
\beq P_E^{(A)} P_F^{(B)} = P_F^{(B)}P_E^{(A)} \:,\quad E \in {\cal B}(\sigma(A))\:, \quad F \in {\cal B}(\sigma(B))\:.\label{mco}\eeq
In this case 
$$||P^{(A)}_E P_F^{(B)} \psi||^2 = || P_F^{(B)} P^{(A)}_E \psi||^2 = ||P^{(A,B)}_{E\times F}\psi||^2$$
where $P^{(A,B)}$ is the joint spectral measure of $A$ and $B$,
has the natural interpretation of the probability to obtain the outcomes $E$ and $F$ for a simultaneous measurement of $A$ and $B$.  If instead $A$ and $B$ are incompatible  it may happen that
$$||P^{(A)}_E P_F^{(B)} \psi||^2 \neq  || P_F^{(B)} P^{(A)}_E \psi||^2\:.$$
Sticking  to the case of  $A$ and $B$ incompatible, exploiting  (\ref{LvN2}),  
\begin{equation}||P^{(A)}_E P_F^{(B)} \psi||^2 = \left|\left|P^{(A)}_E  \frac{P_F^{(B)} \psi}{|| P_F^{(B)} \psi ||}\right|\right|^2  || P_F^{(B)} \psi ||^2\label{measureincomp2}\end{equation} 
has the natural meaning of {\em the probability of obtaining first $F$ and next  $E$ in a subsequent measurement of $B$ and $A$}.

\begin{remark}\label{probtruble}$\null$

{\bf (a)} It is worth stressing that the notion of probability we are using here cannot be a classical notion because of the presence of incompatible observables.  The theory of conditional probability cannot follows the standard rules. The probability $\mu_\psi(E_A|F_B)$, that (in a state defined by a unit vector $\psi$) a certain observable $A$ takes the value $E_A$ when the observable $B$ has the value $F_B$, cannot be computed by the standard procedure
$$\mu_\psi(E_A|F_B)= \frac{\mu_\psi(E_A \mbox{ AND } F_B)}{\mu_\psi(F_B)}$$
 if $A$ and $B$ are incompatible, just because, in general, nothing exists which can be interpreted as the event ``$E_A \mbox{ AND } F_B$'' if $P^{(A)}_E$ and $P^{(B)}_F$ do not commute! The correct formula is
$$\mu_\psi(E_A|F_B)= \frac{\langle \psi|P^{(B)}_F P^{(A)}_E P^{(B)}_F \psi \rangle}{||P^{(B)}_F \psi||^2}$$
which leads to well known different properties with respect to the classical theory, the so called combination of ``probability amplitudes'' in particular. As a matter of fact, up to now we do not have a clear notion of (quantum) probability. This issue will be clarified in the next section.  

{\bf (b)} The reason to pass from operators to their spectral measures in defining compatible observables is that, if $A$ ad $B$ are selfadjoint and defined on different domains, $AB=BA$ does not make sense in general. Moreover it is possible to find counterexamples (due to Nelson) where commutativity of $A$ and $B$ on common dense invariant subspaces does not implies that their spectral measures commute. However, from general results again due to Nelson, one has the following nice result 
(see exercise \ref{ABcomm}).
\begin{proposition}
If selfadjoint  operators, $A$ and $B$, in a complex Hilbert space $\cH$ commute on 
 a common dense invariant domain $D$ where $A^2+B^2$ is essentially selfadjoint, then  the spectral measures of $A$ and $B$ commute. 
\end{proposition} 
\noindent The following result, much  easier to prove,  is also  true \cite{moretti}.
\begin{proposition}\label{propcomm} Let $A$, $B$ be selfadjoint operators in the complex 
Hilbert space $\cH$. If $B\in \gB(\cH)$ the following facts are equivalent,\\

(i) the spectral measures of $A$ and $B$ commute (i.e. (\ref{mco}) holds),\\

(ii) $BA\subset AB$\:,\\

(iii) $Bf(A) \subset f(A)B$, if $f: \sigma(A) \to \bR$ is Borel measurable\:,\\

(iv) $P_E^{(A)}B=BP_E^{(A)}$ if $E \in {\cal B}(\sigma(A))$\:,
\end{proposition}

\noindent Another useful result toward the converse direction  \cite{moretti} is the following.
\begin{proposition} Let $A$, $B$ be selfadjoint operators in the complex 
Hilbert space $\cH$ such that their spectral measures  commute. The following facts hold.\\
{\bf (a)} $ABx=BAx$ if $x \in D(AB)\cap D(BA)$\:.\\
{\bf (b)} $\langle Ax|By\rangle = \langle Bx|Ay\rangle$ if $x,y \in D(A)\cap D(B)$. 
\end{proposition}
\end{remark} \hfill $\blacksquare$

\subsection{Technical Interemezzo: Three Operator Topologies}
In QM there are at least 7 relevant topologies \cite{BrRo} which enter the game discussing sequences of operators, here we limit ourselves to quickly  illustrate the three most important ones. We assume that $\cH$ is a complex Hilbert space though the illustrated examples may be extended to more general context with some re-adaptation.

{\bf (a)} The strongest topology is the {\bf uniform operator topology} in $\gB(\cH)$: It is the topology induced by the operator norm $||\:\:||$ defined in (\ref{ON}).\\ As a consequence of the definition of this topology, a sequence of elements 
$A_n \in \gB(\cH)$ is said to {\bf uniformly} converge to $A\in \gB(\cH)$ when $||A_n-A|| \to 0$ for $n \to +\infty$.\\
We already know that $\gB(\cH)$ is a Banach algebra with respect to that norm and also a $C^*$ algebra. 

{\bf (b)} If ${\gL}(D; \cH)$ with $D\subset \cH$ a subspace, denotes the complex vector space of the operators $A: D \to\cH$, the {\bf strong operator topology} on $\gL(D; \cH)$ is the topology induced by the seminorms $p_x$ with $x \in D$ and $p_x(A) := ||Ax||$ if $A \in \gL(D;\cH)$.\\  As a consequence of the definition of this topology, a sequence of elements 
$A_n \in \gL(D; \cH)$ is said to {\bf strongly} converge to $A\in \gL(D;\cH)$ when $||(A_n-A)x|| \to 0$ for $n \to +\infty$ for every $x\in D$.\\
It should be evident that, if we restrict ourselves to work in $\gB(\cH)$, the uniform operator topology is stronger than the strong operator 
topology.

{\bf (c)} The {\bf weak operator topology} on $\gL(D;\cH)$ is the topology induced by the seminorms $p_{x,y}$ with $x\in \cH$, $y\in D$ and $p_{x,y}(A) := |\langle x|Ay\rangle|$ if $A \in \gL(D; \cH)$.\\  As a consequence of the definition of this topology, a sequence of elements 
$A_n \in \gL(D;\cH)$ is said to {\bf weakly} converge to $A\in \gL(D;\cH)$ when $| \langle x|(A_n-A)y \rangle || \to 0$ for $n \to +\infty$ for every $x\in \cH$ and $y\in D$.\\
It should be evident that, the strong operator topology is stronger than the weak operator 
topology.

\example\label{exstrong} $\null$\\
{\bf (1)} If $f: \bR \to \bC$ is Borel measurable, and $A$ a selfadjoint operator in $\cH$, consider the sets 
$$R_n := \{r \in \bR \:|\: |f(r)|< n\}\quad \mbox{ for $n \in \bN$\:.}$$ It is clear that $\chi_{R_n}f \to f$ pointwise as $n \to +\infty$ and that $|\chi_{R_n}f|^2\leq |f|^2$. As a consequence restricting the operators on the left hand side to  $\Delta_f$,
$$\left.\int_{\sigma(A)} \chi_{R_n}f dP^{(A)}\right|_{\Delta_f}\to  f(A) \quad\mbox{strongly,  for $n\to +\infty$,}$$
as an immediate consequence of Lebesgue's dominate convergence theorem and the first part of (h) in Proposition \ref{propint2}. \\
{\bf (2)} If in the previous example  $f$ is bounded  on $\sigma(A)$,  and $f_n \to f$ uniformly on $\sigma(A)$, (or $P$-essentially uniformly $||f-f_n||^{(P^{(A)})}_\infty \to 0$ for $n \to +\infty$) then 
$$f_n(A) \to f(A)\quad  \mbox{uniformly, as $n\to+ \infty$,}$$
again for the (second part of  (h) in Proposition \ref{propint2}.\hfill $\blacksquare$

{\bf \exercise\label{esanv}} {\em Prove that a selfadjoint operator $A$ in the complex Hlbert $\cH$ admits a dense set of analytic vectors in its domain.}

{\bf Solution}. Consider the class of functions $f_n = \chi_{[-n,n]}$ where $n\in \bN$. As in (1) of  example \ref{exstrong}, we have 
$\psi_n := f_n(A)\psi = \int_{[-n,n]} 1 dP^{(A)}\psi  \to \int_\bR 1 dP^{(A)}\psi = P^{(A)}_{\bR}\psi =  \psi $ for $n \to +\infty$.
Therefore the set $D:= \{\psi_n\:|\: \psi \in \cH\:, n \in \bN\}$ is dense in $\cH$. The elements of $D$ are analytic vectors for $A$ as we go to prove.  Clearly $\psi_n \in D(A^k)$ since $\mu^{(P^{(A)})}_{\psi_n,\psi_n}(E)
= \mu^{(P^{(A)})}_{\psi,\psi}(E\cap [-n,n])$  as immediate consequence of the definition of the measure $\mu^{(P)}_{x,y}$, therefore  $\int_{\bR} |\lambda^k|^2 d \mu^{(P^{(A)})}_{\psi_n,\psi_n}(\lambda)
= \int_{[-n,n]} |\lambda|^{2k} d \mu^{(P^{(A)})}_{\psi,\psi}(\lambda) \leq  \int_{[-n,n]}|n|^{2k} d \mu^{(P^{(A)})}_{\psi,\psi}(\lambda) \leq  |n|^{2k} \int_{\bR} d \mu^{(P^{(A)})}_{\psi,\psi}(\lambda) =  |n|^{2k} ||\psi||^2 < +\infty$. Similarly
$||A^k\psi_n||^2 = \langle A^k\psi_n | A^k \psi_n\rangle =  \langle \psi_n | A^{2k} \psi_n\rangle = \int_{\bR} \lambda ^{2k} d \mu^{(P^{(A)})}_{\psi_n,\psi_n}(\lambda)  \leq |n|^{2k} ||\psi||^2$. We conclude that
$\sum_{k=0}^{+\infty} \frac{(it)^k}{k!} ||A^k\psi_n||$ conveges for every $t\in \bC$ as it is dominated 
by the series $\sum_{k=0}^{+\infty} \frac{|t|^k}{k!} |n|^{2k} ||\psi||^2 = e^{|t|\:|n|^2} ||\psi||^2$. \hfill $\blacksquare$

\section{More Fundamental Quantum Structures}
The question we want to answer now is the following: \\
{\em Is there anything more fundamental behind the phenomenological facts (1), (2), and (3)  discussed in the first section and their formalization presented in Sect. \ref{formalism}?}\\
 An appealing attempt to answer that question and justify the formalism based on the spectral theory 
is due to von Neumann \cite{vonNeumann}  (and subsequently extended by Birkhoff and von Neumann).
This section is devoted to quickly review an elementary part of those ideas, adding however several more  modern results
(see also \cite{varadarajan} for a similar approach).

\subsection{The  Boolean logic of CM}\label{exf}
Consider a classical Hamiltonian system described in symplectic manifold $(\Gamma, \omega)$, where $\omega = \sum_{k=1}^n dq^k\wedge dp_k$ in any system of local symplectic coordinates $q^1,\ldots, q^n,p_1,\ldots,p_n$.
The state of the system at time $t$ is a point $s \in \Gamma$,
in local coordinates $s\equiv (q^1,\ldots, q^n,p_1,\ldots,p_n)$,
whose evolution $\bR \ni t \mapsto s(t)$
is a solution of the {\em Hamiltonian equation} of motion. Always in local symplectic coordinates, they read
\begin{eqnarray}
\frac{dq^k}{dt} &=& \frac{\partial H(t,q,p)}{\partial p_k}\:, \quad k=1,\ldots, n\\
\frac{dp_k}{dt} &=& -\frac{\partial H(t,q,p)}{\partial q^k} :, \quad k=1,\ldots, n\:,
\end{eqnarray}
$H$ being the Hamiltonian function of the system, depending on the (inertial) reference frame.
Every physical {\em elementary  property},  $E$, that the system may possess at a certain time $t$, i.e., which can be true or false at that time, can be identified with a subset $E\subset \Gamma$. The property is true if $s\in E$ and it is not if $s \not \in E$. From this point of view,
the standard set theory  operations $\cap$,  $\cup$, $\subset$,  $\neg$ (where $\neg  E := \Gamma \setminus E$ from now on is the {\bf complement operation}) have a logical interpretation: \\

(i) $E \cap F$ corresponds to the property ``$E$ AND $F$'',

(ii) $E \cup F$ corresponds to the property ``$E$ OR $F$'', 

(iii) $\neg E$ corresponds to the property ``NOT  $F$'', 

(iv) $E \subset  F$ means ``$E$ IMPLIES $F$''.\\

\noindent In this context:\\

(v)  $\Gamma$  is the property which is always true

(vi)  $\emptyset$  is the property which is always false. \\

\noindent This identification is possible because, as is well known, the logical operations have the same algebraic structure of the set theory operations.\\
As soon as we admit the possibility to construct statements including {\em countably infinite number  of disjunctions or conjunctions},  we can enlarge our interpretation towards the abstract measure theory, interpreting the states as  {\em probability Dirac measures} supported on a single point. 
To this end we first restrict the class of possible elementary properties to the Borel $\sigma$-algebra of $\Gamma$, ${\cal B}(\Gamma)$. For various reasons this class of sets seems to be sufficiently large to describe physics (in particular ${\cal B}(\Gamma)$ includes the preimages of measurable sets under continuous functions). A state at time $t$, 
$s\in \Gamma$, can be viewed as a Dirac measure, $\delta_s$, supported on $s$ itself. 
If $E \in {\cal B}(\Gamma)$, $\delta_s(E)=0$ if $s \not \in E$ or $\delta_s(E)=1$ if $s \in E$.\\
If we do not have a perfect knowledge of the system, as for instance it happens in {\em statistical mechanics}, 
the state at time $t$, $\mu$, is a proper probability measure on ${\cal B}(\Gamma)$ which now, is allowed to attain all values of $[0,1]$. If $E \in {\cal B}(\Gamma)$
is an elementary property of the physical system, $\mu(E)$ denotes  the probability that the property $E$ is true for the system at time $t$. 
\begin{remark} The evolution equation of $\mu$, in statistical mechanics is given by the well-known 
{\em Liouville's equation} associate with the Hamiltonian flow. In that case $\mu$ is proportional to the natural symplectic volume measure of $\Gamma$, $\Omega = \omega \wedge \cdots \wedge \omega$ ($n$-times, where $2n = dim(\Gamma)$). 
In fact we have $\mu = \rho \Omega$, where the non-negative function $\rho$ is the so-called {\bf Liouville density} satisfying the famous  {\em Liouville's equation}. In symplectic local coordinates that equation reads
$$\frac{\partial \rho(t,q,p)}{\partial t} +\sum_{k=1}^n \left(\frac{\partial \rho}{\partial q^k} \frac{\partial H}{\partial p_k}-\frac{\partial \rho}{\partial p_k} \frac{\partial H}{\partial q^k}\right)=0\:.$$
We shall not deal any further with this equation in this paper. \hfill $\blacksquare$
\end{remark}
\noindent More complicated classical quantities of the system can be described by {\em Borel measurable} functions $f: \Gamma \to \bR$.
Measurability is a good requirement as it permits one to perform physical operations like computing, for instance,  the {\em expectation value} (at a given time) when the state is $\mu$:
$$\langle f \rangle_\mu = \int_\Gamma f \mu\:.$$
 Also elementary properties can be pictured by measurable functions, in fact they are one-to-one identified with all the Borel measurable functions $g: \Gamma \to \{0,1\}$. The Borel set $E_g$ associated to $g$ is $g^{-1}(\{1\})$ and in fact $g=\chi_{E_g}$. \\
A generic physical quantity, a measurable function  $f: \Gamma \to \bR$, is completely determined by the class of Borel sets (elementary properties) $E^{(f)}_B := f^{-1}(B)$ where $B \in {\cal B}(\bR)$. 
The meaning of $E^{(f)}_B$ is 
\beq  E^{(f)}_B =\mbox{``the value of $f$ belongs to $B$''} \label{meaningE}\eeq
It is possible to prove \cite{moretti} that the map $ {\cal B}(\bR) \ni B \mapsto E^{(f)}_B$ permits one to reconstruct the function $f$.
The sets $E^{(f)}_B := f^{-1}(B)$ form a $\sigma$-algebra as well and 
 the class of sets $E^{(f)}_B$ satisfies the following elementary properties when 
$B$ ranges in ${\cal B}(\bR)$.

 {\bf (Fi)}\:  $E^{(f)}_\bR = \Gamma$, 

{\bf (Fii)}\: $E^{(f)}_B \cap E^{(f)}_C =E^{(f)}_{B\cap C}$,

 {\bf (Fiii)}\: If $N \subset \bN$ and $\{B_k\}_{k\in N} \subset {\cal B}(\bR)$
satisfies  $B_j \cap B_k = \emptyset$ if $k\neq j$, then
$$\cup_{j \in N} E^{(f)}_{B_j} = E^{(f)}_{\cup_{j\in N}B_j}\:.$$
These conditions just say that $ {\cal B}(\bR) \ni B \mapsto E^{(f)}_B$ is a {\bf homomorpism of $\sigma$-algebras}.\\

\noindent For future convenience we observe that our model of {\em classical} elementary properties  can be also viewed as another 
mathematical structure, when referring to the notion of {\em lattice}.
\begin{definition} A partially ordered set $(X, \geq)$ is a  {\bf lattice} when, 
for any $a,b \in X$, \\
{\bf (a)} $\sup\{a,b\}$ exists, denoted $a \vee b$ (sometimes called `join');\\
{\bf (b)} $\inf\{a,b\}$ exists, written $a \wedge b$ (sometimes `meet').\\
(The partially ordered set is not required to be totally ordered.) \hfill $\blacksquare$
\end{definition}

\remark$\null$

 {\bf (a)} In our considered concrete case $X={\cal B}(\bR)$ and $\geq$ is nothing but $\supset$ and thus  $\vee$  means $\cup$
and $\wedge$ has the meaning of $\cap$.

{\bf (b)} In the general case $\vee$ and $\wedge$ turn out to be separately {\em associative}, therefore it make sense to write $a_1  \vee \cdots \vee a_n$ and  $a_1  \wedge \cdots \wedge a_n$ in a lattice.
Moreover they are also separately {\em commutative} so 
$$a_1  \vee \cdots \vee a_n = a_{\pi(1)}  \vee \cdots \vee a_{\pi(n)}\quad \mbox{and}\quad
a_1  \wedge \cdots \wedge a_n = a_{\pi(1)}  \wedge \cdots \wedge a_{\pi(n)}$$
for every permutation $\pi : \{1,\ldots, n\} \to  \{1,\ldots, n\}$.  \hfill $\blacksquare$

\begin{definition} A lattice $(X,\geq)$ is said to be:\\
{\bf (a)} {\bf distributive} if  $\vee$ and $\wedge$ distribute over one another: for any  $a,b,c\in X$,
\beq
a\vee (b\wedge c)  = (a\vee b)\wedge (a\vee c) \:,\quad  a\wedge (b\vee c)  = (a\wedge b)\vee (a\wedge c) \:;\nonumber
\eeq
{\bf (b)} {\bf bounded} if it admits a minimum ${\bf 0}$ and a maximum ${\bf 1}$ (sometimes called `bottom' and `top');\\
{\bf (c)} {\bf orthocomplemented} if bounded and equipped with a mapping $X\ni  a \mapsto \neg a$, where $\neg a$ is the  {\bf orthogonal complement} of $a$, such that:

(i) $a \vee \neg a = {\bf 1}$ for any $a\in X$,

(ii) $a \wedge \neg a = {\bf 0}$    for any $a\in X$,

(iii) $\neg (\neg a) = a$ for any $a\in X$,

(iv) $a\geq b$ implies $\neg b \geq \neg a$ for any $a,b \in X$;\\
{\bf (d)} {\bf $\sigma$-complete}\index{$\sigma$-complete lattice}, if every countable set 
 $\{a_n\}_{n\in \bN} \subset X$ admits least upper bound  $\vee_{n\in \bN} a_n$.\\
A lattice with properties (a), (b) and (c) is called a  {\bf Boolean algebra}. A Boolean algebra satisfying (d) is a {\bf Boolean $\sigma$-algebra}.  \hfill $\blacksquare$
\end{definition}

\begin{definition} If $X$, $Y$ are lattices, a map $h: X\to Y $ is a {\bf  (lattice) homomorphism} when
$$h(a \vee_X b) = h(a) \vee_Y h(b) \:,  \: \: \:h(a \wedge_X b) = h(a) \wedge_Y h(b)\:,\:\:\:\mbox{$a,b \in X$}$$
(with the obvious notations.) 
If $X$ and $Y$ are bounded, a homomorphism $h$ is further required to satisfy
$$h({\bf 0}_X) =  {\bf 0}_Y \:, \quad h({\bf 1}_X) =  {\bf 1}_Y \:.$$
If $X$ and $Y$ are orthocomplemented, a homomorphism $h$ also satisfies
$$ h(\neg_X a) =  \neg_Y h(x)\:.$$
If $X$, $Y$ are $\sigma$-complete, $h$ further fulfills
$$h(\vee_{n\in \bN} a_n) = \vee_{n\in \bN} h(a_n) \:, \mbox{if  $\{a_n\}_{n\in \bN} \subset X$}\:.$$
In all cases (bounded, orthocomplemented, $\sigma$-complete lattices, Boolean ($\sigma$-)algebras) if $h$ is bijective it is called {\bf isomorphism} of the relative structures.
\end{definition}
\noindent It is clear that, just because it is a concrete $\sigma$-algebra,  the lattice of the elementary properties of a classical system 
is a lattice which is  {\em distributive}, 
{\em bounded}  (here $0=\emptyset$ and $1= \Gamma$), {\em orthocomplemented} (the orthocomplement being the complement with respect to $\Gamma$) and {\em $\sigma$-complete}. Moreover, as the reader can easily prove, the above map, ${\cal B}(\bR) \ni B \mapsto E^{(f)}_B$,
is also a homomorphism  of Boolean $\sigma$-algebras.

\remark  Given an abstract Boolean $\sigma$-algebra $X$, does there exist a concrete $\sigma$-algebra of sets that is isomorphic to the previous one? In this respect the following general result holds, known as {\em Loomis-Sikorski theorem}.\footnote{Sikorski S.: {\em On the representation of
Boolean algebras as field of sets.} Fund. Math. {\bf 35}, 247-256
(1948).} This guarantees that every Boolean $\sigma$-algebra is isomorphic to a quotient Boolean $\sigma$-algebra $\Sigma/{\cal N}$, where $\Sigma$ is a concrete $\sigma$-algebra of sets over a measurable space and ${\cal N}\subset \Sigma$ is closed under countable unions; moreover, $\emptyset \in {\cal N}$ and for any $A \in \Sigma$ with $A \subset N \in {\cal N}$, then  $A \in {\cal N}$.  The equivalence relation  is $A\sim B$ iff $A\cup B \setminus (A \cap B) \in {\cal N}$, for any $A,B \in \Sigma$.  It is easy to see the coset space $\Sigma/{\cal N}$ inherits the structure of Boolean $\sigma$-algebra from $\Sigma$ with respect to the (well-defined) partial order relation 
$[A] \geq [B]$ if $A \supset B$, $A,B \in \Sigma$. \hfill $\blacksquare$

\subsection{The non-Boolean Logic of QM, the reason why observables are selfadjoint operators.} It is evident that the classical like picture illustrated in  Sect. \ref{exf}  is untenable if referring to quantum systems. The deep reason is that there are pair of elementary properties $E,F$ of quantum systems which are incompatible.  Here an elementary property is an observable which, if measured by means of a corresponding experimental apparatus,  can only attain two values: $0$ if it is false or $1$ if it is true.
For instance, $E=$ ``the component $S_x$ of the electron is $\hbar/2$'' and $F=$ ``the component $S_y$ of the electron is $\hbar/2$''. There is no physical instrument capable to establish if  $E$ AND $F$  is true or false. We conclude that some of elementary observables of quantum systems cannot be logically combined  by the standard operation of the logic.
The model of Borel $\sigma$-algebra seems not to be appropriate for quantum systems.
However one could try to use some form of lattice structure different form the classical one. The fundamental ideas by von Neumann were the following pair. \\

{\bf (vN1)} Given a quantum system, there is a complex separable Hilbert space $\cH$ 
such that the {\bf elementary observables} -- the ones which only assume values in  $\{0,1\}$ --   are one-to-one represented by all the elements of ${\cal L}(\cH)$,  the orthogonal projectors in $\gB(\cH)$.\\

{\bf (vN2)} Two elementary observables $P$, $Q$ are compatible if and only if they commute as projectors.

\begin{remark} $\null$

{\bf (a)} As we shall see later (vN1) has to be changed for those quantum systems which admit {\em superselection rules}. For the moment we stick to the above version of (vN1).

{\bf (b)} The technical requirement of separability will play a crucial role in several places. \hfill $\blacksquare$
\end{remark}
\noindent Let us analyse the reasons for von Neumann's postulates. First of all we observe that ${\cal L}(\cH)$ is in fact a lattice if one remembers the relation between orthogonal projectors and closed subspaces stated in Proposition \ref{propproj}. 

\notation Refferring to Proposition \ref{propproj}, if $P,Q \in {\cal L}(\cH)$, we write
$P \geq Q$ if and only if $P(\cH) \supset Q(\cH)$.
\hfill $\blacksquare$\\

\noindent $P(\cH) \supset Q(\cH)$ is equivalent to  $PQ= Q$. Indeed, if $P(\cH) \supset Q(\cH)$ then there is a Hilbert basis of $P(\cH)$ $N_P = N_Q\cup N'_Q$ where $N_Q$ ia a Hilbert basis of $Q(\cH)$ and $N'_Q$ of $Q(\cH)^{\perp_P}$, the notion of orthogonal being referred to the Hilbert space $P(\cH)$.
From $Q= \sum_{z \in N_Q} \langle z|\cdot \rangle z$ and  $P= Q+\sum_{z \in N'_Q} \langle z|\cdot \rangle z$ we have $PQ=Q$. The converse implication is obvious.\\

\noindent As preannounced, it turns out that $({\cal L}(\cH), \geq)$ is a lattice and, in particular, it enjoys the following properties (e.g., see \cite{moretti}) whose proof is direct.
 
\begin{proposition}\label{proplat}
Let $\cH$ be a complex separable Hilbert space
and,  if $P \in {\cal L}(\cH)$, define $\neg P := I-P$ (the orthogonal projector onto $P(\cH)^\perp$). With this definition, $({\cal L}(\cH), \geq, \neg)$ 
turns out to be  bounded, orthocomplemented, $\sigma$-complete lattice which is not distributive if $dim(\cH)\geq 2$. \\
More precisely, 

(i) $P\vee Q$ is the orthogonal projector onto $\overline{P(\cH)+Q(\cH)}$. \\
The analogue holds for a countable set $\{P_n\}_{n \in \bN}\subset {\cal P}(\cH)$,
$\vee_{n\in \bN}P_n$ is the orthogonal projector onto $\overline{+_{n\in \bN}P_n(\cH)}$.

(ii) $P\wedge Q$ is the orthogonal projector on $P(\cH)\cap Q(\cH)$.\\
The analogue holds for a countable set $\{P_n\}_{n \in \bN}\subset {\cal P}(\cH)$,
$\wedge_{n\in \bN}P_n$ is the orthogonal projector onto $\cap_{n\in \bN}P_n(\cH)$.

(iii) The bottom and the top are respectively $0$ and $I$.\\
Referring to (i) and (ii), it turns out that 
$$\vee_{n\in \bN} P_n = \lim_{k \to +\infty} \vee_{n\leq  k} P_n \quad \mbox{and}\quad 
\wedge_{n\in \bN} P_n = \lim_{k \to +\infty} \wedge_{n\leq  k} P_n $$
with respect to the strong operator topology.
\end{proposition}

\remark The fact that the distributive property does not hold is evident from the following elementary counterexample in $\bC^2$ (so that it is valid for {\em every} dimension $> 1$). Let $\{e_1,e_2\}$ be the standard basis of $\bC^2$ and define the subspaces
$\cH_1 := span(e_1)$, $\cH_2:= span(e_2)$, $\cH_3:= span(e_1+e_2)$. Finally $P_1$, $P_2$,
$P_3$ respectively denote the orthogonal projectors onto these spaces. By direct inspection one sees that
$P_1\wedge (P_2 \vee P_3) = P_1 \wedge I = P_1$ and $(P_1 \wedge P_2)\vee (P_1 \wedge P_3)= 0 \vee 0 = 0$,
so that 
$P_1\wedge (P_2 \vee P_3) \neq (P_1 \wedge P_2)\vee (P_1 \wedge P_3)$.\hfill $\blacksquare$\\

\noindent The crucial observation is that, nevertheless $({\cal L}(\cH), \geq, \neg)$ includes lots of Boolean $\sigma$ algebras, and precisely the maximal sets of pairwise compatible projectors \cite{moretti}.

\begin{proposition}\label{propclass} Let $\cH$ be a complex separable Hilbert space
and consider the lattice $({\cal L}(\cH), \geq, \neg)$. If ${\cal L}_0 \subset {\cal L}(\cH)$ is a maximal subset of pairwise commuting elements, then ${\cal L}_0$ contains $0$, $I$ is $\neg$-closed and, if equipped
with the restriction of the lattice structure of $({\cal L}(\cH), \geq, \neg)$, turns out to be  a Boolean $\sigma$-algebra.\\
In particular, if $P,Q \in {\cal L}_0$,

(i) $P\vee Q = P+Q-PQ$\:,

(ii) $P \wedge Q = PQ$. 
\end{proposition}

\noindent {\bf Proof}. ${\cal L}_0$ includes both $0$ and $I$ because ${\cal L}_0$ is maximally commutative.
Having (i) and (ii), due to (iii) in proposition \ref{proplat}, the $\sup$ and the $\inf$ of a sequence of projectors of ${\cal L}_0$ commute with the elements of ${\cal L}_0$, maximality implies that they belong to ${\cal L}_0$. Finally (i) and (ii) prove by direct inspection that $\vee$ and $\wedge$ are mutually distributive. Let us prove (ii) and (i) to conclude. If $PQ=QP$, $PQ$ is an orthogonal projector and $PQ(\cH) = QP(\cH)\subset P(\cH)\cap Q(\cH)$. On the other hand, if $x\in P(\cH) \cap Q(\cH)$ then $Px=x$ and $x=Qx$ so that $PQx=x$
and thus $P(\cH)\cap Q(\cH) \subset PQ(\cH)$ and (ii) holds. To prove (i) observe that 
$\overline{<P(\cH), Q(\cH)>}^\perp = P(\cH)^\perp \cap Q(\cH)^\perp$. Using (ii), this can be rephrased as $I-P\vee Q= (I-P)(I-Q)$ which entails (i) immediately.  \hfill $\blacksquare$

\remark $\null$

{\bf (a)} Every set of pairwise commuting orthogonal projectors can be completed to a maximal set as an elementary application of Zorn's lemma.
However, since the commutativity property is {\em not} transitive, there are many possible maximal subsets of pairwise commuting elements in ${\cal L}(\cH)$ with non-empty intersection.  

{\bf (b)} As a consequence of the stated proposition, the symbols $\vee$, $\wedge$ and $\neg$ have the same properties in ${\cal L}_0$ as the corresponding symbols of classical logic $OR$, $AND$ and $NOT$. Moreover $P \geq Q$ can be interpreted as ``$Q$ IMPLIES $P$''.

{\bf (c)} There were and are many attempts to interpret $\vee$ and $\wedge$ as connectives of a new non-distributive  logic when dealing with the whole  ${\cal L}(\cH)$:  a {\em quantum logic}. The first noticeable proposal was due to Birkhoff and von Neumann \cite{BvN}. Nowadays there are lots of quantum logics \cite{BC,egl}  all regarded with suspicion by physicists. Indeed, the most difficult issue is the physical operational interpretation of these connectives taking into account the fact that they put together  incompatible propositions, which cannot be measured simultaneously. An interesting interpretative attempt, due to Jauch, relies upon a result by von Neumann (e.g., \cite{moretti})
$$(P \wedge Q)x  = \lim_{n \to +\infty} (PQ)^nx \quad \mbox{for every $P,Q \in {\cal L}(\cH)$ and $x \in \cH$.}$$
Notice that the result holds in particular if $P$ and $Q$ do not commute, so they are incompatible elementary observables. The right hand side of the identity above can be interpreted as the consecutive and alternated measurement of an infinite sequence of elementary observables $P$ and $Q$.  As 
$$||(P \wedge Q)x||^2  = \lim_{n \to +\infty} ||(PQ)^nx||^2 \quad \mbox{for every $P,Q \in {\cal L}(\cH)$ and $x \in \cH$,}$$
the probabilty that $P \wedge Q$ is true for a state  represented by the unit vector $x\in \cH$ is the probabilty that  the infinite sequence  of consecutive alternated measurements of $P$ and $Q$ produce is true at each step.\hfill $\blacksquare$\\

\noindent We are in a position to clarify why, in this context, observables are PVMs.
Exactly as in CM, an observable  $A$ is a collection of elementary observables  $\{P_E\}_{E\in {\cal B}(\bR)}$ labelled on the Borel sets $E$ of $\bR$. Exactly as for classical quantities, (\ref{meaningE}) we can say that the meaning of $P_E$ is
\beq  P_E =\mbox{``the value of the observable belongs to $E$''} \label{meaningP}\eeq
We expect that all those elementary observables are pairwise compatible and that they satisfy the same properties (Fi)-(Fiii) as for classical quantities. We can complete 
$\{P_E\}_{E\in {\cal B}(\bR)}$ to a maximal set of compatible elementary observables.
 Taking Proposition \ref{propclass} into account (Fi)-(Fiii) translate into

 (i)  $P_\bR=I$, 

(ii) $P_EP_F = P_{E\cap F}$,

 (iii) If $N \subset \bN$  and $\{E_k\}_{k\in N} \subset {\cal B}(\bR)$
satisfies  $E_j \cap E_k = \emptyset$ for $k\neq j$, then
$$\sum_{j \in N} P_{E_j}x= P_{\cup_{j\in N}E_j}x \quad \mbox{for every $x\in \cH$.}$$
(The presence of $x$ is due to the fact that the convergence of the series if $N$ is infinite
is in the strong operator topology as declared in the last statement of Proposition \ref{proplat}.)
{\em In other words we have just  found Definition \ref{defPVM}, specialized to PVM on $\bR$: 
Observables in QM are PVM over $\bR$!}\\ We know that all PVM over $\bR$ are one-to-one  associated to all selfadjoint operators in view of the results presented in the previous section (see (e) in remark \ref{remst}). We conclude that,  adopting von Neumann's framework, in QM observables are  
naturally described by selfadjoint operators,  whose spectra coincide with the set of values attained by the observables.

\subsection{Recovering the Hilbert space structure}
A reasonable question to ask is whether there are better reasons for choosing to describe quantum systems via a lattice of  orthogonal projectors, other than the kill-off argument ``it works''. To tackle the problem we start by listing special properties of the lattice of orthogonal projectors, whose proof is elementary.
\begin{theorem}\label{teoaggpro} 
The bounded, orthocomplemented, $\sigma$-complete lattice ${\cal L}(\cH)$ of Propositions \ref{proplat} and \ref{propclass}  satisfies these additional properties:

(i) {\bf separability} (for $\cH$ separable): if $\{P_a\}_{a\in A} \subset {\cal L}(\cH)$ satisfies 
$P_i P_j=0$, $i\neq j$, then $A$ is at most countable;

(ii) {\bf atomicity and atomisticity}:  
there exist elements in $A\in {\cal L}(\cH)\setminus \{0\}$, called {\bf atoms}, for which  $0 \leq P \leq A$ implies $P=0$ or $P=A$; for any $P\in{\cal L}(\cH)\setminus\{0\}$ there exists an atom $A$ with  $A\leq P$ (${\cal L}(\cH)$ is then called {\bf atomic}); For every  
 $P\in{\cal L}(\cH)\setminus\{0\}$, $P$ is the $\sup$ of the set of  atoms  $A\leq P$
 (${\cal L}(\cH)$ is then called {\bf atomistic});

(iii) {\bf orthomodularity}:  $P\leq Q$ implies $Q=P \vee ((\neg P) \wedge Q)$;

(iv) {\bf covering property}: if $A,P\in {\cal L}(\cH)$, with $A$ an atom, satisfy $A\wedge P = 0$, then  
(1) $P \leq A \vee P$ with $P \neq A \vee P$, and  
(2) $P \leq Q \leq A\vee P$ implies $Q=P$ or $Q= A \vee P$;

(v) {\bf irreducibility}: only $0$ and $I$ commute with every  element of ${\cal L}(\cH)$.\\
The orthogonal projectors onto one-dimensional spaces are the only atoms of ${\cal L}(\cH)$.
 \end{theorem}
 \noindent Irreducibility can easily be proved observing that if $P\in {\cal L}(\cH)$ commutes with all projectors along one-dimensional subspaces, $Px=\lambda_x x$ for every $x\in \cH$. Thus $P(x+y)= \lambda_{x+y}(x+y)$ but also $Px + Py = \lambda_x x +\lambda _yy$ and thus $(\lambda_x- \lambda_{x+y})x = (\lambda_{x+y}-\lambda_y)y$, which entails 
$\lambda_x=\lambda_y$ if $x\perp y$. If $N\subset \cH$ is a Hilbert basis, $Pz= \sum_{x\in N} \langle x|z\rangle \lambda x = \lambda z$ for some fixed $\lambda \in \bC$. Since $P=P^*=PP$, we conclude that either $\lambda =0$ or $\lambda=1$, i.e., either $P=0$ or $P=I$, as wanted.
Orthomodularity is a weaker version of distributivity of $\vee$ with respect to $\wedge$ that we know 
to be untenable in ${\cal P}(\cH)$.\\
Actually each of the listed  properties admits a physical operational interpretation (e.g. see \cite{BC}.)
So, based on the experimental evidence of quantum systems, we could try to prove, in the absence of any Hilbert space, that elementary propositions with experimental outcome in $\{0,1\}$ form a poset. 
More precisely, we could attempt to find a bounded, orthocomplemented $\sigma$-complete  lattice that verifies conditions (i)--(v) above, and then prove this lattice is described by the orthogonal projectors of a Hilbert space.\\
The partial order relation of elementary propositions can be  defined  in various ways. But  it will always  correspond to the logical implication, in some way or another. Starting from  \cite{Mackey} a number of  approaches (either of essentially physical nature, or of formal character) have been developed to this end:  
in particular, those making use of the notion of (quantum)  {\em state}, which we will see in a short while for the concrete case of propositions represented by orthogonal projectors. 
The object of the theory is now \cite{Mackey}  the pair $({\cal O}, {\cal S})$, where 
${\cal O}$ is the class of observables and ${\cal S}$ the one of states. The elementary propositions form a subclass ${\cal L}$ of ${\cal O}$ equipped with a natural poset structure $({\cal L}, \leq)$ (also satisfying a weaker version of some of the conditions (i)--(v)).
A state $s\in {\cal S}$, in particular, defines the probability $m_s(P)$ that $P$ is true for every $P\in {\cal L}$ \cite{Mackey}.
 As a matter of fact, if $P,Q \in {\cal L}$,  $P\leq Q$ means by definition that the probability  $m_s(P)\leq m_s(Q)$ for every state $s\in {\cal S}$.
 More difficult is to justify that the poset thus obtained is a lattice, i.e. that it admits a greatest lower bound $P\vee Q$ and a least upper bound $P\wedge Q$ for  every $P,Q$.  
There are several proposals, very different in nature, to introduce this lattice structure (see \cite{BC} and \cite{egl} for a general treatise) and make the physical meaning explicit in terms of measurement outcome. 
See Aerts in \cite{egl} for an abstract but operational viewpoint and \cite[\S 21.1]{BC} for a summary on  several possible ways to introduce the lattice structure on the partially ordered set of abstract elementary propositions ${\cal L}$.\\
If we accept the lattice structure on elementary propositions of a  quantum system, then we may define the operation of orthocomplementation by the familiar logical/physical negation. 
Compatible propositions can then be defined  in terms of commuting propositions, i.e. commuting 
elements of a orthocomplemented lattice as follows.  
\begin{definition}\label{defgen} {\em Let $({\cal L}, \geq, \neg)$ an orthocomplemented lattice. Two elements $a,b \in {\cal L}$ are said to be:
 
 {\bf orthogonal}
written $a\perp b$,  if $\neg a \geq b$ (or equivalently $\neg b \geq a$);

{\bf commuting}, if 
$a=c_1 \vee c_3$ and $b=c_2 \vee c_3$ with $c_i \perp c_j$ if $i\neq j$.}  \hfill $\blacksquare$
\end{definition}
\noindent These notions of orthogonality and compatibility make sense beacuse, {\em a posteriori}, they  turn out to be  the usual ones when propositions are interpreted via projectors. As the reader may easily prove, two elements $P,Q \in {\cal L}(\cH)$ are orthogonal in accordance with Definition \ref{defgen} if and only if $PQ=QP=0$ (in other words they project onto mutually orthogonal subspaces), and commute in accordance with Definition \ref{defgen} 
if and only if $PQ=QP$. (If $P= P_1+P_3$ and $Q=P_2+P_3$ where the orthogonal projectors satisfy  $P_i\perp P_j=0$ for $i \neq j$, we trivially have $PQ=QP$. If conversely, $PQ=QP$, the said decomposition arises for $P_3:= PQ$, $P_1:= P(I-Q)$, $P_2:= Q(I-P)$.)\\
Now fully-fledged with an orthocomplemented lattice and the notion of compatible propositions, we can attach a physical meaning (an interpretation backed by experimental 
evidence) to the requests that the lattice be orthocomplemented, complete, atomistic,
 irreducible and that it have the covering property \cite{BC}.  
Under these hypotheses and assuming there exist at least  $4$ pairwise-orthogonal atoms, Piron (\cite{Piron64,JP69},  \cite[\S 21]{BC}, Aerts in \cite{egl}) used  projective geometry techniques to show that the lattice of quantum propositions can be canonically identified with the closed (in a generalised sense) subsets of a generalised Hilbert space of sorts. In the latter: (a) the field is replaced by a division ring (usually not commutative) equipped with an involution, and (b) there exists a certain non-singular  Hermitian form associated with the involution. It has been conjectured by many people (see \cite{BC}) that if the lattice is also orthomodular and separable, the division ring can only be picked among $\bR,\bC$ or $\bH$ (quaternion algebra). More recently  Sol\`er\footnote{Sol\`er, M. P.:  Characterization of Hilbert spaces by orthomodular spaces. {\em Communications in Algebra}, {\bf 23}, 219-243 (1995).},
 Holland\footnote{Holland, S.S.: Orthomodularity in infinite dimensions; a theorem of M. Sol\`er. {\em Bulletin of the American Mathematical Society}, {\bf 32}, 205-234, (1995).}
and  Aerts--van Steirteghem\footnote{Aerts, D., van Steirteghem B.: Quantum Axiomatics and a theorem of M.P. Sol\'er. {\em International Journal of Theoretical Physics}. {\bf 39}, 497-502, (2000).} have found sufficient hypotheses, in terms of the existence of infinite orthogonal systems, for this to happen.
Under these hypotheses, if the ring is $\bR$ or $\bC$, we obtain precisely the lattice of orthogonal projectors of the separable  Hilbert space. In the case of $\bH$, one gets a similar generalised structure. 
In all these arguments the assumption of irreducibility is not really crucial: if property (v) fails, the lattice can be split into irreducible sublattices \cite{Jauch,BC}. Physically-speaking this situation is natural in the presence of {\em superselection rules}, of which more soon. \\
It is worth stressing that the covering property in Theorem \ref{teoaggpro} is a crucial property. Indeed there are other lattices relevant in physics verifying all the remaining properties in the afore-mentioned theorem. Remarkably the family of the so-called {\em causally closed sets} in a general spacetime  satisfies  all the said properties  but the covering one\footnote{See H. Casini, {\em The logic of causally closed spacetime subsets}, Class.Quant.Grav. {\bf 19},  2002, 6389-6404}. This obstruction prevents one from endowing a spacetime with a natural (generalized) Hilbert space structure, while it   suggests some ideas towards a formulation of quantum gravity.
\subsection{States as measures on ${\cal L}(\cH)$: Gleason's Theorem}
Let us introduce an important family of operators. This family will plays a decisive r\^ole in the issue concerning a possible justification of the fact that  quantum states are elements of the projective space $P\cH$.

\subsubsection{Trace class operators}
\begin{definition} {\em If $\cal H$ is a complex Hilbert space, $\gB_1(\cH)\subset \gB(\cH)$
denotes the set of {\bf trace class} or {\bf nuclear} operators, i.e. the operators $T\in \gB(\cH)$ satisfying
\beq \sum_{z\in N} \langle z| |T| z \rangle <+\infty \label{trace1}\eeq
for some Hilbertian basis $N \in \cH$ and where $|T|:= \sqrt{T^*T}$ defined via functional calculus.}\hfill $\blacksquare$
\end{definition}

\remark Notice that. above,  $T^*T$ is selfadjoint and $\sigma(T^*T)\in [0,+\infty)$ because of exercise \ref{espos} so that $\sqrt{T^*T}$ is well defined as a function of $T^*T$.\hfill $\blacksquare$\\

\noindent Trace class operators enjoy several remarkable properties \cite{moretti}. Here we only mention the ones relevant for these lecture notes.

\begin{proposition} Let $\cH$ a complex Hilbert space, $\gB_1(\cH)$ satisfy the following properties.\\
{\bf (a)} If $T\in \gB_1(\cH)$ and $N \subset \cH$ is any Hilbertian basis, then (\ref{trace1}) holds and thus 
$$||T||_1 := \sum_{z\in N} \langle z| |T| z \rangle $$
is well defined.\\
{\bf (b)} $\gB_1(\cH)$ is a subspace of $\gB(\cH)$
 which is moreover a two-sided $*$-ideal, namely

(i) $AT, TA \in \gB_1(\cH)$ if $T\in \gB_1(\cH)$ and $A\in \gB(\cH)$,

(ii) $T^*\in \gB_1(\cH)$ if $T\in \gB_1(\cH)$.\\
{\bf (c)} $||\:\:||_1$ is a norm on $\gB_1(\cH)$ making it a Banach space and satisfying

(i) $||TA||_1 \leq ||A||\:||T||_1$ and $||AT||_1 \leq ||A||\:||T||_1$ if $T\in \gB_1(\cH)$ and $A\in \gB(\cH)$,

(ii) $||T||_1=||T^*||_1$ if $T\in \gB_1(\cH)$.\\
{\bf (d)} If $T\in \gB_1(\cH)$, the {\bf trace} of $T$,
$$tr\: T := \sum_{z\in N} \langle z|T z\rangle \in \bC$$
is well defined, does not depend on the choice of the Hilbertian basis $N$ and the sum converges absolutely (so can be arbitrarily re-ordered). 
\end{proposition}
\remark $\null$

{\bf (1)} Obviously we have $tr \:|T| = ||T||_1$ if $T \in \gB_1(\cH)$.

{\bf (2)} The trace just possesses the properties one expects from the finite dimensional case. In particular, \cite{moretti},

 (i) it is linear on $\gB_1(\cH)$, 

(ii) $tr\: T^* = \overline{tr \: T}$ if $T\in \gB_1(\cH)$,

(iii) the trace satisfies the {\bf cyclic property},
\beq
tr(T_1\cdots T_n) = tr(T_{\pi(1)} \cdots T_{\pi(n)})
\eeq
if at least one of the $T_k$ belongs to $\gB_1(\cH)$, the remaining ones are in $\gB(\cH)$, and $\pi : \{1,\ldots, n\}\to  \{1,\ldots, n\}$ is a {\em cyclic} permutation.\hfill $\blacksquare$\\

\noindent The trace of $T \in \gB_1(\cH)$ can computed on a basis of eigenvectors in view of the following further result \cite{moretti}. Actually (d) and (e) easily follow from (a),(b),(c), (d), and the spectral theory previously developed.
\begin{proposition} Let $\cH$ a complex Hilbert space and  $T^*=T \in\gB_1(\cH)$.
The following facts hold.\\
{\bf (a)} $\sigma(T) \setminus \{0\}= \sigma_p(T)\setminus \{0\}$. If $0\in \sigma(T)$
it may be either the unique element of $\sigma_c(T)$ or an element of $\sigma_p(T)$. \\
{\bf (b)} Every eigenspace $\cH_\lambda$ has finite dimension $d_\lambda$ provided  $\lambda \neq 0$.\\
{\bf (c)} $\sigma_p(T)$ is made of at most  countable number of reals such that 

 (i) $0$ is unique possible accumulation point,

(ii)  $||T||= \max_{\lambda \in \sigma_p(T)} |\lambda|$.\\
{\bf (d)} There is a Hilbert basis of eigenvectors $\{ x_{\lambda, a}\}_{\lambda \in \sigma_p(T), a = 1,2,\ldots, d_\lambda}$ ($d_0$ may be infinite)  and
$$tr(T) = \sum_{\lambda \in \sigma_p(T)} d_\lambda \lambda\:,$$
where the sum converges absolutely (and thus can be arbitrarily re-ordered).\\
{\bf (e)} Referring to the basis presented in (d), the spectral decomposition of $T$ reads
$$T=  \sum_{\lambda \in \sigma_p(T)} \lambda P_\lambda $$
where $P_\lambda =  \sum_{a = 1,2,\ldots, d_\lambda}  \langle x_{\lambda,s}| \:\:\: \rangle x_{\lambda, a}$
and the sum is computed in the strong operator topology  and can be re-ordered arbitarily.
The convergence holds in the uniform topology too if the set of eigenspaces are suitably ordered in the count.
\end{proposition}

\begin{corollary}
$tr\:  : \gB_1(\cH)\to \bC$ is continuous with respect to the norm $||\:\:||_1$ because $|tr T| \leq tr |T| = ||T||_1$ if $T \in \gB_1(\cH)$.
\end{corollary}

\begin{proof} If $T \in \gB(\cH)$, we have the {\em polar decomposition} $T= U|T|$ (see, e.g., \cite{moretti}) where $U \in \gB(\cH)$ is isometric on $Ker(T)^\perp$ and $Ker(U)= Ker(T)= Ker (|T|)$ so that, in particular $||U||\leq 1$. Let $N$ be a Hilbertian basis of $\cH$ made of eigenvectors of $|T|$
(it exists for the previous theorem since $|T|$ is trace class).
We have
$$|tr \:T| = \left| \sum_{u\in N}\langle u|U\: |T| u \rangle\right|= \left| \sum_{u\in N}\langle u|U u \rangle \lambda_u\right|
\leq \sum_{u\in N} |\lambda_u|\:|\langle u|U u \rangle|\:.$$
Next observe that $|\lambda_u|= \lambda_u$ because $|T|\geq 0$ and $|\langle u|U u \rangle| \leq ||u||\: ||Uu|| \leq 1 ||Uu|| \leq ||u|| =1$ and thus,
$|tr \:T| \leq \sum_{u\in N} \lambda_u = \sum_{u\in N} \langle u ||T| u\rangle = tr |T| = ||T||_1$.
\end{proof}

\subsubsection{The notion of quantum state and the crucial theorem by Gleason} As commented in (a) in remark \ref{probtruble}, the probabilistic interpretation of quantum states is not well defined because there is no a true probability measure in view of the fact that there are incompatible observables. The idea is to re-define the notion of probability in the  bounded, orthocomplemented, $\sigma$-complete lattice like  ${\cal L}(\cH)$ instead of on a $\sigma$-algebra. Exactly as in CM, where the generic states 
are probability measures on Boolean lattice ${\cal B}(\Gamma)$ of the elementary properties of the system (Sect. \ref{exf}), we can think of states of a quantum system as $\sigma$-additive probability measures over the non-Boolean lattice of the elementary observables ${\cal L}(\cH)$.
\begin{definition}\label{defstates}
{\em Let $\cH$ be a complex Hilbert space. A {\bf quantum state} in $\cH$ is a map $\rho : {\cal L}(\cH) \to [0,1]$ such that the following requirements are satisfied.

(1) $\rho(I) =1$\:.

(2) If $\{P_n\}_{n \in N}\subset  {\cal L}(\cH)$, for $N$ at most countable satisfies 
$P_k(\cH)\perp P_h(\cH)$ when $h\neq k$ for  $h,k \in N$, then
\beq \rho(\vee_{k\in N}P_k) = \sum_{k \in N} \rho(P_k)\:.\label{sigmaadd}\eeq
The set of the states in $\cH$ will be denoted by $\gS(\cH)$. \hfill $\blacksquare$}
\end{definition}

\remark $\null$ 

{\bf (a)} The condition $P_k(\cH)\perp P_h(\cH)$ is obviously equivalent to $P_kP_h =0$. Since (taking the adjoint) we also obtain $P_h P_k =0= P_kP_h$, we conclude that we are dealing with pairwise compatible elementary observables. Therefore Proposition \ref{propclass} permits us to equivalently re-write the $\sigma$-additivity (2) as follows.\\
(2)  If $\{P_n\}_{n \in N}\subset  {\cal L}(\cH)$, for $N$ at most countable satisfies 
$P_k P_h =0$ when $h\neq k$ for  $h,k \in N$, then
\beq \rho\left(\sum_{k \in N}P_k\right) = \sum_{k \in N} \rho(P_k)\:,\label{sigmaadd2}\eeq
the sum on the left hand side  being computed with respect to the strong operator topology if $N$ is infinite.

{\bf (b)} Requirement (2), taking (1) into account implies $\rho(0)=0$.

{\bf (c)} Quantum states do exist. It is immediately proved that, in fact, $\psi \in \cH$ with $||\psi|| =1$ defines a quantum state $\rho_\psi$ as
\beq \rho_\psi(P) = \langle \psi|P \psi\rangle \quad P \in {\cal L}(\cH) \label{purestate}\:.\eeq
This is in nice agreement with what we already know and proves that these types of quantum states are one-to-one with the elements of $P\cal H$ as well known.\\
However these states do not exhaust $\gS(\cH)$. In fact, it immediately arises from Definition \ref{defstates} that the set of the states is convex: If $\rho_{1},\ldots, \rho_{n} \in \gS(\cH)$ then $\sum_{j=1}^n p_k\rho_k \in \gS(\cH)$ if $p_k \geq 0$ and $\sum_{k=1}^n p_k =1$. These convex combinations of states generally do not have the form $\rho_\psi$.

{\bf (d)} Restricting ourselves to a maximal set ${\cal L}_0$ of pairwise commuting projectors, which in view of Proposition \ref{propclass} has the abstract structure of a $\sigma$-algebra, a quantum state $\rho$ reduces thereon to a standard probability measure. In this sense the ``quantum probability'' we are considering extends the classical notion. Differences show up just when one deals with conditional probability  involving incompatible elementary observables. \hfill $\blacksquare$\\

\noindent An interesting case of (c) in the remark above is a convex combination of states induced by unit vectors as in (\ref{purestate}), where $\langle \psi_k|\psi_h\rangle = \delta_{hk}$,
$$\rho = \sum_{k=1}^n p_k \rho_{\psi_k}\:.$$
By direct inspection, completing the finite orthonormal system $\{\psi_k\}_{k =1,\ldots,n}$  to a full Hilbertian basis of $\cH$, one quickly proves that,  defining
\beq T = \sum_{k=1}^n p_k \langle \psi_k|\:\: \rangle \psi_k\label{T}\eeq
$\rho(P)$ can be computed as 
$$\rho(P) = tr(T P) \quad P \in {\cal L}(\cH)$$
In particular it turns out that
$T$  is in $\gB_1(\cH)$, satisfies 
$T\geq 0$ (so it is selfadjoint for (3) in exercise \ref{vex}) and $tr \: T =1$. 
As a matter of fact, (\ref{T}) is just the spectral decomposition of $T$, whose spectrum is 
$\{p_k\}_{k=1,\ldots,n}$.
This result is general \cite{moretti}
\begin{proposition}\label{proppregleason} Let $\cH$ be a complex Hilbert space and let $T \in \gB_1(\cH)$ satisfy
$T\geq 0$ and $Tr\: T =1$, then the map 
$$\rho_T : {\cal L}(\cH) \ni P \mapsto tr(T P)$$
is well defined and $\rho_T \in \gS(\cH)$.
\end{proposition}

\noindent The very remarkable fact  is that these operators exhaust $\gS(\cH)$
 if $\cH$ is separable with dimension $\neq 2$, as established by Gleason in a celebrated theorem we restate re-adapting it  to these lecture notes (see \cite{moretti} for a the original statement and \cite{gleasonbook} for a general treatise on the subject).

\begin{theorem}[Gleason's Theorem]\label{Gleasontheorem}
Let $\cH$ be a complex Hilbert space of finite dimension $\neq	 2$, or infinite-dimensional and  separable. If $\rho \in \gS(\cH)$ there exists a unique operator $T \in \gB_1(\cH)$ with $T\geq 0$
and $tr \: T =1$ such that $tr(T P) = \rho(P)$ for every $P \in {\cal L}(\cH)$.
\end{theorem}

\noindent Concerning the existence of $T$, Gleason's proof works for real Hilbert spaces too. If the Hilbert space is complex, the operator $T$ associated to  $\rho$ is unique for the following reason. Any other $T'$ of trace class such that $\rho(P)=tr(T'P)$ 
for any $P\in {\cal L}(\cH)$ must also satisfy $\langle x|(T-T')x\rangle =0$ for any $x\in \cH$. 
If $x=0$ this is clear, while if $x\neq 0$ we may complete the vector $x/||x||$ to a basis, in which 
$tr((T-T')P_x)=0$ reads $||x||^{-2}\langle x|(T-T')x\rangle =0$, where $P_x$ is the projector onto $span(x)$.
By  (3) in exercise \ref{vex}, we obtain $T-T'=0$\footnote{In a real Hilbert space $\langle x|Ax\rangle=0$ for all $x$ does not imply $A=0$. Think of real anti symmetric matrices in $\bR^n$ equipped with the standard scalar product.}.

\remark\label{ossgelason} $\null$

{\bf (a)} Imposing $\dim\cH \neq 2 $ is mandatory,  a well known counterexample can be found, e.g. in \cite{moretti}.

{\bf (b)} Particles with 
 spin $1/2$, like electrons, admit a Hilbert space -- in which the observable spin is defined -- of dimension $2$.  
 The same occurs to the Hilbert space in which the polarisation of light is described (cf. helicity of photons).  When these systems are described in full, however, for instance including  degrees of freedom relative to position or momentum, they are representable on a separable Hilbert space of  infinite dimension.

{\bf (c)} Gleason's characterization of states has an important consequence known asthe  {\em Kochen-Specker theorem}. It proves that in QM there are no states assigning probability $1$
to some  elementary observables and $0$ to the remaining ones, differently to what happens in CM. 
\begin{theorem}[Kochen-Specker Theorem] Let $\cH$ be a complex Hilbert space of finite dimension $\neq 2$, or infinite-dimensional and  separable. There is no quantum state $\rho : {\cal L}(\cH) \to [0,1]$, in the sense of Def. \ref{defstates}, such that
$\rho({\cal L}(\cH)) = \{0,1\}$
\end{theorem}

\begin{proof} Define $\bS := \{x\in \cH \:|\: ||x||=1\}$ endowed with the topology induced by $\cH$, and let $T\in \gB_1(\cH)$ be the  representative of $\rho$ using  Gleason's theorem. The map $f_\rho: \bS \ni x \mapsto \langle x|Tx\rangle = \rho(\langle x| \:\: \rangle x) \in \bC$ is continuous because $T$ is bounded. We have $f_\rho(\bS) \subset \{0,1\}$, where $\{0,1\}$ is equipped with the topology induced by $\bC$. Since $\bS$ is connected its image must be connected also. So either $f_\rho(\bS)=\{0\}$ or $f_\rho(\bS)=\{1\}$. In the first case $T=0$ which is impossible because $tr T =1$, in the second case $tr T \neq 2$ which is similarly impossible.
\end{proof}
\noindent This negative  result produces no-go theorems in some attempts to explain QM in terms of CM introducing {\em hidden variables} \cite{stanford}.\hfill $\blacksquare$\\

\remark {\em In view of Proposition \ref{proppregleason} and 
 Theorem  \ref{Gleasontheorem}, assuming that $\cH$ has finite dimension or is separable,  we henceforth  identify $\gS(\cH)$ with the subset of $\gB_1(\cH)$ of positive operators with unit trace. We simply disregard the states in $\cH$ with dimension $2$ which are not of this form especially taking (b) in remark \ref{ossgelason} into account.} \hfill $\blacksquare$\\

\noindent We are in a position  to state some definitions of interest for physicists, especially the distinction between pure and mixed states, so we proceed  to analyse the structure of the space of the states. To this end,
 we remind the reader that, if  $C$ is a convex set in a vector space,  $e\in C$ is called  {\bf extreme} if it cannot be written as  $e= \lambda x + (1-\lambda)y$, with $\lambda \in (0,1)$, $x,y \in C\setminus\{e\}$.\\
We have the following simple result whose proof can be found in \cite{moretti}.

\begin{proposition}\label{extremalstates} Let $\cH$ be a complex separable Hilbert space.\\
{\bf (a)} $\gS(\cH)$ is a convex closed subset in $\gB_1(\cH)$ whose extreme points are those of the form:
$\rho_\psi := \langle \psi|\:\:\rangle\psi$ for every vector $\psi\in \cH$ with $||\psi||=1$.
(This sets up a bijection between extreme states and elements of $P\cH$.)\\
{\bf (b)} A state $\rho \in \gS(\cH)$  is extreme if and only if $\rho\rho =\rho$.
(All the elements of $\gS(\cH)$ however satisfy $\langle x| \rho \rho x \rangle \leq \langle x| \rho x \rangle$ for all $x\in \cH$.)\\
{\bf (c)} Any state $\rho \in \gS(\cH)$ is a linear combination of extreme states, including infinite combinations in the strong operator topology. In particular there is always a decomposition
$$\rho = \sum_{\phi\in N} p_\phi \langle \phi|\:\:\rangle \phi,$$
where $N$ is an  eigenvector basis for $\rho$, $p_\phi\in [0,1]$ for any $\phi\in N$, and  $$\sum_{\phi\in N}p_\phi=1\:.$$
\end{proposition}

\noindent The stated proposition allows us to introduce some notions and terminology relevant in physics.
First of all, extreme elements in $\gS(\cH)$ are usually called {\bf pure states} by physicists. We shall denote their set is denoted $\gS_p(\cH)$. Non-extreme states are instead called
 {\bf mixed states}, {\bf mixtures} or {\bf non-pure states}.
If  $$\psi = \sum_{i\in I}a_i\phi_i\:,$$ with $I$ finite or countable 
(and the series converges in the topology of $\cH$ in the second case),
 where the vectors $\phi_i\in \cH$ are all non-null and $0\neq a_i\in \bC$, physicists  say that  the  state 
$\langle \psi|\:\: \rangle\psi$ is called  an {\bf coherent superposition} of the states 
 $\langle \phi_i|\:\: \rangle\phi_i/||\phi_i||^2$.\\
The possibility of creating pure states by non-trivial combinations of 
vectors associated to other pure states  is called, in the jargon of QM, 
{\bf superposition principle of  (pure) states}\\
There is however another type of superposition of states. If  $\rho\in \gS(\cH)$ satisfies:
$$\rho = \sum_{i\in I} p_i \rho_i$$
with $I$ finite,  $\rho_i \in \gS(\cH)$, $0\neq p_i\in [0,1]$ for any $i\in I$, and $\sum_i p_i=1$, 
the state $\rho$ is called {\bf incoherent superposition}
of states $\rho_i$ (possibly pure).\\
If $\psi,\phi\in \cH$ satisfy $||\psi||=||\phi|| =1$ the following terminology is very popular: The complex number $\langle\psi|\phi\rangle$ is the {\bf transition amplitude}
or {\bf probability amplitude} of the state $\langle \phi|\:\:\rangle\phi$ 
on the state $\langle \psi|\:\:\rangle \psi$, moreover
the non-negative real number $|\langle \psi|\phi\rangle|^2$ is the  {\bf transition probability} of the state $\langle \phi|\:\: \rangle\phi$ on the state $\langle\psi|\:\:\rangle \psi$.\\
We make some comments about these notions. 
Consider the pure state $\rho_\psi \in \gS_p(\cH)$, written $\rho_\psi = \langle \psi|\:\:\rangle \psi$ for some  $\psi \in \cH$ with $||\psi || =1$. What we want to emphasise is that this pure state is also an orthogonal projector $P_\psi:= \langle \psi|\:\:\rangle \psi$, so it must correspond to an elementary observable of the system (an {\em atom} using the terminology of Theorem \ref{teoaggpro}).
 The na\"ive and natural interpretation\footnote{We cannot but notice how this interpretation  muddles the semantic and syntactic levels. Although this could be problematic in a formulation within formal logic, the use physicists make of the interpretation eschews the issue.} 
of that observable is this: {\em``the system's state is the pure state given by the vector $\psi$''.}
We can therefore  interpret the square modulus of the transition amplitude 
$\langle \phi|\psi\rangle$ as follows. If $||\phi||=||\psi||=1$, as the definition of transition amplitude imposes, 
$tr(\rho_\psi P_\phi) = |\langle \phi|\psi\rangle|^2$, where $\rho_\psi := \langle\psi|\:\:\rangle \psi$ and $P_\phi =\langle \phi|\:\:\rangle \phi$.
Using (4) we conclude:  \\
{\em $|\langle \phi|\psi\rangle|^2$ is the probability that the state, given (at time $t$) by the vector $\psi$, following a measurement (at time $t$) on the system becomes determined by $\phi$.}\\
Notice  $|\langle \phi|\psi\rangle|^2= |\langle \psi|\phi\rangle|^2$, so the probability transition of the state determined by  $\psi$  on the state determined by $\phi$ coincides with the analogous probability where the vectors  are swapped. This fact is, {\it a priori}, highly non-evident in physics.\\

\noindent Since we have introduced a new notion of state the axiom concerning the collapse of the state (Sect. \ref{formalism}) must be improved in order to encompass all states of $\gS(\cH)$. 
The standard formulation of QM assumes the following axiom (introduced by von Neumann and generalised by L\"uders) about what occurs to the physical system, in state $\rho \in \gS(\cH)$ at time $t$, when subjected to the measurement of an elementary observable 
$P\in {\cal L}(\cH)$, if the latter is true (so in particular $tr (\rho P) >0$, prior to the measurement). We are  referring to {\em non-destructive} testing, also known as {\em indirect measurement} or {\em first-kind measurement}, where the physical system examined (typically a particle) is not absorbed/annihilated by the instrument. They are idealised versions of the actual processes used in labs, and only in part they can be modelled in such a way. \\

\noindent {\bf Collapse of the state revisited}. If the quantum system  is in state $\rho \in \gS(\cH)$ at time $t$ and  proposition $P\in  {\cal L}(\cH)$ is true after a measurement at time $t$, the system's state immediately afterwards is: 
$$\rho_P := \frac{P\rho P}{tr(\rho P)}\:.$$
In particular, if $\rho$ is pure and determined by the unit vector $\psi$, the state immediately after measurement is still pure, and determined by:
$$\psi_P =\frac{P\psi}{||P\psi||}\:.$$

\noindent Obviously, in either case $\rho_P$ and $\psi_P$ define states. In the former, in fact,  $\rho_P$ is  positive of trace class, with unit trace, while in the latter  $||\psi_P||=1$.\\

\remark $\null$

{\bf (a)}  Measuring a property of a physical quantity goes through the interaction between the system and an instrument (supposed to be  macroscopic and obeying the laws of classical physics).  Quantum Mechanics, in its standard formulation, does not establish what a measuring instrument is, it only says they exist; nor is it capable of describing  the interaction of instrument and quantum system set out in the  von Neumann L\"uders' postulate quoted above. Several viewpoints and conjectures exist on how to complete the physical description of the measuring process; these are called, in the slang of QM, {\bf collapse}, or {\bf reduction},  {\bf of the state} or {\bf of the wavefunction} (see \cite{moretti} for references).

{\bf (b)} Measuring instruments are commonly employed to {\em prepare a system in a certain pure state}. Theoretically-speaking the preparation of a {\em pure} state\index{preparation of system in a pure state} is carried out like this. A finite collection of  {\em compatible}  propositions
$P_1,\ldots, P_n$ is chosen so that the projection subspace  of 
$P_1 \wedge \cdots \wedge P_n = P_1\cdots P_n$ is {\em one-dimensional}. In other words $P_1\cdots P_n = (\psi|\:\:)\psi$ for some vector with $||\psi||=1$. The existence of such propositions 
is seen in practically all quantum systems used in experiments. (From a theoretical point of view these are {\em atomic} propositions)
Then propositions $P_i$ are simultaneously measured  on several identical copies of the physical system of concern (e.g., electrons), whose initial states, though,  are unknown. 
If for one system the measurements of all propositions are successful,  the post-measurement state is determined by the vector $\psi$, and the system was 
{\bf prepared} in that particular pure state.\\
Normally each projector $P_i$ belongs to the PVM $P^{(A)}$ of an observable  $A_i$ 
whose spectrum is made of isolated points (thus a pure point spectrum) and $P_i = P^{(A)}_{\{\lambda_i\}}$ with $\lambda_i \in \sigma_p(A_i)$.

{\bf (c)} Let us finally explain how to practically obtain non-pure states from pure ones. Consider $q_1$ identical copies of system  $S$ prepared in the pure state associated to $\psi_1$, $q_2$ copies of $S$   prepared in the pure state associated to $\psi_2$ and so on, up to $\psi_n$. If we mix these states each one will be in the non-pure state: 
$\rho = \sum_{i=1}^n p_i \langle \psi_i|\:\:\rangle \psi_i\:,$
where $p_i := q_i/\sum_{i=1}^n q_i$. In general, $\langle \psi_i|\psi_j\rangle$ is not zero if $i\neq j$, so the above expression for $\rho$ is not the decomposition with respect to an eigenvector basis for $\rho$.  This procedure hints at the existence of two different types of probability, one intrinsic and due to the quantum nature of state $\psi_i$, the other epistemic, and  encoded in the probability $p_i$. But this is not true: once a non-pure state has been created, as above, there is no way, within QM, to distinguish the states forming the mixture. For example, the same $\rho$ could have been obtained mixing other  pure states than those determined by the $\psi_i$. In particular, one could have used those in the decomposition of $\rho$ into a basis of its eigenvectors. For physics, no kind of measurement  would distinguish the two mixtures.
\hfill $\blacksquare$\\

\noindent Another delicate point is that, dealing with mixed states, definitions
(\ref{defexpt}) and (\ref{Delta22}) for, respectively the expectation value $\langle A\rangle_\psi$ and the standard deviation $\Delta A_\psi$ of an observable $A$ referred to  the pure state $\langle \psi| \:\: \rangle \psi$ with $||\psi||=1$ are no longer  valid.
We just say that extended natural definitions can be stated referring to the probability measure associated to both the mixed state $\rho \in \gB_1(\cH)$ (with $\rho\geq 0$ and tr $\:\rho =1$) and the observable,
$$\mu_\rho^{(A)} : {\cal B}(\bR) \ni E \mapsto tr(\rho P^{(A)}_E)\:.$$
We refer the reader to \cite{moretti} for a technical discussion on these topics.
\subsection{von Neumann algebra of observables, superselection rules}
The aim of this section is to focus on the class of observables of a quantum system, described in the complex Hilbert space $\cH$, exploiting some elementary results  of the theory of {\em von Neuman algebras}. Up to now, we have 
tacitly supposed that {\em all} selfadjoint operators in $\cH$ represent observables,  {\em all} orthogonal projectors 
represent elementary observables, {\em all} normalized vectors represent pure states. This is not the case in 
physics due to the presence of the so-called {\em superselection rules}. Within the Hilbert space approach 
the modern  tool to deal with this notion is the mathematical structure of a {\em von Neumann algebra}. For this reason 
we spend the initial part of this section to introduce this mathematical tool.

\subsubsection{von Neumann algebras} Before we introduce it, let us define first the {\em commutant} of an operator algebra and state an important preliminary theorem. If $\gM\subset \gB(\cH)$ is a subset in the algebra of bounded operators on the complex Hilbert space $\gB(\cH)$, the {\bf commutant} of $\gM$ is:
\beq \gM' := \{T \in \gB(\cH)\:\: | \: \: TA-AT =0\quad \mbox{for any $A
\in \gM$}\}\:.\label{commutant}\eeq
If $\gM$ is closed under the adjoint operation  (i.e. $A^* \in \gM$ if $A \in \gM$) the 
commutant $\gM'$ is certaintly  a  $^*$-algebra with unit. In general: $\gM_1' \subset \gM_2'$
if $\gM_2 \subset \gM_1$ and $\gM\subset (\gM')'$, which imply $\gM' = ((\gM')')'$. 
Hence we cannot reach beyond the second commutant by iteration.\\
The  continuity of the product  of operators in the uniform topology says that the  commutant $\gM'$ is closed in the 
uniform topology, so if $\gM$ is closed under the adjoint operation, its commutant 
$\gM'$ is a $C^*$-algebra ($C^*$-subalgebra) in $\gB(\cH)$.\\
$\gM'$ has other pivotal topological properties in this general setup. It is easy to prove that $\gM'$ is both strongly and weakly closed. This holds, despite the product of operators is not continuous with respect to the strong operator topology, because separate continuity in each variable is sufficient. \\
In the sequel we shall adopt the standard convention used for von Neumann algebras and write  $\gM''$ in place of $(\gM')'$ {\it etc}. The next crucial result is due to von Neumann (see e.g. \cite{moretti}).

\begin{theorem} [von Neumann's double commutant theorem] \label{teoDC}
If  $\cH$  is a complex Hilbert space and $\gA$ a unital $^*$-subalgebra in $\gB(\cH)$, the following statements are equivalent.\\
{\bf (a)}  $\gA = \gA''$.\\
{\bf (b)}  $\gA$ is weakly closed.\\
{\bf (c)}  $\gA$ is strongly closed.
\end{theorem}

\noindent At this juncture we are ready to define von Neumann algebras.
\begin{definition}\label{defAvN} {\em Let $\cH$ be a complex Hilbert space.  A {\bf von Neumann algebra}  in $\gB(\cH)$ is a $^*$-subalgebra of $\gB(\cH)$, with unit, that satisfies any of the   equivalent properties  appearing in von Neumann's theorem \ref{teoDC}.} \hfill $\blacksquare$
\end{definition}
\noindent In particular $\gM'$ is a von Neumann algebra provided $\gM$ is a  $^*$-closed subset of $\gB(\cH)$, because $(\gM')'' = \gM'$ as we saw above. 
Note how, by construction, a von Neumann algebra in $\gB(\cH)$ is a $C^*$-algebra with unit, or better, a $C^*$-subalgebra with unit of $\gB(\cH)$.\\
It is not hard to see that the intersection of  von Neumann algebras is a von Neumann algebra.
If $\gM \subset \gB(\cH)$ is closed under the adjoint operation,  $\gM''$ turns out to be the smallest (set-theoretically) von Neumann algebra containing $\gM$ as a subset \cite{BrRo}. Thus $\gM''$ is called the {\bf von Neumann algebra generated}
by $\gM$.\\

\noindent Since in QM it is natural to deal with unbounded selfadjoint operators, the  definition of commutant is extended to the case of a set of generally unbounded selfadjoint operators, exploiting the fact that these operators admit spectral measures made of bounded operators. 
\begin{definition}
{\em If $\gN$ is a set of (generally unbounded) selfadjoint operators in the complex Hilbert space $\cH$, the {\bf commutant} $\gN'$ of $\gN$, is defined as the commutant in the sense of (\ref{commutant}) of the set of  all the spectral measures $P^{(A)}$ of every $A \in \gN$. \\
The von Neuman algebra $\gN''$ generated by $\gN$ is defined as  $(\gN')'$, where the external prime is the one of definition (\ref{commutant}).} \hfill $\blacksquare$
\end{definition}

\remark Notice that, if the selfadjoint operators are all bounded,  
 $\gN'$ obtained this way coincides with the one already defined  in (\ref{commutant}) as a consequence of of (ii) and (iv) of Proposition \ref{propcomm} (for a bounded selfadjoint operator $A$).
Thus $\gN'$ is well-defined and gives rises to a von Neumann algebra because the set of spectral measures is $^*$-closed.  $\gN''$ is a von Neumann algebra too for the same reason.
 \hfill $\blacksquare$\\

\noindent We are in a position to state a technically important result which concerns both the spectral theory and the notion of von Neumann algebra \cite{moretti}.

\begin{proposition}\label{propBeCaagg} Let $\gN= \{A_1, \ldots, A_n\}$ be a finite collection of  self-adjoint operators in the separable Hilbert space $\cH$ whose spectral measures commute. The von Neumann algebra  $\gN''$ coincides with the collection of operators $$f(A_1,\ldots,A_n) := \int_{supp(P^{({\bf A})})} f(x_1,\ldots,x_n) dP^{({\bf A})}\:,$$ with $f : supp(P^{({\bf A})}) \to \bC$ measurable and bounded.
\end{proposition}

\subsubsection{Lattices of von Neumann algebras} 
To conclude this elementary  mathematical survey, we will say some words about von Neumann algebras and their associated lattices of orthogonal projectors. \\
Consider a von Neumann algebra $\gR$  on the complex Hilbert space $\cH$.  It is easy to prove that the set ${\cal L}_\gR(\cH)\subset \gR$ of the orthogonal projectors included in $\gR$ form a lattice, which is bounded by $0$ and $I$, orthocomplemented with respect to the orthocomplementation operation of  ${\cal L}(\cH)$ and $\sigma$-complete (because this notion involves only the strong topology ((iii) in Proposition \ref{proplat}) and $\gR$ is closed with respect to  that topology in view of Theorem \ref{teoDC}. Moreover ${\cal L}_\gR(\cH)$ is orthomodular,  and separable like 
the whole ${\cal L}(\cH)$, assuming that $\cH$ is separable. It is interesting to note that, as expected, ${\cal L}_\gR(\cH)$ contains all information about $\gR$ itself since the following result holds.

\begin{proposition}\label{propLR}
Let $\gR$ be a von Neumann algebra on the complex Hilbert space $\cH$ and consider the lattice ${\cal L}_\gR(\cH)\subset \gR$ of the orthogonal projectors in $\gR$.
then the equality ${\cal L}_\gR(\cH)'' = \gR$ holds.
\end{proposition}

\begin{proof} 
Since ${\cal L}_\gR(\cH) \subset \gR$, we have ${\cal L}_\gR(\cH)' \supset \gR'$
and ${\cal L}_\gR(\cH)'' \subset \gR''= \gR$. Let us prove the other inclusion.
 $A \in \gR$ can always be decomposed as a linear combination of two self adjoint 
operators of $\gR$, $A+A^*$ and $i(A-A^*)$. So we can restrict ourselves  to the case of $A^*=A \in \gR$, proving that $A \in {\cal L}_\gR(\cH)''$ if $A\in \gR$. The PVM of $A$  belongs to $\gR$
because of (ii) and (iv) of Proposition  \ref{propcomm}: $P^{(A)}$ commutes with every bounded operator $B$ which commutes with 
$A$. So $P^{(A)}$ commutes, in particular, with the elements of $\gR'$ because  $\gR \ni A$. We conclude that every $P_E^{(A)} \in \gR'' = \gR$. 
Finally, there is a sequence of simple functions $s_n$ uniformly converging to $id$ in a compact $[-a,a] \supset \sigma(A)$ (e.g, see \cite{moretti}). By construction $\int_{\sigma(A)} s_n dP^{(A)} \in 
{\cal L}_\gR(\cH)''$ because it is a linear combination of elements of $P^{(A)}$ and ${\cal L}_\gR(\cH)''$ is a linear space. Finally $\int_{\sigma(A)} s_n dP^{(A)}\to A$ for $n\to +\infty$ uniformly, and thus strongly, as seen in (2) of example \ref{exstrong}.
 Since ${\cal L}_\gR(\cH)''$ is closed with respect to the strong topology, we must have  $A \in {\cal L}_\gR(\cH)''$, proving that ${\cal L}_\gR(\cH) \supset \gR$ as wanted.
\end{proof}

\subsubsection{General algebra of observables and its centre} Let us pass to physics and we apply these notions  and results. 
 Relaxing the hypothesis that all selfadjoint operators in the separable Hilbert space $\cH$ associated to a quantum system represent observables, there are many reasons to assume that the observables of a quantum system are represented (in the sense we are going to illustrate)  by the selfadjoint elements of an algebra of von Neumann, we hereafter indicated by $\gR$,  called the {\bf von Neumann algebra of observables} (though only the selfadjoint elements are observables).
Including non-selfadjoint elements $B \in \gR$ is armless,  as they can always be one-to-one decomposed into a pair of selfadjoint elements $$B= B_1+iB_2 =\frac{1}{2}(B+B^*) + i \frac{1}{2i}(B-B^*)\:.$$
The fact that the elements of $\gR$ are bounded does not seem a physical problem. If $A=A^*$ is unbounded and represents an observable it does not belong to $\gR$. Nevertheless the associated {\em class} of bounded  selfadjoint operators  
$\{A_n\}_{n \in \bN}$ where
$$A_n := \int_{[-n,n]\cap \sigma(A)} \lambda dP^{(A)}(\lambda)\:,$$
embodies the same information as $A$ itself. 
$A_n$ is bounded due to Proposition \ref{propboundob} because the support of its spectral measures is included in $[-n,n]$.
Physically speaking, we can say that $A_n$ is nothing but  the observable $A$ when it is measured with an instrument unable to produce outcomes larger than $[-n,n]$. All real  measurement instruments are similarly limited. We can safely assume that every $A_n$ belongs to $\gR$. Mathematically speaking, the whole (unbounded)  observable  $A$ is recovered as the limit in the {\em strong operator topology} $A= \lim_{n\to +\infty} A_n$
((1) in examples \ref{exstrong}). Moreover the union of the spectral measures of all the $A_n$ is that of $A$. Finally the spectral measure of $A$ belongs to $\gR$ since the spectral measure of every $A_n\in \gR$ does, as has been established in the proof of Proposition \ref{propLR} above.\\
Within this framework the orthogonal projectors $P \in \gR$ represent all elementary observables 
of the system. The lattice of these projectors, ${\cal L}_\gR(\cH)$, encompass the amount of  information about observables as established Proposition \ref{propLR}.
As said above  ${\cal L}_\gR(\cH)\subset \gR$ is bounded, orthocomplemented,  $\sigma$-complete, orthomodular and separable like
the whole ${\cal L}(\cH)$ (assuming that $\cH$ is separable) but there is no guarantee for the validity of the other 
properties listed in Theorem \ref{teoaggpro}. The natural question is whether $\gR$
is $^*$-isomorphic to  $\gB(\cH_1)$ for a suitable complex Hilbert space $\cH_1$, which would automatically imply that also the remaining properties were true. In particular there would exist atomic elements in ${\cal L}_\gR(\cH)$ and the covering property would be satisfied. A necessary condition is that, exactly as it happens for $\gB(\cH_1)$, there are no non-trivial elements in  $\gR\cap \gR'$, since $\gB(\cH_1)\cap \gB(\cH_1)' = \gB(\cH_1)' = \{cI\}_{c \in \bC}$.
\begin{definition} A von Neumann algebra $\gR$ is a {\bf factor}  when its {\bf centre}, the subset $\gR \cap \gR'$ of elements commuting with the whole algebra, is trivial: $\gR \cap \gR' =\{c I\}_{c\in \bC}$. \hfill $\blacksquare$
\end{definition}
\noindent 
\remark It is possible to prove that a von Neumann algebra is always a direct sum or a direct integral of factors. Therefore factors play a crucial role.
The classification of factors, started by von Neumann and Murray, is one of the key chapters in the theory of operator algebras, and has enormous consequences in the algebraic theory of quantum fields.  The factors isomorphic to  $\gB(\cH_1)$ for some complex Hilbert space $\cH_1$,  are called of {\em type $I$}. These factors admit atoms, fulfil the covering property (orthomodularity and irreducibility are always true). Regarding separability, it depends on separability of $\cH_1$ and requires a finer classification in factors of {\em type $I_n$} where $n$ is a cardinal number. There are however factors of type $II$ and $III$ which do not admit atoms and are not important in elementary QM.  \hfill $\blacksquare$\\

\noindent The centre of the von Neumann algebra of observables enters the physical theory in a nice way. A common situation dealing with quantum systems is the existence of a {\bf maximal set of compatible observables}, i.e. a finite maximal class $\gA= \{ A_1, \ldots, A_n\}$ of pairwise compatible observables. The notion of maximality here means  that, if a (bounded) selfadjoint operator commutes  with all the observables in $\gA$, then it is a {\em function} of them. In perticular it is an observable as well.
In view of proposition \ref{propBeCaagg} the existence of a maximal set of compatibel observables is equivalent to say that there is a finite set of  observables $\gA$ such  that $\gA' =\gA''$. 
We have the following important consequence 
\begin{proposition} If a  quantum physical system admits a maximal set of compatibel observables, then the commutant $\gR'$ of the von Neumann algebra of observables $\gR$ is Abelian and coincides with the center  of $\gR$. 
\end{proposition}

\begin{proof}
As the spectral measures of each $A \in \gA$ belong to $\gR$, it must be (i) $\gA'' \subset \gR$.
Since $\gA' =\gA''$, (i) yields $\gA' \subset \gR$ and thus, taking the commutant, (ii) $\gA'' \supset \gR'$. Comparing (i) and (ii) we have $\gR' \subset \gR$. In other words $\gR' = \gR' \cap \gR$. In particular, $\gR'$ must be Abelian.
\end{proof}

\example $\null$\\
{\bf (1)} Considering a quantum particle without spin and referring to the rest space $\bR^3$ of an inertial reference frame, $\cH = L^2(\bR^3, d^3x)$. A maximal set of compatible observables 
is the set of the three position operators $\gA_1= \{X_1,X_2,X_3\}$ or the the set of the three momenta  operators $\gA_2= \{P_1,P_2,P_3\}$. $\gR$ is the von Neumann algebra generated by $\gA_1\cup \gA_2$. It is possible to prove that the commutant (which coincides with the centre) of this von Neumann algebra is trivial (as it includes a unitary irreducible representation of the Weyl-Heisenberg group) so that $\gR= \gB(\cH)$
(see also Theorem \ref{SNM}).\\
{\bf (2)} If adding the spin space (for instance dealing with an electron ``without charge''), we have $\cH = L^2(\bR^3, d^3x)\otimes \bC^2$. Referring to (\ref{spin}) a maximal set of compatible observables is, for instance, $\gA_1= \{X_1\otimes I ,X_2\otimes I ,X_3\otimes I, I\otimes S_z\}$, another is $\gA_2= \{P_1\otimes I ,P_2\otimes I ,P_3\otimes I, I\otimes S_x\}$. As before $(\gA_1\cup \gA_2)''$ is the von Neumann algebra of observables of the system (changing the component of the spin passing from $\gA_1$ to $\gA_2$ is crucial for  this result).
 Also in this case, it turns out that the commutant of the von Neumann algebra of observables is trivial yielding  $\gR= \gB(\cH)$.\hfill $\blacksquare$

\subsubsection{Superselection charges and coherent sectors} We must have accumulated enough  formalism to successfully investigate the structure of the Hilbert space (always supposed to be separable) and the algebra of the observables when  not all selfadjoint operators represent observables and not all orthogonal projectors are intepreted as  elementary observables. 
Re-adapting the approach by Wightman \cite{wightman} to our framework, 
we make two assumptions generally describing the so called {\em superselection rules} for QM formulated in a (separable) Hilbert space where $\gR$ denotes the von Neumann algebra of observables.\\

 {\bf (SS1)} There is a maximal set of compatible observables in $\gR$, so that  $\gR'= \gR' \cap \gR$.\\

{\bf (SS2)}  $\gR'\cap \gR$ contains a finite class of observables  $\gQ =\{Q_1,\ldots, Q_n\}$, with  $\sigma(Q_k)= \sigma_p(Q_k)$, $k=1,2,\ldots, n$, generating the centre:
$\gQ'' \supset \gR'\cap \gR$.\\
(If the $Q_k$ are unbounded, $\gQ\subset \gR'\cap \gR$ means that the PVM of the $Q_j$ are included in $\gR'\cap \gR$.)\\

\noindent The $Q_k$ are called {\bf superselection charges}.\\

\noindent As the reader can easily prove, the joint spectral measure $P^{(\gQ)}$ in $\bR^n$ has support given exactly by $\times_{k=1}^n \sigma_p(Q_k)$ and, if $E \subset \bR^n$,
\beq P^{(\gQ)}_{E} = \sum_{(q_1,\ldots, q_n) \in \times_{k=1}^n \sigma_p(Q_k) \cap E} P^{(Q_1)}_{\{q_1\}}\cdots  P^{(Q_n)}_{\{q_n\}}\label{gendec} \eeq
We have the following remarkable result where we occasionally adopt the notation ${\bf q} := (q_1,\ldots, q_n)$ and $\sigma(\gQ):= \times_{k=1}^n \sigma_p(Q_k)$.
\begin{proposition}
Let $\cH$ be a complex separable Hilbert and suppose that the von Neumann algebra $\gR$ in $\cH$ satisfies {\bf(SS1)} and {\bf (SS2)}. The following facts hold.\\
{\bf (a)} $\cH$ admits the following direct decomposition into closed pairwise orthogonal subspaces, called {\bf superselection sectors} or {\bf coherent sectors},
\beq \cH = \bigoplus_{{\bf q}\in \sigma(\gQ)} \cH_{{\bf q}} \label{bigdec}\eeq
where 
$$\cH_{{\bf q}} := P^{(\gQ)}_{{\bf q}}\cH\:.$$
and each  $\cH_{{\bf q}}$ is invariant and irreducible under $\gR$.\\
{\bf (b)} An analogous direct decomposition occurs for $\gR$.
\beq \gR = \bigoplus_{{\bf q}\in \sigma(\gQ)} \gR_{{\bf q}} \label{bigdecA}\eeq
where 
$$\gR_{{\bf q}} := \left\{ \left. A|_{\cH_{{\bf q}}}\:\right|\: A \in \gR\right\}$$
is a von Neumann algebra on $\cH_{{\bf q}}$ considered as Hilbert space in its own right. Finally,
$$\gR_{{\bf q}} = \gB(\cH_{{\bf q}})$$
{\bf (c)} Each map $$\gR \ni A \mapsto  A|_{\cH_{{\bf q}}} \in \gR_{{\bf q}}$$
is a $^*$-algebra representation of $\gR$ (Def.\ref{defrap}). Representations associated with 
different values of ${\bf q}$ are (unfithful and) unitarily inequivalent: In other words 
there is no isometric surjective map 
$U: \cH_{{\bf q}}\to \cH_{{\bf q}'}$ such that 
$$U A|_{\cH_{{\bf q}}} U^{-1} = A|_{\cH_{{\bf q}'}}$$
when ${\bf q}\neq {\bf q}'$
\end{proposition}

\begin{proof}
(a)  Since 
$P^{(\gQ)}_{{\bf q}} P^{(\gQ)}_{{\bf s}} =0$ if ${\bf q} \neq {\bf s}$
and
$\sum_{{\bf q}\in \sigma_p(\gQ) }P^{(\gQ)}_{{\bf q}}=I$,
$\cH$ decomposes as in (\ref{bigdec}). Since  $P^{(\gQ)}_{{\bf q}}$ belongs to the centre of $\gR$, the subspaces of the decomposition are invariant under the action of each element of $\gR$. Let us pass to the irreducibility.
If $P \in \gR'\cap \gR$ is an orthogonal projector it must be a function of the $Q_k$ by hypotheses:
$P = \int_{\bR^n} f(x) dP^{(\gQ)}(x)$
since $P=PP\geq 0$ and $P=P^*$, exploiting the measurable functional calculus, we easily find that $f(x) = \chi_E(x)$ for some $E \subset supp(P^{(\gQ)})$. In other words $P$ is an element of the joint PVM of $\gQ$: that PVM exhausts all orthogonal projectors in $\gR'\cap \gR$.
Now, if  $\{0\} \neq \cK \subset \cH_{{\bf s}}$ is an invariant closed  subspace for $\gR$, its orthogonal projector $P_\cK$ must commute with $\gR$, so it must belong to the centre for (SS2) and thus it belongs to $P^{(\gQ)}$ for (SS1) and, more precisely it must be of the form $P_\cK = P^{(\gQ)}_{{\bf s}}$ because $P_\cK \leq P^{(\gQ)}_{{\bf s}}$ by hypothesis, but there are no projectors smaller that $P^{(\gQ)}_{{\bf s}}$ in the PVM of $\gQ$. So $\cK = \cH_{{\bf s}}$.\\
(b) $\gR_{{\bf q}} := \left\{ \left. A|_{\cH_{{\bf q}}}\:\right|\: A \in \gR\right\}$ is a von Neumann algebra on $\cH_{{\bf s}}$ considered as a Hilbert space in its own right as it arises by direct inspection. (\ref{bigdecA}) holds by definition. Since $\cH_{{\bf q}}$ is irreducible for $\gR_{{\bf q}}$, we have
$\gR_{{\bf s}} = \gR_{{\bf s}}'' = \gB(\cH_{{\bf s}})$.
Each map $\gR \ni A \mapsto  A|_{\cH_{{\bf q}}} \in \gR_{{\bf q}}$
is a representation of $^*$-algebras as follows by direct check.
If ${\bf q} \neq {\bf q}'$ --for instance $q_1 \neq q_1'$-- there is no isometric surjective map 
$U: \cH_{{\bf q}}\to \cH_{{\bf q}'}$ such that 
$$U A|_{\cH_{{\bf q}}} U^{-1} = A|_{\cH_{{\bf q}'}}$$
If such an operator existed one would have, contrarily to our hypothesis $q_1 \neq q_1'$,
$q_1I_{\cH_{{\bf q}'}} = U Q_1|_{\cH_{{\bf q}}} U^{-1} = Q_1|_{\cH_{{\bf q}'}} =q_1'I_{\cH_{{\bf q}'}}$ so that $q_1=q'_1$.
\end{proof}
\noindent We have found that, in the presence of superselection charges, the Hilbert space decomposes into pairwise orthogonal subspaces which are  invariant and irreducible with respect to the algebra of the observables, giving rise to inequivalent representations of the algebra itself. Restricting ourselves to each such subspace, QM takes its standard form as all orthogonal projectors are representatives of  elementary observables, differently from what happens in the whole Hilbert space where there are orthogonal projectors 
which cannot represent observables: These are the projectors which do not commute with $P^{(\gQ)}$.\\
There are several superselection structures as the one pointed out in physics. The three most known are of very different nature: The superselection structure of the 
{\em electric charge}, the superselection structure of 
{\em integer/semi integers values of the angular momentum}, and the one related to the mass in non-relativistic physics, i.e., {\em Bargmann's superselection rule}.

\example The electric charge is the typical example of superselction charge. For instance, referring to an electron, its Hilbert space is $L^2(\bR^3, d^3x)\otimes \cH_{s}\otimes \cH_{e}$. The space of the electric charge is $\cH_e= \bC^2$ and therein $Q= e\sigma_z$
(see (\ref{pauli})). Many other observables could exist in $\cH_{e}$ in principle, but the elecrtic charge superselection rule imposes that the only possible observables are functions of $\sigma_z$. The centre of the algebra of observables is $I\otimes I \otimes f(\sigma_3)$ for every function $f : \sigma(\sigma_z) = \{\-1,1\} \to \bC$.
We have the decomposition in coherent sectors
$$\cH = (L^2(\bR^3, d^3x)\otimes \cH_{s}\otimes \cH_+ ) \bigoplus (L^2(\bR^3, d^3x)\otimes \cH_{s}\otimes \cH_-)\:,$$
where  $\cH_\pm$ are respectively the eigenspaces of $Q$ with eigenvalue $\pm e$. \hfill $\blacksquare$\\

\remark $\null$ 

{\bf (a)} A fundamental requirement  is that the superselection charges have punctual spectrum. If instead $\gR\cap \gR'$ includes an 
operator $A$ with a continuous part in its spectrum ($A$ may also be the strong limit on  $D(A)$ of a sequence of elements in $\gR\cap \gR'$), the established proposition does not hold. $\cH$ cannot be decomposed into a direct sum of closed subspaces. In this case 
it decomposes into a direct integral and we find a much more complicated structure whose physical meaning seems dubious.

{\bf (b)} The represntations $\gR \ni A \mapsto  A|_{\cH_{{\bf q}}} \in \gR_{{\bf q}}$
are not faithful (injective), because both $I$ and $P^{(\gQ)}_{{\bf s}}$ have the same image under the representation.

{\bf (c)} The discussed picture is not the most general one though we only deal with it  in these notes. There are quantum physical systems such that their $\gR'$ is not Abelian (think of chromodynamics where $\gR'$ includes a faithful representation of $SU(3)$) so that the centre of $\gR$ does not contain the full information about $\gR'$. In this case, the non-Abelian group of the unitary operators in $\gR'$ is called the {\bf gauge group} of the theory.  The existence of a gauge group is compatible with the presence of superselection rules which are completely described by the centre $\gR'\cap \gR$. The only  difference is that now
 $\gR_{{\bf q}} = \gB(\cH_{{\bf q}})$ cannot be possible for every coherent subspace otherwise we would have $\gR' = \gR \cap \gR'$.

 \hfill  $\blacksquare$

\subsubsection{States in the presence of superselection rules} Let us come to the problem to characterize  the states when a superselection structure is assumed on a complex separable Hilbert space $\cH$ in accordance with {\bf (SS1)} and {\bf (SS2)}. In principle we can extend Definition \ref{defstates} already given for the case of $\gR$ with trivial centre. As usual 
 ${\cal L}_{\gR}(\cH)$ indicates the 
lattice of orthogonal projectors in $\gR$, which we know to be
bounded by $0$ and $I$, orthocomplemented, $\sigma$-complete, orthomodular and separable,
but not atomic and it does not satisfy the covering property in general. The atoms are one-dimensional projectors exactly as pure sates, so we may expect some difference at that level when $\gR \neq \gB(\cH)$. 

\begin{definition}\label{defstates2}
Let $\cH$ be a complex separable Hilbert space. A {\bf quantum state} in $\cH$, for a quantum sistem with von Neumann algebra of observables $\gR$,
is a map $\rho : {\cal L}_{\gR}(\cH) \to [0,1]$ such that the following requirement are satisfied.

(1) $\rho(I) =1$\:.

(2) If $\{Q_n\}_{n \in N}\subset  {\cal L}_{\gR}(\cH)$, for $N$ at most countable satisfies 
$Q_k \wedge Q_h =0$ when $h,k \in N$, then
\beq \rho(\vee_{k\in N}Q_k) = \sum_{k \in N} \rho(Q_k)\:.\label{sigmaadd22}\eeq
The set of the states  will be denoted by $\gS_{\gR}(\cH)$. \hfill $\blacksquare$
\end{definition}
\noindent 
If there is a superselection structure we have the decompositions we re-write down into a simpler version,
\beq \cH = \bigoplus_{k\in K} \cH_{k} \:, \quad \gR = \bigoplus_{k\in  K} \gR_{k}\:, \quad \gR_{k} = \gB(\cH_k)\:, \: k \in K\label{bigdecr}\eeq
where $K$ is some finite or countable set.
The lattice ${\cal L}_{\gR}(\cH)$, as a consequence of (\ref{sigmaadd22}), decomposes as  (the notation should be obvious)
\beq
{\cal L}_{\gR}(\cH) = \bigvee_{k\in K} {\cal L}_{\gR_k}(\cH_k) =  \bigvee_{k\in K} {\cal L}(\cH_k)
\eeq
where 
$$ {\cal L}_{\gR_k}(\cH_k) \bigwedge  {\cal L}_{\gR_h}(\cH_h) = \{0\} \quad \mbox{if $k\neq h$}\:.$$
In other words $Q \in {\cal L}_{\gR}(\cH)$ can uniquely be written as
$Q = +_{k\in K} Q_k$  where   $Q_k \in {\cal L}(\gB(\cH_k))$. In fact 
$Q_k = P_kQ_k$, where $P_k$ is the orthogonal projector onto $\cH_k$.\\
In this framework, it is possible to readapt Gleason's result simply observing that 
a state $\rho$ on ${\cal L}_\gR(\cH)$ as above defines a state $\rho_k$ on 
${\cal L}_{\gR_k}(\cH_k) = {\cal L}(\cH_k)$ by 
$$\rho_k (P) := \frac{1}{\rho(P_k)} \rho(P)\:, \quad P \in {\cal L}(\cH_k)\:.$$
If $dim(\cH_k)\neq 2$ we can exploit Gleason's theorem.

\begin{theorem}
Let $\cH$ be a complex separable Hilbert space and assume that the von Neumann algebra $\gR$ in $\cH$ satisfies {\bf(SS1)}
and {\bf (SS2)}, so that the decomposition (\ref{bigdecr}) in coherent sectors is valid where we suppose $dim \cH_k \neq 2$ for every $k\in K$. The following facts hold.\\
{\bf (a)} If $T \in \gB_1(\cH)$ satisfies $T\geq 0$ and $tr\: T=1$ then
$$\rho_T : {\cal L}_{\gR}(\cH) \ni P \mapsto tr(T P)$$ is an elemeont of $\gS_{\gR}(\cH)$ that is a state on ${\cal L}_{\gR}(\cH)$.\\
{\bf (b)} For  $\rho\in \gS_{\gR}(\cH)$ there is a $T \in \gB_1(\cH)$ satisfies $T\geq 0$ and $tr\: T=1$ such that $\rho = \rho_T$.\\
{\bf (c)} If $T_1, T_2 \in \gB_1(\cH)$ satisfy same hypotheses as $T$ in (a), then  
 $\rho_{T_1}=\rho_{T_2}$ is valid if and only if $P_kT_1P_k = P_kT_2P_k$ for all $k\in K$,
$P_k$ being the orthogonal projector onto $\cH_k$.\\
{\bf (d)}  A unit vector  $\psi \in \cH$ defines a pure state only if  belongs to a coherent sector. More precisely,
a state $\rho\in \gS_{\gR}(\cH)$ is pure, that is extremal, if and only if there is $k_0\in K$, $\psi \in \cH_{k_0}$  with $||\psi||=1$ such that
$$\rho(P)=0 \quad \mbox{if $P \in {\cal L}(\cH_{k})$, $k\neq k_0$ \quad and \quad  $\rho(P) =
\langle \psi|P \psi \rangle$ if $P \in {\cal L}(\cH_{k_0})$}$$
\end{theorem}

\begin{proof} (a) is obvious from Proposition \ref{proppregleason}, as restricting a state $\rho$ on ${\cal L}(\cH)$ to ${\cal L}_{\gR}(\cH)$ we still obtain a state as one can immediately verify.
Let us prove (b). Evidently, every $\rho|_{{\cal L}(\cH_k)}$ is a positive measure with
$0\leq \rho(P_k) \leq 1$. We can apply Gleason's theorem finding $T_k\in \gB(\cH_k)$ with $T_k\geq 0$
and $Tr\: T_k = \rho(P_k)$ such that $\rho(Q)= tr(T_kQ)$ if $Q\in {\cal L}(\cH_k)$. 
Notice also that $||T_k||\leq \rho(P_k)$ because
$$||T_k||= \sup_{\lambda \in \sigma_p(T_k)} |\lambda| =
\sup_{\lambda \in \sigma_p(T_k)} \lambda
 \leq \sum_{\lambda \in \sigma_p(T_k)} d_\lambda \lambda = Tr\:T_k = \rho(P_k)\:.$$
If $Q \in {\cal L}_{\gR}(\cH)$, $Q = \sum_kQ_k$, where $Q_k:= P_kQ \in {\cal L}(\cH_k)$, $Q_kQ_h=0$ if $k \neq h$ and thus, by $\sigma$-additivity,
$$\rho(Q)= \sum_k \rho(Q_k) = \sum_k tr(T_kQ_k)$$
since $\cH_k \perp \cH_h$, this identity can be rewritten as 
$$\rho(Q) = tr(TQ)$$
provided $T := \oplus_{k}T_k\in \gB_1(\cH)$. It is clear that $T \in \gB(\cH)$ because, if $x\in \cH$ and $||x||=1$ then, as $x = \sum_k x_k$ with $x_k\in \cH_k$,  $||Tx|| \leq  \sum_k||T_k||\:||x_k|| \leq 
\sum_k||T_k|| 1 \leq \sum_k \rho(P_k) =1$. In particular $||T||\leq 1$. $T\geq 0$ because each $T_k\geq 0$. Hence $|T| = \sqrt{T^*T} = \sqrt{TT} = T$ via functional calculus, and also $|T_k|=T_k$. Moreover, using the spectral decomposition of $T$, whose PVM commutes with each $P_k$, one easily has $|T| = \oplus_k |T_k| = \oplus_k T_k$.
The condition 
$$1= \rho(I) =\sum_k \rho(P_k) = \sum_k tr(T_kP_k)= \sum_k tr(|T_k|P_k)$$
is equivalent to say that $tr\: |T| =1$  using a Hilbertian basis of $\cH$ made of the union of bases in each $\cH_k$. We have obtained, as wanted, that $T \in \gB_1(\cH)$, $T\geq 0$, $tr\: T=1$ and $\rho(Q) = tr(TQ)$ for all $Q \in {\cal L}_{\gR}(\cH)$.\\
 (c) The proof straightforwardly follows form  ${\cal L}_{\gR_k}(\cH_k)
= {\cal L}(\gB(\cH_k))$ because $\gR_k = \gB(\cH_k)$ and, evidently,  $\rho_{T_1} = \rho_{T_2}$ if and only if $\rho_{T_1}|_{{\cal L}(\gB(\cH_k))} = \rho_{T_2}|_{{\cal L}(\gB(\cH_k))}$ for all $k\in K$. Regarding (d) it is clear that if $\rho$ encompasses more than one component $\rho|_{{\cal L}(\cH_k)}\neq 0$ cannot be extremal because is, by construction, a convex combination of other states which vanishes in some of the given coherent subspace. Therefore only states such that only one restriction $\rho|_{{\cal L}(\cH_{k_0})}$ does not vanish may be extremal. Now (a) of Proposition \ref{extremalstates} implies that, among these states, the extremal ones are precisely those of the form said in (d) of the thesis. 
\end{proof}

\remark $\null$\\
{\bf (a)} Take $\psi = \sum_{k\in K} c_k\psi_k$ where the $\psi_k \in \cH_k$ are unit vectors and also suppose that $||\psi||^2 = \sum_k |c_k|^2 =1$.
This vector induces a state $\rho_\psi$ on $\gR$ by means of the standard procedure (which is nothing but the trace procedure with respect to $T_\psi := \langle \psi| \:\: \rangle \psi$!)
$$\rho_\psi(P) = \langle \psi|P \psi \rangle \quad P \in {\cal L}_{\gR}(\cH)\:.$$
In this case however, since $PP_k=P_kP$ and $\psi_k = P_k\psi_k$ we have
$$\rho_\psi(P) = \sum_{k}\sum_h \overline{c_k}c_h \langle \psi_k|P_kPP_h \psi_k \rangle
= \sum_{k}\sum_h \overline{c_k}c_h \langle \psi_k|PP_kP_h \psi_k \rangle $$ $$ =
\sum_{k}\sum_h \overline{c_k}c_h \langle \psi|PP_k \psi \rangle\delta_{kh}  = \sum_k |c_k|^2 \langle \psi_k|P \psi_k \rangle = tr(T'_\psi P)$$
where 
$$T'_\psi = \sum_{k\in K} |c_k|^2 \langle \psi_k | \:\: \rangle \psi_k$$
We conclude that the apparent pure state $\psi$ and the mixed state $T'_\psi$ cannot be distinguished, just because the algebra $\gR$ is too small to make a difference. Actually they define the same state at all and this is an elementary case of (c) in the above theorem with $T_1 = \langle \psi| \:\: \rangle \psi$ and $T_2= T'_\psi$. \\
This discussion, in the language of physicist is often stated as follows: \\ {\em No coherent superpositions $\psi = \sum_{k\in K} c_k\psi_k$ of pure states $\psi_k \in \cH_k$ of different coherent sectors are possible, only incoherent superpositions $\sum_{k\in K} |c_k|^2 \langle \psi_k | \:\: \rangle \psi_k$ are allowed.}\\
{\bf (b)} It should be clear that the one-to-one correspondence between pure states and atomic elementary observables (one-dimensional projectors) here does not work. Consequently, notions like {\em probability amplitude} must be handled with great care. In general, however, everything goes right if staying in a fixed superselection sector $\cH_k$ where the said correspondence exists. \hfill $\blacksquare$

\subsection{Quantum Symmetries: unitary projective representations}
The notion of symmetry in QM is quite abstract. Actually there are three distinct ideas, respectively  by Wigner, Kadison and Segal \cite{Simon}. Here we focus on the first pair  only. 
Physically speaking, a {\em symmetry} is an active transformation on the quantum system changing its state.  It is supposed that this transformation preserves some properties of the physical system  and here we have to distinguish between the two afore-mentioned cases.
However in both cases the transformation is required to be reversible (injective)  and to cover (surjective) the space of the states. Symmetries are supposed to mathematically describe some concrete transformation acting  on the physical system. Sometimes their action, in practice, can be cancelled simply changing the reference frame. This is not the general case however, even if  this class of symmetries plays a relevant role in physics.

\subsubsection{Wigner and Kadison theorems, groups of symmetries}
Consider a quantum system described in the complex Hilbert space $\cH$
with dimension $\neq 2$, separable whenever infinite dimensional. We assume that either  $\cH$ is the whole Hilbert space in the absence of superselection charges or it denotes a single coherent sector. Let $\gS(\cH)$ and $\gS_p(\cH)$ respectively indicate the convex body of the quantum states and the set of pure states, referred to the sector $\cH$ if it is the case. 

\begin{definition} {\em If $\cH$ is a complex Hilbert space with dimension $\neq 2$, separable if infinite dimensional, we have the following definitions.\\
{\bf (a)} A {\bf Wigner symmetry} is a bijective map $$s_W : \gS_p(\cH) \ni \langle \psi| \:\: \rangle \psi \to \langle \psi'| \:\: \rangle \psi' \in \gS_p(\cH)$$ which preserves the probabilties of transition. In other words
$$|\langle\psi_1|\psi_2\rangle|^2 = |\langle\psi_1'|\psi_2'\rangle|^2 \quad \mbox{if}\quad  \psi_1 \:,\psi_2 \in \cH\:\:\: \mbox{with}\:\: ||\psi_1||=||\psi_2||=1\:.$$
{\bf (b)} A {\bf Kadison symmetry} is a bijective map $$s_K : \gS(\cH) \ni \rho \to \rho' \in \gS(\cH)$$ which preserves the convex structure of the space of the states. In other words
 $$(p\rho_1 + q \rho_2)' = p\rho_1' + q\rho_2'\quad \mbox{if}\quad \rho_1,\rho_2 \in \gS(\cH) \quad\mbox{and\:\: $p,q\geq 0$ with $p+q=1$.} $$}
\hfill $\blacksquare$
\end{definition}
\noindent We observe that the first definition is well-posed even if unit vectors define pure states just up to a phase, as the reader can immediately prove, because transition probabilities are not affected by that ambiguity.\\
Though the definitions are evidently of different nature, they lead to the same mathematical object, as established in a pair of famous characterization theorems we quote into a unique statement.  We need a preliminary definition.
\begin{definition} {\em Let $\cH$ a complex Hilbert space. A map $U : \cH \to \cH$ is said to be an {\bf antiunitary operator} if it is surjective, isometric and $U(ax+by)=
\overline{a}Ux + \overline{b}Uy$ when $x,y \in \cH$ and $a,b \in \bC$.} \hfill $\blacksquare$
\end{definition}
\noindent We come to the celebrated theorem.
The last statement is obvious, the difficult parts are (a) and (b) (see, e.g.,\cite{moretti}).
\begin{theorem}[Wigner and Kadison theorems]
Let $\cH$ be a complex Hilbert space with dimension $\neq 2$, separable if infinite dimensional. The following facts hold.\\
{\bf (a)} For every  Wigner symmetry  $s_W$ there is an operator $U : \cH \to \cH$, which can be either unitary or anti unitary  (depending on $s_w$) such that 
\beq s_w : \langle \psi| \:\: \rangle \psi \to \langle U\psi| \:\: \rangle U\psi\:,
\quad \forall \langle \psi| \:\: \rangle \psi \in \gS_p(\cH)\:.\label{Ws}\eeq
$U$ and $U'$ are associated to the same $s_W$ if and only if $U'=e^{ia}U$ for $a \in \bR$.\\
{\bf (b)}  For every  Kadison symmetry  $s_K$ there is an operator $U : \cH \to \cH$, which can be either unitary or anti unitary  (depending on $s_K$) such that 
\beq s_w : \rho  \to U\rho U^{-1} \:,\quad \forall \rho \in \gS(\cH)\:. \label{Ks}\eeq
$U$ and $U'$ are associated to the same $s_K$ if and only if $U'=e^{ia}U$ for $a \in \bR$.\\
{\bf (c)} $U: \cH \to \cH$, either unitary or antiunitary, simultaneously defines a Wigner and a Kadison symmetry by means of (\ref{Ws}) and (\ref{Ks}) respectively.
\end{theorem}

\remark $\null$

{\bf (a)} It is worth stressing that the Kadison notion of symmetry is an extension of 
the Wigner one, after the result above. In fact, a Kadison  symmetry $\rho \mapsto U\rho U^{-1}$ restricted to  one-dimensional projectors preserves the probability transitions, as immediately follows from the identity $|\langle \psi|\phi \rangle|^2 = tr(\rho_\psi\rho_\phi)$ and the cyclic property of the trace, where we use the notation 
$\rho_\chi = \langle \chi| \:\: \rangle \chi$. In particular we can use the same operator $U$ to represent also the found  Wigner symmetry.  

{\bf (b)} If superselection rules are present, in general quantum symmetries are described in a similar way 
with unitary or antiunitary operators acting in a single coherent sector or also  swapping different sectors \cite{moretti}.
\hfill $\blacksquare$	\\

\noindent If a unitary or antiunitary operator $V$ represents a symmetry $s$, it has an action on observables, too.
If $A$ is an observable (a selfadjoint operator on $\cH$), we define the {\bf transformed observable} along the action of $s$ as \beq s^*(A):= VAV^{-1}\:.\label{VAV}\eeq
Obviously $D(s^*(A))= V(D(A))$.
It is evident that this definition is not affected by the ambiguity of the arbitrary phase in the choice of $V$ when $s$ is given. \\
According with (i) in Proposition \ref{propint2} the spectral measure of $s^*(A)$ is $$P^{(s^*(A))}_E= V P_E^{(A)} V^{-1} = s^*(P_E^{(A)})$$ as expected.  \\
The meaning of $s^*(A)$ should be evident: The probability that the observable $s^*(A)$ produces the outcome $E$ when the state is $s(\rho)$ (namely   $tr(P_E^{(s^*(A))}s(\rho))$) is the same as the probability that the observable $A$ produces the outcome $E$ when the state is $\rho$ (that is $tr(P_E^{(A)} \rho)$). Changing simultaneously and coherently observables and states nothing changes. 
 Indeed
$$tr(P_E^{(s^*(A))}s(\rho)) = tr(VP_E^{(A)} V^{-1} V\rho V^{-1}) =   tr(VP_E^{(A)}\rho V^{-1}) $$
$$ =tr(P_E^{(A)}\rho V^{-1}V) =  tr(P_E^{(A)} \rho)\:.$$

\begin{example}\label{exrep} $\null$\\
{\bf (1)} Fixing an inertial reference frame, the pure state of a quantum particle is defined, up to phases,  as a unit norm element $\psi$ of  $L^2(\bR^3, d^3x)$, where $\bR^3$ stands for the rest three space of the reference frame. The group of isometries $IO(3)$ of $\bR^3$ equipped with the standard Euclidean structure acts on states by means of symmetries the sense of Wigner and Kadison. If $(R,t) : x \mapsto Rx+t$ is the action of the generic element of $IO(3)$, where $R\in O(3)$ and $t \in \bR^3$, the associated quantum (Wigner) symmetry $s_{(R,t)}(\langle \psi |\:\: \rangle \psi) = \langle U_{(R,t)}\psi |\:\: \rangle U_{(R,t)}\psi$
is completely fixed by the unitary  operators  $U_{(R,t)}$. They are defined as
$$(U_{(R,t)}\psi)(x) := \psi((R,t)^{-1}x) \:, \quad x \in \bR^3\:, \psi \in L^2(\bR^3, d^3x)\:, \quad ||\psi||=1\:.$$
The fact that the Lebesgue measure is invariant under $IO(3)$ immediately proves that $U_{(R,t)}$ is unitary. 
It is furthermore easy to prove that, with the given definition
\beq U_{(I,0)}  = I \:, \quad U_{(R,t)}U_{(R',t')}=U_{(R,t)\circ (R',t')} \:, \quad \forall (R,t), (R',t') \in IO(3)\:. \label{groupIO}\eeq
{\bf (2)} The so called {\em time reversal} transformation classically corresponds to invert the sign of all the velocities of the physical system. It is possible to prove \cite{moretti} (see also (3) in exercise \ref{remconstt} below) that, in QM and for systems whose energy is bounded below but not above, the time reversal symmetry cannot be represented by unitary transformations, but only antiunitary.
In the most elementary situation as in (1), the time reversal is defined by means of the anti unitary operator  
$$(T\psi)(x) := \overline{\psi(x)} \:, \quad x \in \bR^3\:,  \psi \in L^2(\bR^3, d^3x)\:, \quad ||\psi||=1\:.$$
{\bf (3)} According to the example in (1), let us focus on   the subgroup of $IO(3)$  of  displacements  along $x_1$ parametrized by $u\in \bR$,
$$\bR^3 \ni x \mapsto x + u{\bf e}_1\:,$$
where  ${\bf e}_1$ 
denotes  the unit vector in $\bR^3$ along $x_1$. For every  value of the parameter $u$,
we indicate by $s_u$  the corresponding (Wigner) quantum symmetry,  $s_u(\langle \psi |\:\: \rangle \psi) = \langle U_u\psi |\:\: \rangle U_u\psi$ with 
$$(U_{u}\psi)(x) = \psi(x - u{\bf e_1})\:, \quad u \in \bR\:,$$
The action of this symmetry on the observable $X_k$ turns out to be
$$s^*_u(X_k) = U_u X_k U^{-1}_u= X_k + u\delta_{k1} I\:, \quad u \in \bR \:.$$
\hfill $\blacksquare$
\end{example}

\subsubsection{Groups of quantum symmetries} As in (1) in the example above, very often in physics one deals with groups of symmetries. In other words, there is a certain group $G$, with unit element $e$ and group product $\cdot$, and one  associates each element $g\in G$ to a symmetry $s_g$ (if Kadison or Wigner is immaterial here, in view of the above discussion). In turn, $s_g$ is associated to an operator $U_g$, unitary or antiunitary.

\remark In the rest of this section, we assume that all the $U_g$ are unitary. \hfill $\blacksquare$\\

\noindent It would be nice to fix these operators $U_g$ in order that the map $G \ni g \mapsto U_g$ be a {\bf unitary representation} of $G$ on $\cH$, that is
\beq
U_e = I\:, \quad U_gU_{g'} = U_{g\cdot g'} \quad g,g' \in G \label{gR}
\eeq
The identities (\ref{groupIO}) found in (1) in example \ref{exrep} shows that it is possible at least in certain cases.
In general the requirement (\ref{gR}) does not hold. What we know is that $U_{g\cdot g'}$
equals $U_gU_{g'}$ just   {\em up to phases}:
\beq U_gU_{g'}U_{g\cdot g'}^{-1}= \omega(g,g') I \quad \mbox{with $\omega(g,g')\in U(1)$ for all  $g,g' \in G$}. \label{gR2}\eeq 
For $g=e$ this identity gives in particular \beq U_e = \omega(e,e)I \label{Uee}\:. \eeq The numbers $\omega(g,g')$ are called {\bf multipliers}. 
They cannot be completely arbitrary, indeed associativity of composition of operators $(U_{g_1}U_{g_2})U_{g_3}=
U_{g_1}(U_{g_2}U_{g_3})$ yields the identity
\beq
\omega(g_1,g_2)\omega(g_1\cdot g_2,g_3) = \omega(g_1, g_2\cdot g_3)\omega(g_2,g_3)\:,\quad g_1,g_2,g_3 \in G\label{coci}
\eeq
which also implies
\beq
\omega(g,e) = \omega(e,g) = \omega(g',e)\:, \quad \omega(g,g^{-1}) = \omega(g^{-1},g)\:,\quad g, g' \in G\:.
\eeq
\begin{definition} {\em If $G$ is a group, a map  $G\ni g \mapsto U_g$ -- where the $U_g$ are unitary operators in the complex Hilbert space $\cH$ -- is named a {\bf unitary projective representation} of $G$ on $\cH$ if (\ref{gR2}) holds (so that also (\ref{Uee}) and (\ref{coci}) are valid). Moreover,
 
(i) two  unitary projective representation $G\ni g \mapsto U_g$ and $G\ni g \mapsto U'_g$ are said to be {\bf equivalent} if $U'_g= \chi_g U_g$, where $\chi_g \in U(1)$ for every $g\in G$. That is the same as requiring  that there are numbers $\chi_h \in U(1)$, if $h \in G$, such that
\beq
\omega'(g,g') = \frac{\chi_{g\cdot g'}}{\chi_g\chi_{g'}} \:\omega(g,g')\quad \forall g,g' \in G
\eeq
with obvious notation;

(ii) a unitary projective representation with $\omega(e,e)=\omega(g,e)=\omega(e,g)=1$ for every $g\in G$ is said to be {\bf normalized}.} \hfill $\blacksquare$
\end{definition}
\noindent 
\remark $\null$

 {\bf (a)} It is easily  proved that every unitary projective representation is always  equivalent to  a normalized representation. 

 {\bf (b)}  It is clear that two  projective unitary representations are equivalent if and only if they are  made of the same Wigner (or Kadison) symmetries.

{\bf (c)} In case of superselection rules, continuous symmetries representing a connected topological group 
do not swap  different  coherent sectors when acting on pure states \cite{moretti}.  

{\bf (d)} One may wonder if it is possible to construct a group representation $G\ni g \mapsto V_g$ where 
the operators $V_g$ may be both unitary or antiunitary. If  every  $g\in G$ can be written as  $g=h\cdot h$
for some $h$ depending on $g$ -- and this is the case if $G$ is a connected Lie group -- all the operators $U_g$ must be unitary because $U_g =U_hU_h$
is necessarily linear no matter if $U_h$ is linear or anti linear. The presence of arbitrary phases  does not change the result. 
\hfill $\blacksquare$\\

\noindent  Given a  unitary projective representation, a technical problem is to check if it is equivalent to a unitary representation, because unitary representations are much  simpler to handle. This is a difficult problem \cite{varadarajan,moretti} which is tackled especially when $G$ is a {\em topological group} (or {\em Lie group}) and the representation satisfies the following natural {\em continuity property}
\begin{definition}\label{defrapcont} {\em A unitary projective representation of the topological group $G$, $G\ni g \mapsto U_g$ on the Hilbert space $\cH$ is said to be {\bf continuous} if  the map 
$$G \ni g \mapsto |	\langle \psi| U_g \phi \rangle |$$
is continuous for every $\psi,\phi \in \cH$.}\hfill $\blacksquare$\\
\end{definition}
\noindent The notion of continuity defined above is natural as it regards continuity of probability transitions.
A well known co-homological condition assuring that a unitary projective representation
of  Lie groups is equivalent to a unitary one is due to Bargmann \cite{BaRa,moretti}.

\begin{theorem}[Bargmann's criterion] Let $G$ be a connected and simply connected  (real finite dimensional) Lie group with Lie algebra $\gg$. Every continuous unitary projective  representation of $G$ in a complex Hilbert space is equivalent to a strongly  continuous unitary representation of $G$ if, for every bilinear antisymmetric map $\Theta : \gg\times \gg \to \bR$ such that
$$\Theta([u,v],w)+ \Theta([v,w],u)+ \Theta([w,u],v) =0\:,\quad \forall u,v,w \in \gg$$
there is a linear map $\alpha : \gg \to \bR$ such that $\Theta(u,v) = \alpha([u,v])$, for all $u,v \in \gg$.
\end{theorem}

\remark The condition is equivalent to require that the second cohomology group $H^2(G,\bR)$ is trivial.  $SU(2)$ for instance satisfies the requirement.
 \hfill $\blacksquare$\\

\noindent  However, non-unitarisable unitary projective representations do exist and one has to deal with them. There is nevertheless a way to circumvent the technical problem. Given a unitary projective representation $G \ni g \mapsto U_g$ with multiplicators $\omega$, let us put on $U(1) \times G$ the group structure  arising by the product $\circ$
$$(\chi, g) \circ  (\chi',g') = (\chi\chi' \omega(g,g'), g\cdot g')$$
and indicate by $\hat{G}_\omega$ the obtained group. The map 
$$\hat{G}_\omega \ni (\chi, g) \mapsto \chi U_g =: V_{(\chi,g)}$$
is a {\em unitary representation} of $\hat{G}_\omega$. If the initial representation is normalized, $\hat{G}_\omega$ is said to be a {\bf central extension of $G$} by means of $U(1)$ \cite{varadarajan,moretti}. Indeed, the elements $(\chi,e)$, $\chi \in U(1)$, commute with all the elements of $\hat{G}_\omega$ and thus they belong to the centre of the group.

\remark These types of unitary representations of central extensions   play a remarkable role in physics. Sometimes $\hat{G}_\omega$ with a particular choice for $\omega$ is seen as the {\em true} group of symmetries at quantum level, when $G$ is the {\em classical group} of symmetries. There is a very important case. If $G$ is the {\em Galileian group} -- the group of transformations between inertial reference frames in classical physics, viewed as {\em active} transformations --  as clarified by Bargmann \cite{moretti} the only physically relevant unitary projective representations in QM are just the ones which are {\em not} equivalent to unitary representations! The multiplicators embody the information about the {\em mass} of the system. This phenomenon gives also rise to a famous superselection structure in the Hilbert space of quantum systems admitting the Galileian group as a symmetry group, known as  {\em Bargmann's superselection rule} \cite{moretti}.\hfill $\blacksquare$\\

\noindent To conclude we just  state  a  technically  important result \cite{moretti} which introduces the 
one-parameter strongly continuous unitary groups as crucial tool in QM.

\begin{theorem}\label{teoone}
Let  $\gamma : \bR \ni r \mapsto U_r$ be  a continuous  unitary projective  representation 
of the additive topological group $\bR$ on the complex Hilbert space $\cH$. The following facts hold.\\
{\bf (a)}  $\gamma$ is  equivalent to a strongly continuous unitary representation $\bR \ni r \mapsto V_r$ of the same topological additive group on $\cH$. \\
{\bf (b)} A strongly continuous unitary representation $\bR \ni r \mapsto V'_r$ is equivalent to  $\gamma$ if and only if  $$V'_r = e^{icr}V_r$$ for some constant $c \in\bR$ and all $r\in \bR$.
\end{theorem}
\noindent The above unitary representation can also be defined as  {\em strongly continuous  one-parameter  unitary  group}.
\begin{definition} {\em  If $\cH$ is a Hilbert space,  $V : \bR \ni r \mapsto V_r \in \gB(\cH)$, such that

 (i) $V_r$ is unitary for every $r\in \bR$\;

 (ii) $V_rV_s =V_{r+s}$ for all $r,s \in \bR$,\\
\noindent  is called  {\bf one-parameter  unitary  group}. 
It is called  {\bf strongly continuous  one-parameter  unitary  group} if in addition to (i) and (ii) we also have

 (iii)  V is continuous referring to the strong operator topology. In other words $V_r\psi \to V_{r_0}\psi$ for $r\to r_0$ and every $r_0 \in \bR$ and $\psi \in \cH$.} \hfill $\blacksquare$
\end{definition}

\remark$\null$

{\bf (a)} It is evident that, in view of the group structure,  a  one-parameter  unitary  group $\bR \ni r \mapsto V_r \in \gB(\cH)$  is strongly continuous if and only if 
is strongly continuous for $r=0$.

{\bf (b)} It is a bit less evident but  true  that a  one-parameter  unitary  group $\bR \ni r \mapsto V_r \in \gB(\cH)$
is strongly continuous if and only if it is {\em weakly continuous} at  $r=0$.  
Indeed, if $V$ is weakly continuous at $r=0$, for every  $\psi \in \cH$, we have
$$||U_r\psi - \psi||^2 = ||U_r\psi||^2 + ||\psi||^2 - \langle \psi| U_r\psi \rangle - \langle U_r\psi|\psi \rangle
= 2||\psi||^2 - \langle \psi| U_r\psi \rangle - \langle U_r\psi|\psi \rangle \to 0$$
 for  $r\to 0$. \hfill $\blacksquare$

\subsubsection{One-parameter strongly continuous unitary groups:  von Neumann and  Stone theorems} 
Theorem \ref{teoone} establishes that, dealing with continuous  unitary projective  representation 
of the additive topological group $\bR$, one can always reduce to work with proper  strongly continuous 
one-parameter  unitary  groups. 
So, for instance, the action on a quantum system of rotations around an axis  can always described by means of strongly continuous  one-parameter  unitary  groups.  There is a couple of technical results of very different nature which are very useful in QM.
The former is due to von Neumann  \cite{moretti} and proves that the  one-parameter  unitary  group which are not strongly continuous are not so many  in separable Hilbert spaces.
\begin{theorem}
If $\cH$ is a separable complex  Hilbert space and $V : \bR \ni r \mapsto V_r \in \gB(\cH)$ is a one parameter unitary group, it is strongly continuous if and only if the maps $\bR \ni r \mapsto \langle \psi|U_r \phi \rangle$
are Borel measurable for all $\psi, \phi \in \cH$.
 \end{theorem}
\noindent The second proposition  we quote \cite{moretti} is a celebrated result due to Stone (and later extend to the famous {\em Hille-Yosida theorem} in Banach spaces).  We start by noticing that, if $A$ is a selfadjoint 
operator in a Hilbert space, $U_t := e^{itA}$, for $t \in \bR$, defines a strongly continuous  one-parameter  unitary  group as one easily proves using the functional calculus. The result is remarkably reversible. 

\begin{theorem}[Stone theorem]\label{teostone}
Let $\bR \ni t \mapsto U_t \in \gB(\cH)$ be  a strongly continuous  one-parameter  unitary  group
in the  complex Hilbert space  $\cH$. The following facts hold.\\
{\bf (a)} There exists a unique selfadjoint operator, called the {\bf generator} of the group, $A:D(A)\to \cH$ in $\cH$, such that
\beq
U_t= e^{-itA}\:,\quad t \in \bR\:.
\eeq 
{\bf (b)} The generator is determined  as 
\beq A\psi = i\lim_{t \to 0} \frac{1}{t}(U_t -I)\psi \label{dstone}\eeq
and $D(A)$ is made of the vectors $\psi \in \cH$ such that the right hand side of (\ref{dstone}) exists in $\cH$.\\
{\bf (c)}  $U_t(D(A))\subset D(A)$ for all $t\in \bR$ and
$$AU_t\psi =U_tA\psi \quad \mbox{if $\psi \in D(A)$ and $t\in \bR$.}$$
\end{theorem}
\remark $\null$

{\bf (a)} For a selfadjoint operator $A$, the expansion $$e^{-itA}\psi = \sum_{n=0}^{+\infty} \frac{(-it)^n}{n!}A^n \psi$$
generally does {\em not} work for $\psi \in D(A)$. It works in two cases however: (i) if $\psi$ is an {\em analytic} vector of $A$ (Def. \ref{defanalitic} and this result is due to Nelson), (ii) if $A\in \gB(\cH)$ which is equivalent to say that $D(A)=\cH$. In the latter case, one more strongly finds $e^{-itA} = \sum_{n=0}^{+\infty} \frac{(-it)^n}{n!}A^n$, referring to the uniform operator topology. \cite{moretti}.

{\bf (b)} One parameter unitary  group generated by selfadjoint operators can be used to check if the associated observables are compatible in view of the following nice result \cite{moretti}.
\begin{proposition} If $A$ and $B$ are selfadjoint operators in the complex Hilbert space $\cH$, the identity holds
$$e^{-itA}e^{-isB}= e^{-isB}e^{-itA} \quad \forall t,s \in \bR$$ 
 if and only if the spectral measures of $A$ and $B$ commute. \hfill $\blacksquare$
\end{proposition}

\subsubsection{Time evolution, Heisenberg picture and quantum Noether theorem}
Consider a quantum system described in the Hilbert space $\cH$ when an inertial reference frame is fixed. Suppose that, physically speaking, the system is either isolated or interacts with some external stationary environment. With these hypotheses,  the time evolution of states is axiomatically described by
a continuous symmetry, more precisely,  by a continuous one-parameter group of  unitary projective operators $\bR \ni  t \mapsto  V_t$. In view of Theorems \ref{teoone}  and \ref{teostone}, this group is equivalent to a strongly continuous one-parameter group of  unitary  operators $\bR \ni t \mapsto U_t$ and, up to additive constant, there is a unique selfadjoint operator $H$, called the {\bf Hamiltonian operator}
such that (notice the sign in front of the exponent) 
\beq
U_t = e^{-itH}\:, \quad t \in \bR\:.
\eeq
The observable represented by $H$ is usually identified  with {\em the energy of the system} in the considered reference frame.\\
Within this picture, if $\rho \in \gS(\cH)$ is the state of the system  at $t=0$, as usual described by a positive trace-class operator with unit trace,  the state at time $t$  is $\rho_t = U_t \rho U^{-1}_t$. If the initial state is pure and represented by  the unit vector $\psi\in \cH$, the state at time $t$ is $\psi_t := U_t\psi$. In this case, if 
$\psi \in D(H)$ we have that $\psi_t \in D(H)$ for every $t\in \bR$ in view of (c) in Theorem \ref{teostone}
and furthermore, for (b) of the same theorem
\beq -iH\psi_t =  \frac{d\psi_t}{dt}\label{Seq} \:.\eeq
where the derivative is computed wit respect to the topology of $\cH$.
One recognises in Eq. (\ref{Seq})  the general form of {\bf Sch\"odinger equation}.

\remark It is possible to study quantum systems interacting with some external system which is not stationary. In this case the Hamiltonian observable depends parametrically on time as already introduced in remark \ref{rem8}. In these cases a Schr\"odinger equation is assumed to describe the  time evolution of the system giving rise to a groupoid  of unitary operators \cite{moretti}. We shall not enter into the details of this technical issue here. \hfill $\blacksquare$\\

\noindent Adopting the above discussed framework,  observables do not evolve and states do. This framework is called 
{\bf Schr\"odinger picture}.
There  is however another approach to describe time evolution called {\bf Heisenberg picture}. In this representation 
states do not evolve in time but observables do. If $A$ is an observable at $t=0$, its evolution at time $t$
is the observable
$$A_t := U^{-1}_t A U_t\:.$$
Obviously $D(A_t)= U^{-1}_t(D(A))= U_{-t}(D(A))= U_t^*(D(A))$.
According with (i) in Proposition \ref{propint2} the spectral measure of $A_t$ is $$P^{(A_t)}_E= U^{-1}_t P_E^{(A)} U_t$$ as expected. The probability that, at time $t$, the observable $A$ produces the outcome $E$ when the state is $\rho$ at $t=0$, can equivalently be computed both using the standard picture, where states evolve 
as $tr(P_E^{(A)}\rho_t) $,
or Heisenberg picture where observables do obtaining $tr(P_E^{(A_t)} \rho)$. Indeed
$$tr(P_E^{(A)}\rho_t) = tr(P_E^{(A)} U_t^{-1}\rho U_t) =  tr(U_tP_E^{(A)} U_t^{-1}\rho) = tr(P_E^{(A_t)} \rho)\:.$$
The two pictures are completely equivalent to describe physics.  Heisenberg picture permits to give the following important definition

\begin{definition}
 {\em In the complex Hilbert space $\cH$ equipped with a strongly continuous unitary one-parameter group  representing the time evolution $\bR \ni t \mapsto U_t$, an observable represented by the  selfadjoint operator $A$ is said to be a {\bf constant of motion} with  respect to $U$, if  $A_t =A_0$. }
\end{definition}
\noindent The meaning of the definition should be clear: Even if the state evolve, the probability to obtain 
an outcome $E$, measuring a constant of motion $A$, remains  stationary.  Also expectation values and standard deviations do not change in time.\\  
We are now in a position to state the equivalent of the {\em Noether theorem} in QM.

\begin{theorem}[Noether quantum theorem]\label{neother}
Consider a quantum system described  in the complex Hilbert space $\cH$ equipped with a strongly continuous unitary one-parameter group  representing the time evolution $\bR \ni t \mapsto U_t$. If $A$
is an observable represented by a (generally unbounded) selfadjoint operator $A$ in $\cH$, the following facts are equivalent.

{\bf (a)} $A$ is a constant of motion: $A_t=A_0$ for all $t \in \bR$.

{\bf (b)}  The one-parameter group of symmetries generated by $A$, $\bR \ni s \mapsto e^{-isA}$ is a {\bf group of dynamical symmetries}: It commutes with time evolution \beq e^{-isA} U_t = U_t e^{-isA}\quad  \mbox{ for all $s,t \in \bR$}\label{N1}\:.\eeq 
In particular  transforms evolutions
of pure states into evolutions of (other) pure states, i.e.,  $e^{-isA}\:   U_t\psi = U_t \: e^{-isA}\psi $. 

{\bf (c)}  The action on  observables (\ref{VAV}) of  the one-parameter group of symmetries generated by $A$, $\bR \ni s \mapsto e^{isA}$ leaves $H$
invariant. That is
$$e^{-isA} H e^{isA} = H\:, \quad  \mbox{ for all $s\in \bR$} \:.$$
\end{theorem}

\begin{proof} Suppose that (a) holds. By definition $U^{-1}_tAU_t =A$. By (i) in Proposition \ref{propint2} we have that  $U^{-1}_te^{-isA} U_t =e^{-isA}$ which is equivalent to (b).  
If (b) is true, we have that $e^{-isA} e^{-itH} e^{isA} = e^{-itH}$. Here an almost direct application of Stone theorem 
yields $e^{-isA} H e^{isA} = H$.  Finally suppose that (c) is valid. Again  (i) in Proposition \ref{propint2} produces
$e^{-isA} U_t e^{isA} = U_t$ which can be rearranged into
$U^{-1}_te^{-isA} U_t =e^{-isA}$. Finally Stone theorem leads to  $U^{-1}_tAU_t =A$ which is (a), concluding the proof. 
\end{proof}

\remark$\null$

{\bf (a)} In physics textbooks the above statements are almost always stated using time derivatives
and commutators. This is useless and involves many subtle troubles with domains of the involved operators.

{\bf (b)} The theorem can be extended to observables $A(t)$ parametrically depending on time already in the Schr\"odinger picture \cite{moretti}. In this case (a) and (b) are equivalent too. With this more general situation, (\ref{N1}) in (b) has to be re-written as 
$$e^{-isA(t)} U_t = U_t e^{-isA(0)}\quad  \mbox{ for all $s,t \in \bR$}$$
and Heisenberg evolution considered in (a) encompasses both time dependences
$$A_t = U_t^{-1}A(t) U_t\:.$$
At this juncture,  (c) can similarly be stated but, exactly as it happens in Hamiltonian classical mechanics, it has a more complicated interpretation \cite{moretti}.\\
An example is the generator of the boost one-parameter subgroup along the axis ${\bf n}$ of transformations of the Galileian group  $\bR^3 \ni x \mapsto x+ tv {\bf n}\in \bR^3$, where the speed $v\in \bR$ is the parameter of the group. The generator is \cite{moretti} the unique  self adjoint extension of \beq K_{\bf n}(t)= \sum_{j=1}^3 n_j  (m  X_j|_D - tP_j|_{D})\:,\label{boostF}\eeq
the constant $m>0$ denoting the mass of the system and $D$ being the G\r{a}ding or the Nelson  domain of the representation of (central extension of the) Galileian group as we will discuss later.

{\bf (c)} In QM there are symmetries described by operators which are simultaneously selfadjoint and unitary, so they are also observables and can be measured. The {\bf parity} is one of them: $({\cal P}\psi)(x) := \psi(-x)$
for a particle described in $L^2(\bR^3, d^3x)$. These are constants of motion 
($U^{-1}_t {\cal P} U_t = {\cal P}$)
if and only if they are dynamical symmetries (${\cal P} U_t = {\cal P}U_t$). This phenomenon has no classical corresponding.

{\bf (d)} The {\em time reversal symmetry}, when described by an anti unitary operator $T$ is supposed to satisfy: $THT^{-1} = H$. However, since it is antilinear gives rise to the identity (exercise)
$Te^{-itH}T^{-1} = e^{-itTHT^{-1}}$, so that $TU_t = U_{-t}T$ as physically expected.
There is no conserved quantity associated with this operator because it is not selfadjoint.

{\bf \exercise}\label{remconstt} $\null$ \\
\noindent {\bf (1)} {\em Prove that if the Hamiltonian observable does not depend on time  is a constant of motion}.

{\bf Solution}. In this case the time translation  is described by $U_t = e^{itH}$ and trivially it commute with $U_s$. Noether theorem implies the thesis $\hfill \Box$

\noindent {\bf (2)} {\em Prove that for the free particle in $\bR^3$ the momentum along $x_1$
is a constant of motion as consequence of translational invariance along that axis.
Assume that the unitary group representing translations along $x_1$ is $U_u$ with
$(U_u\psi)(x) = \psi(x - u {\bf e}_1)$ if $\psi \in L^2(\bR^3, d^3x)$.}

{\bf Solution}. The Hamiltonian is $H= \frac{1}{2m}\sum_{j=1}^3 P_j^2$. It commutes with the one-parameter unitary group describing displacements along $x_1$, because as one can prove the said groups is generated by $P_1$ itself: $U_u := e^{-iu P_1}$. Theorem \ref{neother} yields the thesis. \hfill $\Box$

\noindent {\bf (3)}  {\em Prove that if $\sigma(H)$ is bounded below but not above, the time reversal symmetry cannot be unitary.}

{\bf Solution}. We look for an operator, unitary or antiunitary such that $TU_t = U_{-t}T$ for all $t \in \bR$. If the operator is unitary, the said identity easily implies 
$THT^{-1} = -H$ and therefore, with obvious notation, $\sigma(THT^{-1}) = -\sigma(H)$.
(e) in remark \ref{remspectra} immediately yields $\sigma(H)= -\sigma(H)$ which is false  if $\sigma(H)$ is bounded below but not above. \hfill $\Box$

 \subsubsection{Strongly continuous unitary representations of Lie groups, Nelson theorem}
Topological and Lie groups are intensively used in QM \cite{BaRa}. More precisely they are studied in terms of their strongly continuous unitary representations. The reason to consider strongly continuous representations is that they immediately induce continuous representations of the group in terms of quantum symmetries (Def. \ref{defrapcont}). 
In the rest of the section  we consider only  the case of a {\em real Lie group}, $G$, whose   Lie algebra is indicated by $\gg$ endowed with the Lie bracket or commutator $\{\:\:,\:\:\}$.

\begin{definition}{\em  If $G$ is a Lie  group, a {\bf strongly continuous unitary representation} of $G$ over the complex Hilbert space $\cH$ is a group homomorphism $G \ni g \mapsto U_g \in \gB(\cH)$ such that every $U_g$ is unitary and $U_g \to U_{g_0}$, in the strong operator topology, if $g \to g_0$. }\hfill $\blacksquare$
\end{definition}

\noindent We leave to the reader the elementary proof that strong continuity is equivalent to strong continuity at the unit element of the group and in turn, this is equivalent to weak continuity at the unit element of the group.\\
A fundamental technical fact is that the said unitary representations are associated with representations of the Lie algebra of the group in terms of (anti)selfadjoint operators.
These operators are often physically interpreted as constants of motion (generally parametrically depending on time) when the Hamiltonian of the system belongs to the representation of the Lie algebra. We want to study this relation between the representation of the group on  the one hand and the representation of the Lie algebra on the other hand. First of all we define the said operators representing the Lie algebra.

\begin{definition}\label{defasagen}
{\em Let $G$ be a real Lie group and consider a
 strongly continuous unitary representation $U$ of $G$ over the complex Hilbert space $\sH$.\\
If $\sA \in \gg$ let  $\bR \ni t \mapsto \exp(t\sA) \in G$ be the generated  one-parameter Lie subgroup.
The {\bf self-adjoint generator associated with $\sA$} $$A : D(A) \to \cH$$  is the generator of the strongly continuous one-parameter unitary group $$\bR \ni t \mapsto U_{\exp\{t\sA\}} = e^{-isA}$$ in the sense of Theorem  \ref{teostone}.}\hfill $\blacksquare$
\end{definition}
\noindent The expected result is that these generators (with a factor $-i$) define a representation of the Lie algebra of the group. The utmost reason is that  they are associated to the unitary one-parameter subgroups exactly as the elements of the Lie algebra are associated to the Lie one-parameter subgroups. In particular we expect that the Lie parenthesis correspond to the commutator of operators. The technical problem is that the generators $A$ may have different domains. Thus we look for a common invariant (because the commutator must be defined thereon) domain, where all them can be defined. This domain should embody all the amount of  information about the operators $A$ themselves, disregarding the fact that they are defined in larger domains. In other words we would like that the domain be  a {\em core} ((3) in Def. \ref{defcore})
for each generator. There are several candidates for this space, one of the most appealing is the se called {\em G\r{a}rding space}.

\begin{definition}\label{defgarding}
{\em Let $G$ be a (finite-dimensional real) Lie group and consider a
 strongly continuous unitary representation $U$ of $G$ over the  complex Hilbert space $\cH$. If  $f \in C_0^\infty(G; \bC)$ and $x\in \cH$, define
\beq
x[f] := \int_G f(g) U_g x \:dg
\eeq
where $dg$ denotes the Haar measure over $G$ and the integration is defined in a weak sense exploiting Riesz' lemma: Since the map $\cH \ni x \mapsto \int_G f(g) \langle y|U_g x\rangle dg $ is continuous (the proof being elementary), 
 $x[f]$ is the unique vector in $\cH$ such that 
$$\langle  y| x[f]\rangle = \int_G f(g) \langle y|U_g x\rangle dg  \:, \quad  \forall y \in \cH\:.$$  The complex span of all vectors $x[f] \in \cH$ with $f \in C_0^\infty(G;\bC)$ and $x\in \sH$ is called {\bf G\r{a}rding space} of the representation and is denoted by $D^{(U)}_G$.}\hfill $\blacksquare$
\end{definition}

\noindent The subspace  $D^{(U)}_G$ enjoys very remarkable properties   we state in the next theorem.
In the following $L_g :  C_0^\infty(G; \bC) \to C_0^\infty(G;\bC)$ denotes the standard left-action of $g\in G$ on complex valued smooth compactly supported functions defined on $G$:
\beq
(L_gf)(h):= f(g^{-1}h) \quad \forall h \in G\:,
\eeq
and, if $\sA \in \gg$, $X_{\sA} : C_0^\infty(G;\bC) \to C_0^\infty(G;\bC)$ is the smooth vector field  over $G$ (a smooth differential operator)  defined as:
\beq
\left(X_{\sA}(f)\right)(g) := \lim_{t\to 0} \frac{f\left(\exp\{-t\sA\}g \right)-f(g)}{t}\quad \forall g \in G\:.
\eeq
so that that map
\beq
\gg \ni \sA \mapsto X_{\sA}
\eeq
defines a representation of $\gg$ in terms of vector fields (differential operators) on $C_0^\infty(G; \bC)$. We conclude with the following theorem \cite{BaRa}, establishing  that the 
G\r{a}rding space has all the expected properties.

\begin{theorem}\label{teogarding} Referring to Definitions \ref{defasagen} and \ref{defgarding}, the G\r{a}rding space  $D^{(U)}_G$ satisfies the following properties.\\

{\bf (a)} $D^{(U)}_G$ is dense in $\cH$.\\

{\bf (b)} If $g\in G$, then  $U_g(D^{(U)}_G) \subset D^{(U)}_G$. More precisely, if $f \in C_0^\infty(G)$, $x \in \cH$, $g\in G$, it holds
\beq
U_gx[f] = x[L_gf]\:.\label{Ugf}
\eeq

{\bf (c)} If $\sA \in \gg$, then $D^{(U)}_G\subset D(A)$  and furthermore $A(D^{(U)}_G) \subset D^{(U)}_G$. More precisely
\beq
-iA x[f] = x[X_{\sA}(f)] 
\eeq

{\bf (d)} The map  
\beq \gg \ni \sA \mapsto -iA|_{D^{(U)}_G} =: U(\sA)\eeq
is a Lie algebra representation in terms of anti symmetric operators defined on the common dense invariant domain $D^{(U)}_G$. In particular if $\{\:\:,\:\:\}$ is the Lie commutator of $\gg$ we have:
$$[U(\sA), U(\sA')] = U(\{\sA, \sA'\}) \quad \mbox{if $\sA, \sA'\in\gg$.}$$

{\bf (e)}  $D^{(U)}_G$ is a core for every  selfadjoint  generator $A$ with $\sA \in \gg$, that is
\beq
A = \overline{A|_{D^{(U)}_G}}\:,\quad \forall \sA \in \gg\:. 
\eeq
\end{theorem}

\noindent Now we tackle the inverse problem: We suppose to have a certain representation of a Lie algebra $\gg$ in terms of symmetric operators defined in common invariant domain of a complex Hilbert space $\cH$.
We are interested in lifting this representation to a whole strongly continuous representation of the unique simply connected Lie group $G$ admitting $\gg$ as Lie algebra. This is a much more difficult problem solved by Nelson. \\

\noindent Given a strongly continuous representation $U$ af a (real) Lie group $G$, there is another space $D^{(U)}_N$\index{$D_N$} with similar features to $D^{(U)}_G$. Introduced by Nelson  \cite{BaRa}, it turns out to be more useful than the G\r{a}rding space to {\em recover} the representation $U$ by exponentiating the Lie algebra representation. \\
By definition $D^{(U)}_N$ consists of vectors $\psi \in \cH$ such that  $G\ni g \mapsto U_g\psi$ is {\em analytic}
in  $g$, i.e. expansible  in power series in (real) analytic coordinates around every point of $G$. 
The elements of $D^{(U)}_N$ are called {\bf analytic vectors of the representation} $U$ and $D^{(U)}_N$ is the {\bf space of analytic vectors of the representation} $U$. It turns out that  $D^{(U)}_N$ is invariant for every  $U_g$, $g \in G$.\\
A remarkable  relationship exists between analytic vectors in $D^{(U)}_N$ and analytic vectors according to Definition \ref{defanalitic}. Nelson proved the following important result \cite{BaRa}, which implies that $D^{(U)}_N$ is dense in $\cH$, as we said, because analytic vectors for a self-adjoint operator are dense (exercise \ref{esanv}).
An operator is introduced, called {\em Nelson operator}, that sometimes has to do with the {\em Casimir operators} \cite{BaRa} of the represented group.

\begin{proposition}\label{propN} Let $G$  be a (finite dimensional real) Lie group and $G \ni g \mapsto U_g$ a  strongly continuous unitary representation on the Hilbert space $\cH$. Take $\sA_1,\ldots, \sA_n \in \gg$ a basis and define {\bf Nelson's operator} on $D^{(U)}_G$ by  
$$\Delta :=\sum_{k=1}^n U(\sA_k)^2\:,$$
where the $U(\sA_k)$ are, as before, the selfadjoint  generators $A_k$ restricted to the G\r{a}rding domain $D^{(U)}_G$. Then\\
{\bf (a)} $\Delta$ is essentially selfadjoint on $D^{(U)}_G$.\\
{\bf (b)} Every  analytic vector of the selfadjoint operator $\overline{\Delta}$ is analytic an element of $D^{(U)}_N$, in particular $D^{(U)}_N$ is dense.\\
{\bf (c)} Every vector in $D^{(U)}_N$ is analytic for every self-adjoint operator $\overline{U(\sA_k)}$, which is thus essentially selfadjoint 
in  $D^{(U)}_N$ by Nelson's criterion.
\end{proposition}
\noindent We  finally state the well-known theorem of Nelson that enables to associate representations of the only simply connected Lie group with a given Lie algebra to  representations of that Lie algebra. 

\begin{theorem}[Nelson theorem]\label{teoN} Consider a real  $n$-dimensional Lie algebra $V$ of operators $-iS$ -- with each  $S$ symmetric on the Hilbert space $\cH$, defined on a common invariant subspace $\cD$ dense in $\cH$ and $V$-invariant -- with the usual commutator of operators as Lie bracket.\\ 
Let $-iS_1,\cdots, -iS_n \in V$ be a basis of $V$ and define Nelson's operator with domain $\cD$:
$$\Delta := \sum_{k=1}^n S_k^2\:.$$ 
If $\Delta$ is essentially self-adjoint, there exists a strongly continuous unitary representation 
$$G_V \ni g \mapsto U_g$$ 
on $\cH$, of the unique simply connected Lie group $G_V$ with  Lie algebra  $V$. \\ $U$  is completely determined by the fact that the closures
$\overline{S}$, for every $-iS\in V$, are the selfadjoint generators of the representation
of the one-parameter subgroups of $G_V$ in the sense of Def. \ref{defasagen}.\\
In particular, the symmetric operators $S$ are essentially selfadjoint on $\cD$, their closure being selfadjoint.
\end{theorem}

{\bf \exercise\label{ABcomm}} {\em Let  $A,B$ be  selfadjoint operators in the complex Hilbert space $\cH$ with a common invariant dense domain $D$ where they are symmetric and  commute. Prove that if  $A^2+B^2$ is essentially self adjoint on $D$, then the spectral measures of $A$ and $B$  commute.}

{\bf Solution}. Apply Nelson's theorem observing that $A,B$ define the Lie algebra of the additive Abelian Lie group $\bR^2$ and  that $D$ is a core for  $A$ and $B$,
because they are essentially selfadjoint therein again by Nelson theorem. $\hfill \blacksquare$

\example\label{exL} $\null$\\
{\bf (1)} Exploiting spherical polar coordinates, the Hilbert space $L^2(\bR^3, d^3x)$ can be factorised as $L^2([0,+\infty), r^2dr) \otimes L^2(\bS^2, d\Omega)$, where $d\Omega$ is the natural rotationally invariant Borel measure on the sphere  $\bS^2$ with unit radius in $\bR^3$,
with $\int_{\bS^2} 1 d\Omega =4\pi$. In particular a Hilbertian basis of $L^2(\bR^3, d^3x)$ is therefore made of
the products $\psi_n(r) Y^l_{m}(\theta,\phi)$ where $\{\psi_n\}_{n\in \bN}$ is any Hilbertian basis in $L^2([0,+\infty), r^2dr) $ and $\{Y^l_m\:|\: l=0,1,2,\ldots \:, m=0,\pm 1,\pm 2, \ldots \pm l \}$ is  the standard Hilbertian basis of {\em spherical harmonics} of  $L^2(\bS^2, d\Omega)$
\cite{BaRa}. Since the function $Y^l_m$ are smooth on $\bS^2$,  it is possible to arrange the basis of $\psi_n$ made of  compactly supported smooth functions whose derivatives in $0$ vanish 
at every order, in order that 
$\bR^3 \ni x \mapsto (\psi_n \cdot Y^l_{m})(x)$ are elements of  $\bC^\infty(\bR^n; \bC)$ (and therefore also 
of ${\cal S}(\bR^3)$). Now consider the three symmetric operators defined on the common dense invariant domain   ${\cal S}(\bR^3)$
$${\cal L}_k = \sum_{i,j=1}^3\epsilon_{kij} X_iP_j|_{{\cal S}(\bR^3)}$$
where $\epsilon_{ijk}$ is completely antisymmetric in $ijk$ and  $\epsilon_{123}=1$.
By direct inspection one sees that 
$$[-i {\cal L}_k, -i{\cal L}_h] =  \sum_{r=1}^3 \epsilon_{khr} (-i {\cal L}_r) $$
so that the finite real span of the operators $i{\cal  L}_k$ is  a representation of the Lie algebra of the simply connected real Lie group  $SU(2)$ (the universal covering of $SO(3)$). 
Define the Nelson operator ${\cal L}^2 := -\sum_{k=1}^3 {\cal L}_k^2$ on ${\cal S}(\bR^3)$. Obviously this is a symmetric operator. A well known computation proves that
$${\cal L}^2 \: \psi_n(r) Y^l_{m} = l(l+1) \: \psi_n(r) Y^l_{m}\:.$$
We conclude that ${\cal L}^2$ admits a Hilbertian basis of eigenvectors. Corollary \ref{CN} implies that  ${\cal L}^2$
is essentially self adjoint.  Therefore we can apply Theorem \ref{teoN} concluding that there exists a strongly continuous  unitary representation $SU(2) \ni M \mapsto U_M$ of  $SU(2)$ (actually it can be proved to be also of $SO(3)$).  The three  selfadjoint operators  $L_k := \overline{{\cal L}_k}$ are the generators of the one-parameter of rotations around the corresponding three orthogonal Cartesian axes $x_k$, $k=1,2,3$.
The one-parameter subgroup of rotations around the generic unit vector  ${\bf n}$, with components $n_k$, admits the  selfadjoint 
generator $L_{\bf n} = \overline{\sum_{k=1}^3 n_k{\cal L}_k }$.
The  observable  $L_{\bf n}$ has the physical meaning of the  ${\bf n}$-{\em component  of the angular momentum}  of the particle described in $L^2(\bR^3,d^3x)$. It turns out that, for $\psi \in L^2(\bR^3, d^3x)$,
\beq (U_{M}\psi)(x) = \psi(\pi(M)^{-1} x)\:,\quad M \in SU(2) \quad \:, x \in \bR^3 \label{rot}\eeq
where $\pi : SU(2) \to SO(3)$ is the standard covering map. (\ref{rot}) is the action of the rotation group on pure states in terms of quantum symmetries.  This 
representation is, in fact, a subrepresentation of the unitary representation of $IO(3)$ already found in (1) of example \ref{exrep}.\\
{\bf (2)} Given a quantum system, a quite general situation is the one where the quantum symmetries of the systems are described by a strongly continuous representation $V : G \ni g \mapsto V_g$ on the Hilbert space $\cH$ of the system, and  the time evolution is the representation of a one-parameter Lie subgroup with generator $\sH \in \gg$. So that 
$$V_{\exp(t\sH)}= e^{-itH} =: U_t\:.$$ This is the case, for instance, of relativistic quantum particles, where $G$ is the {\em special orthochronous Lorentz group}, $SO(1,3)_+$, (or its universal covering $SL(2,\bC)$). Describing non-relativistic quantum particles, the relevant group $G$ is an $U(1)$ central extension of the universal covering of the (connected orthochronous) Galileian group. \\ In this situation, every  element of $\gg$ determines a constant of motion. Actually there are two cases.

(i) If $\sA \in \gg$ and $\{\sH, \sA\}=0$, then the Lie subgroups $\exp(t\sH)$ and $\exp(s\sA)$ commute as, for example, follows from {\em Baker-Campbell-Hausdorff}  formula (see \cite{BaRa,moretti},  for instance). Consequently $A$ is a constant of motion because $V_{\exp(t\sH)} =e^{-itH}$ and $V_{\exp(s\sA)} = e^{-isA}$ commute as well and Theorem \ref{neother} is valid. In this case $e^{-isA}$ defines a dynamical symmetry in accordance with the afore-mentioned theorem. This picture applies in particular, referring to a free particle, to $A= J_{\bf n}$, the observable describing total angular momentum along the unit vector ${\bf n}$ computed in an inertial  reference frame.

(ii) A bit more complicated is the case of $\sA \in \gg$ with
$\{\sH, \sA\} \neq 0$. However, even in this case $\sA$ defines a constant of motion 
in terms of selfadjont  operators (observables) belonging to the representation of the Lie algebra of $G$. The difference with respect to the previous case is that, now, the constant of motion {\em parametrically depend on time}. We therefore have a class of observables $\{A(t)\}_{t\in \bR}$ in the Schr\"odinger picture, in accordance with (b) in remark \ref{remconstt}, such  that $A_t := U^{-1}_t A(t) U_t$ are the corresponding observables in the Heisenber picture. The equation stating that we have a constant of motion is therefore $A_t= A_0$. \\ Exploiting the natural action of the Lie one-parameters subgroups on $\gg$,  let us define the time parametrised class of elements of the Lie algebra  $$\sA(t) := \exp(t\sH)\sA\exp(-t\sH) \in \gg\:,\quad t \in \bR\:.$$ If $\{\sA_k\}_{k=1,\ldots,n}$ is a basis of $\gg$, it must consequently hold
\beq \sA(t) = \sum_{k=1}^n a_k(t) \sA_k\label{decgg}\eeq for some real-valued smooth functions $a_k=a_k(t)$.
By construction, the corresponding class of selfadjoint generators $A(t)$, $t\in \bR$, define a parametrically time dependent constant of motion. Indeed, since (exercise)
$$\exp(s\exp(t\sH)\sA\exp(-t\sH)) =  \exp(t\sH) \exp(s\sA)\exp(-t\sH)\:,$$ we have 
$$-iA(t)  =  \frac{d}{ds}|_{s=0}V_{\exp(t\sH)\sA\exp(-t\sH)}= \frac{d}{ds}|_{s=0}
V_{\exp(t\sH) \exp(s\sA)\exp(-t\sH)}$$ $$= \frac{d}{ds}|_{s=0} V_{\exp(t\sH)} V_{\exp(s\sA)}V_{\exp(-t\sH)} = -i U_t A U^{-1}_t$$
Therefore
$$A_t = U^{-1}_t A(t) U_t = U^{-1}_t U_t A U^{-1}_tU_t = A = A_0\:.$$
In view of Theorem \ref{teogarding}, as the map $\gg \ni A \mapsto A|_{D_G^{(V)}}$ is a Lie algebra isomorphism, we can recast (\ref{decgg}) for selfadjoint generators
\beq
A(t)|_{D_G^{(V)}} = \sum_{k=1}^n a_k(t) A_k|_{D_G^{(V)}}
\eeq
(where $D_G^{(V)}$ may be replaced by $D_N^{(V)}$ as the reader can easily establish, taking advantage of Proposition \ref{propN} and  Theorem \ref{teoN}). Since $D_G^{(V)}$ (resp. $D_N^{(V)}$) is a core 
for $A(t)$, it also hold
\beq
A(t) = \overline{\sum_{k=1}^n a_k(t) A_k|_{D_G^{(V)}}}\:,
\eeq
the bar denoting the closure of an operator as usual.
(The same is true replacing $D_G^{(V)}$ for $D_N^{(V)}$.)
An important case, both for the non-relativistic and the relativistic case is the selfadjoint generator  $K_{\bf n}(t)$
associated with the boost transformation  along the unit vector ${\bf n}\in \bR^3$, the rest space of the inertial reference frame where the boost transformation is viewed as an active transformation. In fact, referring to the Lie generators of (a $U(1)$ central extension of the universal covering of the connected orthochronous)
Galileian group, we have $\{h, k_{\bf n}\} = - p_{\bf n} \neq 0$, where $p_{\bf n}$ is the generator of spatial translations along ${\bf n}$, corresponding to the observable momentum along the same axis when passing to selfadjoint generators.
 The non-relativistic expression  of $K_{\bf n}(t)$, for a single particle,  appears in (\ref{boostF}). For a more extended discussion on the non-relativistic case see \cite{moretti}. A pretty complete discussion including the relativistic case is contained in \cite{BaRa}.
$\hfill \blacksquare$

\subsubsection{Selfadjoint version of Stone - von Neumann - Mackey Theorem} \label{SvNM} A remarkable consequence of Nelson's theorem is a selfadjoint operator version of Stone-von Neumann  theorem usually formulated in terms of unitary operators  \cite{ercolessi,moretti}, proving that the CCRs always give rise to the standard representation in $L^2(\bR^n, d^nx)$. We state and prove this version of the theorem,  adding a  last statement which is the selfadjoint version of  Mackey completion  to Stone von Neumann statement \cite{moretti}.

\begin{theorem}[Stone - von Neumann - Mackey Theorem]\label{SNM} Let $\cH$ be a complex Hilbert space and suppose that there are $2n$ selfadjoint operators in $\cH$ we indicate with $Q_1,\ldots, Q_n$ and $M_1,\ldots, M_n$ such the following requirements are valid.\\

{\bf (1)} There is a common dense invariant subspace $D\subset \cH$ where  the CCRs hold
\begin{equation}
[Q_h, M_k]\psi = i\hbar \delta_{hk} \psi \:,\quad [Q_h, Q_k]\psi =0\:, \quad  [M_h, M_k]\psi = 0\quad \psi \in D\:, \quad h,k =1,\ldots, n\:. \label{CCRd}
\end{equation}

{\bf (2)} The representation is irreducible, in the sense that there is no closed subspace $\cK \subset \cH$ such that $Q_k(\cK\cap D(Q_k))\subset \cK$ and  $M_k(\cK\cap D(M_k))\subset \cK$ for $k=1,\ldots, n$.\\

{\bf (3)}  The operator $\sum_{k=1}^n Q_k^2|_D + M_k^2|_D$ is essentially self adjoint.\\
Under these conditions, there is a Hilbert space isomorphism, that is a   surjective isometric map, $U: \cH \to L^2(\bR^n, d^nx)$ such that
\beq UQ_kU^{-1}=X_k\quad \mbox{and}\quad UM_kU^{-1}=P_k\quad k=1,\ldots, n\label{Ur}\eeq 
where  $X_k$ and $P_k$ respectively are the standard position (\ref{Xm}) and momentum (\ref{Pm}) selfadjoint operators in $L^2(\bR^n,d^nx)$. In particular $\cH$ results to be separable.

If (1), (2) and (3) are valid with the exception that the representation is not reducible, then $\cH$ decomposes into  an orthogonal Hilbertian sum $\cH = \oplus_{r\in R} \cH_k$ where $R$
is finite or countable if $\cH$ is separable,  the $\cH_r\subset \cH$ are closed subspaces with $$Q_k(\cH_r\cap D(Q_k)) \subset \cH_r\quad \mbox{and} \quad M_k(\cH_r\cap D(M_k)) \subset \cH_r$$ for all $r\in R$, $k=1, \ldots, n$ and the restrictions of all the  $Q_k$ and $M_k$    to each $\cH_r$ satisfy (\ref{Ur}) for suitable  surjective isometric maps $U_r: \cH_r \to L^2(\bR^n, d^nx)$
\end{theorem}

\begin{proof} If (1)  holds, the restrictions to $D$ of  the selfadjoint operators $Q_k$, $M_k$  define symmetric operators (since they are selfadjoint and $D$ is dense and  included in their domains), also their powers are symmetric since $D$ is invariant.
If also (2) is valid, in view of Nelson theorem (since evidently  the symmetric operator $I|_D^2 + \sum_{k=1}^n Q_k^2|_D + M_k^2|_D$ is essentially selfadjoint if $\sum_{k=1}^n Q_k^2|_D + M_k^2|_D$ is),
there is a strongly continuous unitary representation $W \ni g \mapsto V_g \in \gB(\cH)$ of the simply connected $2n+1$-dimensional  Lie group $W$ whose Lie algebra is defined by (\ref{CCRd})  (correspondingly re-stated  for the operators $-iI, -iQ_k, -iM_k$)
together with $[-iQ_h,-i I]=[-iM_k,-iI]=0$, where the operator $-iI$
restricted to $D$  is the remaining Lie generator. $W$ is the  {\em Weyl-Heisenberg} group \cite{moretti}. The selfadjoint generators of this representation are just the operators $Q_k$ and $P_k$ (and $I$),  since they coincide with the closure of   their restrictions  to $D$,  because they are selfadjoint (so they admit unique selfadjoint extensions) and $D$ is a core. If furthermore the Lie algebra representation is irreducible,
the unitary representation is irreducible, too: If $\cK$ were an invariant subspace for the unitary operators, Stone theorem would imply that $\cK$  be also invariant under the selfadjoint generators of the one parameter Lie subgroups associated to each $Q_k$ and $P_k$. This is impossible if the Lie algebra representation is irreducible as we are assuming. The standard version of Stone-von Neumann theorem \cite{moretti} implies that there is isometric surgective  operator $U : \cH \to L^2(\bR^n, d^nx)$ such that $W \ni g \mapsto UV_gU^{-1}\in \gB(L^2(\bR^n,d^nx))$  is the standard unitary representation of the group $W$ in $ L^2(\bR^n, d^nx)$ genernated by $X_k$ and $P_k$ (and $I$) \cite{moretti}. Again, Stone theorem immediately yields (\ref{Ur}). The last statement easily follows from the standard form of Mackey's theorem completing Stone-von Neumann result \cite{moretti}.
\end{proof}

\remark $\null$

{\bf (a)} The result {\em a posteriori} gives, in particular,   a strong justification of the requirement that the Hilbert space of an elementary quantum system, like a particle, must be separable. 

{\bf (b)} Physical Hamiltonian operators have spectrum bounded from below to avoid thermodynamical instability. This fact prevents the definition of a ``time operator'' canonically conjugated with $H$ following the standard way. This result is sometime quoted as {\bf Pauli theorem}.
As a consequence, the meaning of 
Heisenberg relations $\Delta E \Delta T \geq \hbar/2$ is different from the meaning of the analogous relations for position and momentum. It is however possible to define a sort of time osservable just 
extending the notion of PVM to the notion of POVM (positive valued operator measure) \cite{ercolessi,moretti}. POVMs are exploited to describe concrete physical phenomena related to measurement procedures, especially in quantum information theory \cite{busch,BGL}.
$\hfill \blacksquare$

\begin{corollary} If the Hamiltonian operator $\sigma(H)$ of a quantum system is bounded below, there is no selfadjoint operator (time operator) $T$ satisfying the standard CCR with $H$ and the hypotheses (1), (2), (3) of Theorem \ref{SNM}.
\end{corollary}

\begin{proof} The couple $H,T$ should be mapped to a corresponding couple $X,P$ in $L^2(\bR, dx)$, or a direct sum of such spaces, by means of a Hilbert space isomorphism.
In all cases the spectrum of $H$ should therefore be identical to the one of $X$, namely is $\bR$. This fact is false by hypotheses.
\end{proof}
\section{Just few words about the  Algebraic Approach}
The fundamental theorem \ref{SvNM} of Stone-von Neumann  and Mackey is stated in the jargon of theoretical physics as follows:\\
 ``{\em all irreducible representations of the CCRs  with a finite, and fixed, number of  degrees of freedom are unitarily equivalent},''.\\
The expression   {\em unitarily equivalent} refers to the existence of the Hilbert-space isomorphism  $U$, and the finite number of degrees of freedom is the dimension of the Lie algebra spanned by the generators $I, X_k,P_k$.\\
What happens then in infinite dimensions? This is the case when dealing with {\em quantum fields},  where the $2n+1$ generators $I, X_k, P_k$ ($k=1,2,\ldots, n$),  are replaced by a {\em continuum} of generators, the 
so-called  {\em field operators  at fixed time} and the {\bf conjugated momentum at fixed time}: $I, \Phi(f), \Pi(g)$ which are smeared by arbitrary functions  $f,g \in C^\infty_0(\bR^3)$. Here $\bR^3$ is the rest space of a given reference frame in the spacetime. Those field operators satisfy commutation relations similar to the ones of  $X_k$ and $P_k$ (e.g., see \cite{H,araki,IV}).
 Then  the Stone--von Neumann theorem no longer holds.  In this case, theoretical physicists would say that\\
 ``{\em there exist irreducible non-equivalent CCR representations with an infinite number of  degrees of freedom}''.\\
What happens in this situation, in practice, is that one finds two {\em isomorphic} $^*$-algebras of field  operators,
the one generated by  $\Phi(f), \Pi(g)$ in the Hilbert space $\cH$ and the other 
generated by  $\Phi'(f), \Pi'(g)$ in the Hilbert space $\cH'$
that admit  {\em no}  {\em  Hilbert space} isomorphism  $U: \sH'\to \sH$ satisfying:
 $$U\Phi'(f) \:U^{-1}  = \Phi(f)\:, \quad  U\Pi'(g) \:U^{-1}  = \Pi(g)
 \:\:\:\: \mbox{for any pair $f,g \in  C^\infty_0(\bR^3)$.}$$
Pairs of this kind are called {\em (unitarily) non-equivalent}. Jumping from the  finite-dimensional case to the infinite-dimensional one corresponds to passing from Quantum Mechanics to Quantum Field Theory (possibly relativistic, and on curved spacetime \cite{IV}).   The presence of non-equivalent representations of one single physical system  shows that a formulation in a fixed Hilbert space is fully inadequate, a least  because it insists on a fixed Hilbert space, whereas the physical system is characterized by a more abstract object: An algebra of observables which may be represented in different Hilbert spaces in terms of operators. These representations are not unitarily equivalent and none  can be considered  more fundamental than the remaining ones.
We must abandon  the structure of  Hilbert space in order to 
lay the foundations of quantum theories in broader generality.\\
This programme has been widely developed (see e.g., \cite{BrRo,strocchi,H,araki}), starting from the pioneering work of von Neumann himself, and is nowadays called  {\em algebraic formulation of quantum (field) theories}. Within this framework it was possible to formalise, for example, field theories in curves spacetime in relationship to the quantum phenomenology of black-hole thermodynamics.

\subsection{Algebraic formulation}  
The algebraic formulation prescinds, anyway, from the nature of the quantum system and may be stated for systems with finitely many  degrees of freedom  as well \cite{strocchi}. The new viewpoint relies  upon two assumptions \cite{H,araki,strocchi,strocchi2011,moretti}.\\

\noindent {\bf AA1}.  A physical system $S$ is described by its {\bf observables}, viewed now as selfadjoint elements in a certain $C^*$-algebra $\gA$ with unit $1\!\!1$ associated 
to $S$. \\

\noindent {\bf AA2}. An {\bf algebraic state} on $\gA_S$ is a 
linear functional $\omega : \gA_S \to \bC$ such that:
$$\omega(a^*a) \geq 0\quad \forall\,a \in \gA_S, \qquad \omega(1\!\!1)=1\;,$$
that is, {\em positive} and {\em normalised to $1$}.\\

\noindent We have to stress  that $\gA$ is not seen as a concrete $C^*$-algebra of  operators (a von Neumann algebra for instance) on a given  Hilbert space, but remains an abstract $C^*$-algebra. Physically, $\omega(a)$ is the {\em expectation value} of the observable  $a\in \gA$ in state $\omega$.

\begin{remark}\label{remC} $\null$

{\bf (a)} $\gA$ is usually called {\em the algebra of observables of $S$} though, properly speaking, the observables are the selfadjoint elements of $\gA$ only.

{\bf (b)} Differently form the Hilbert space formulation, the algebraic approach can be adopted to describe {\em both classical and quantum systems}. The two cases are distinguished on the base of commutativity of the algebra of observables $\gA_S$: A commutative algebra is assumed to describe a classical system whereas a non-commutative one is supposed to be associated with a quantum systems.

{\bf (c)} The notion of {\bf spectrum} of an element $a$ of a $C^*$-algebra $\gA$, with unit element $1\!\!1$, is defined analogously  to the operatorial case \cite{moretti}. $\sigma(a) := \bC \setminus \rho(a)$ where we have introduced  the {\bf resolvent} set:
$$\rho(a) := \{\lambda \in \bC \:|\: \exists (a-\lambda 1\!\!1)^{-1} \in \gA\}\:.$$
When applied to the elements of $\gB(\cH)$, this definition coincides with the one proviously discussed for operators in view of (2) in exercise \ref{enr}. It turns out that if $a^*a=aa^*$, namely $a \in \gA$ is normal, then
$$||a|| = \sup_{\lambda \in \sigma(a)} |\lambda|\:.$$
The right hand side of the above identity is called {\bf spectral radius} of $a$.
If $a$ is not normal, $a^*a$ is selfadjoint and thus normal in any cases. Therefore  the $C^*$-property of the norm $||a||^2 = ||a^*a||$
permits us to  write down $||a||$ in terms of the spectrum of $a^*a$. As the spectrum is a completely algebraic property, we conclude that it is impossible to change the norm of a $C^*$-algebra preserving the $C^*$-algebra property of the new norm. A unital $^*$-algebra admits at most one $C^*$-norm.

{\bf (d)} Unital $C^*$-algebras are very rigid structures. In particular, every $*$-homomorphism $\pi : \gA \to \gB$ (which is a pure algebraic notion)
between two unital $C^*$-algebras is necessarily \cite{moretti} norm decresing ($||\pi(a)||\leq ||a||$) thus continuous.
Its image, $\pi(\gA)$, is a $C^*$-subalgebra of $\gB$.
 Finally $\pi$ is injective if and ony if it is isometric. The spectra satisfy a certain permanence property \cite{moretti}, with obvious meaning of the symbols
$$\sigma_{\gB}(\pi(a)) = \sigma_{\pi(\gA)}(a) \subset \sigma_\gA(a)\:, \quad \forall a \in \gA\:,$$
where the last inclusion becomes and equality if $\pi$ is injective.
\hfill $\blacksquare$
\end{remark}

\noindent The most evident {\em a posteriori } justification of the algebraic approach lies in its powerfulness \cite{H}. However
 there have been a host of attempts to account for  assumptions  {\bf AA1} and {\bf AA2} and their physical meaning  in full generality (see the study of \cite{Emch}, \cite{araki} and \cite{strocchi,strocchi2011} and especially the work of I. E. Segal
\cite{SegalCS} based on so-called {\em Jordan algebras}). Yet none seems to be definitive \cite{Lost}. \\
An evident difference with respect to the standard QM, where states are measures on the lattice of elementary propositions, is that  we have now a complete   identification of the notion of state with that of expectation value.
This  identification would be  natural within the Hilbert space formulation, where the class of observables includes the elementary ones, represented by orthogonal projectors, and corresponding to ``Yes-No'' statements.  The expectation value of such an observable coincides with the probability that the outcome of the measurement is ``Yes''. The set of all those probabilities  defines, in fact,  a quantum state of the system as we know.  However, the analogues of these elementary propositions generally do not belong to the $C^*$-algebra of observables  in the algebraic formulation.
Nevertheless,  this is not an  insurmountable obstruction.  Referring to a completely general physical system and following  \cite{araki},  the most general notion of state,  $\omega$,  is  the assignment of all probabilities, $w_\omega^{(A)}(a)$, that the outcome of the measurement of the observable $A$ is $a$, for all observables $A$ and all of values $a$. 
On the other hand,  it is known   \cite{strocchi} that all experimental information on the measurement of an observable $A$ in the state $\omega$ -- the probabilities $w_\omega^{(A)}(a)$  in particular  -- is recorded in the expectation values of the polynomials of $A$.
 Here, we should think of $p(A)$ as the observable whose values are the values $p(a)$ for all values $a$ of $A$. This characterization of an observable
is theoretically supported by the various solutions to the   {\em moment problem} in probability measure theory.
To adopt this paradigm we have thus to assume that the set of observables must include at least all real polynomials $p(A)$ whenever it  contains the observable $A$. This is in agreement with the much stronger requirement  {\bf AA1}.

\subsubsection{The GNS reconstruction theorem} The set of algebraic states on $\gA_S$ is a convex subset in the dual  $\gA_S'$ of $\gA_S$: if $\omega_1$ and $\omega_2$ are positive and normalised linear functionals,  $\omega = \lambda \omega_1 + (1-\lambda)\omega_2$ is clearly still the same for any $\lambda \in [0,1]$.\\
Hence, just as we saw for the standard formulation, we can define  {\em pure algebraic states} as  extreme elements of the convex body.

\begin{definition} {\em An algebraic state $\omega : \gA \to \bC$ on the  $C^*$-algebra with unit $\gA$ is called a {\bf pure algebraic state} if it is extreme in the set of algebraic states. 
An algebraic state that is not pure is called {\bf mixed}.}\hfill $\blacksquare$
\end{definition}

\noindent Surprisingly, most of the entire abstract apparatus introduced, given by a $C^*$-algebra and a set of 
states, admits elementary Hilbert space representations when a reference algebraic state is fixed. This is by virtue of a famous procedure that Gelfand,  
Najmark and Segal came up with, and that we prepare to present \cite{H,araki,strocchi,moretti}.
\begin{theorem}[GNS reconstruction theorem] \label{TGNS} 
Let $\gA$ be a $C^*$-algebra with unit $1\!\!1$ and $\omega: \gA \to \bC$ a positive linear functional with  $\omega(1\!\!1)=1$. Then the following holds.\\
{\bf (a)} There exist a triple  $(\cH_\omega, \pi_\omega, \Psi_\omega)$, where $\cH_\omega$ is a Hilbert space, the map $\pi_\omega : \gA \to \gB(\cH_\omega)$ a $\gA$-representation over $\cH_\omega$ and 
$\Psi_\omega \in \cH_\omega$, such that:

(i) $\Psi_\omega$ is cyclic for $\pi_\omega$. In other words,  $\pi_\omega(\gA) \Psi_\omega$ is  dense in $\cH_\omega$,

(ii) $\langle\Psi_\omega|\pi_\omega(a) \Psi_\omega\rangle = \omega(a)$ for every $a\in \gA$.\\
{\bf (b)} If $(\cH, \pi, \Psi)$  satisfies (i) and (ii), there exists a unitary operator 
$U: \cH_\omega \to \cH$ such that  $\Psi = U\Psi_\omega$ and $\pi(a) = U\pi_\omega(a)U^{-1}$ for any 
$a \in \gA$.
\end{theorem} 

\remark The GNS representation $\pi_\omega : \gA \to \gB(\cH_\omega)$ is a $^*$-homomorphism and thus (c) in remark \ref{remC} applies. In particular $\pi_\omega$ is norm decreasing and continuous. Moreover, again referring to the same remark, if $\pi_\omega$ is faithful -- i.e., injective -- it is isometric and preserves the spactra of the elements. If $a \in \gA$ is selfadjoint $\pi_\omega(a)$ is a selfadjoint {\em operator} and its spectrum has the well-known quantum meaning. 
This meaning, in view of the property of permanence of the spectrum, can be directly attributed to the spectrum of $a \in \gA$:  {\em If $a\in \gA$ represents an abstract observable,  $\sigma(a)$ is the set of the possible values attained by $a$}. \hfill $\blacksquare$\\

\noindent As we initially said,  it turns out that different algebraic states $\omega$, $\omega'$ give generally rise to unitarily inequivalent GNS representations $(\cH_\omega, \pi_\omega, \Psi_\omega)$ and $(\cH_{\omega'}, \pi_{\omega'}, \Psi_{\omega'})$: There is no isometric surjective operator $U : \cH_{\omega'} \to \cH_{\omega}$
such that $$U \pi_{\omega'}(a) U^{-1}=\pi_{\omega}(a) \quad \forall a \in \sA\:.$$
The fact that one may  simultaneously deal with all these inequivalent representations is a representation of the power of the algebraic approach with  respect to the Hilbert space framework.\\
However one may also focus on states referred to a fixed GNS representation.
 If $\omega$ is an algebraic state on $\gA$, every statistical operator on the Hilbert space of a GNS representation of $\omega$ -- i.e. every positive, trace-class  operator with unit trace 
$T\in \gB_1(\cH_\omega)$ -- determines an algebraic state
$$\gA \ni a \mapsto tr\left(T \pi_\omega(a) \right)\:,$$
evidently.
This is true, in particular, for $\Phi \in \cH_\omega$ with  $||\Phi||_\omega =1$, in which case the above definition reduces to 
$$\gA \ni a \mapsto \langle\Phi| \pi_\omega(a) \Phi\rangle_\omega\:.$$

\begin{definition} {\em If $\omega$ is an algebraic state on the $C^*$-algebra with unit $\gA$, every algebraic state on  $\gA$ obtained either from a density operator or a unit vector, in  a GNS representation of $\omega$, is called {\bf normal state} \index{normal states of an algebraic state} of $\omega$. Their set  $Fol(\omega)$ is the {\bf folium} \index{$Fol(\omega)$}\index{folium of an algebraic state} of the algebraic state $\omega$}.\hfill $\blacksquare$
\end{definition}
\noindent Note that in order to determine $Fol(\omega)$ one can use a fixed GNS representation  of $\omega$. In fact, as the GNS representation of $\omega$ varies,  normal states do not change, as implied by part (b) of the  GNS theorem.\\

\subsubsection{Pure states and irreducible representations}   To conclude we would like to explain how pure states are characterised in the algebraic framework.
To this end we have the following simple result (e.g., see  \cite{H,araki, strocchi, moretti}.

\begin{theorem}[Characterisation of pure algebraic states]\label{TEOSTATPURI}
Let $\omega$ be an algebraic state on the $C^*$-algebra with unit $\gA$ and
 $(\cH_\omega, \pi_\omega, \Psi_\omega)$ a corresponding GNS triple. Then  $\omega$ is pure if and only if 
 $\pi_\omega$ is irreducible.
\end{theorem}

\noindent The algebraic notion of pure state is in nice agreement with the Hilbert space formulation result where pure states are represented by  unit vectors (in the absence of superselection rules). Indeed we have the following proposition which make a comparison between the two notions.

\begin{proposition}
Let $\omega$ be a pure state on the $C^*$-algebra with unit $\gA$ and $\Phi \in \sH_\omega$ a unit vector. Then \\
{\bf (a)}  the functional
$$\gA \ni a \mapsto \langle\Phi| \pi_\omega(a) \Phi\rangle_\omega\:,$$
defines a pure algebraic state  and $(\sH_\omega, \pi_\omega, \Phi)$ is a  GNS triple for it.
In that case, GNS representations of  algebraic states given by non-zero vectors in $\cH_\omega$  are all unitarily equivalent.\\
{\bf (b)} Unit vectors $\Phi, \Phi' \in \sH_\omega$ give the same (pure) algebraic state if and only if $\Phi = c\Phi'$ for some $c\in \bC$, $|c|=1$, i.e. if and only if $\Phi$ and $\Phi'$ belong to the same ray. 
\end{proposition}

\noindent The correspondence pure (algebraic) states vs. state vectors, automatic in the  standard formulation, holds in  Hilbert spaces of GNS representations of 
 pure algebraic states, but in general not for mixed algebraic states. The following exercise focusses on this apparent problem.

{\bf \exercise} {\em Consider, in the standard (not algebraic) formulation, a physical system described on 
the  Hilbert space $\cH$ and a mixed state $\rho \in \gS(\cH)$. The map  $\omega_\rho : \gB(\cH)\ni A \mapsto tr(\rho A)$ defines an algebraic state on the $C^*$-algebra $\gB(\cH)$. By the GNS theorem, there exist another Hilbert space $\cH_\rho$, a representation $\pi_\rho : \gB(\cH) \to \gB(\cH_\rho)$ an unit vector $\Psi_\rho \in \sH_\rho$ such that
$$tr(\rho A) = \langle \Psi_\rho| \pi_\rho(A) \Psi_\rho\rangle\:$$
for  $A \in \gB(\cH)$. Thus it seems that the initial mixed state has been transformed into a  pure state! How is this fact explained?}

{\bf Solution}. There is no transformtion from mixed to pure state because 
the mixed state is represented by a vector, $\Psi_\rho$,  in a different Hilbert space, $\cH_\rho$. Moreover, there is no Hilbert space isomorphism  $U : \cH \to \cH_\rho$ with  
$UAU^{-1} = \pi_\rho(A)$, so that $U^{-1}\Psi_\rho \in \cH$.
In fact, the representation $\gB(\cH) \ni A \mapsto A \in \gB(\cH)$ is irreducible,  whereas  $\pi_\rho$ cannot be irreducible (as it would be if $U$ existed), because the state $\rho$ is not an extreme point in the space of non-algebraic states, and so it cannot be extreme in the larger space of algebraic states. $\hfill \blacksquare$

\section*{Acknowledgments}
The author is grateful to S. Mazzucchi who carefully read this manuscript, and to 
G. Marmo, S. Mazzucchi, R. Picken and M. S\'anchez for useful discussions about the content of these lectures.

\end{document}